\documentclass[a4paper, english]{article}
\usepackage[top=60pt,bottom=70pt,left=78pt,right=76pt]{geometry}

\title{The Yoneda Reduction of Polymorphic Types \\ (extended version)}

\author{Paolo Pistone\\ {\small DISI, Universit\`a di Bologna,} \\
{\small Mura Anteo Zamboni,7, 40126, Bologna (Italy)} \\
{\small\url{paolo.pistone2@unibo.it}} \\ \ \\
Luca Tranchini \\ {\small Wilhelm-Schickard-Institut, Universit\"at T\"ubingen, }\\
{\small 
Sand 13, D-72076, T\" ubingen (Germany)}\\
{\small\url{
luca.tranchini@gmail.com}}}
\date{}

\usepackage{etex}
\usepackage{rotating}
\usepackage{graphicx}
\usepackage{tabularx}
\usepackage{amsfonts}
\usepackage{amsthm}
\usepackage{amsmath}
\usepackage{mathabx2}
\usepackage{proof}
\usepackage{mathdots}
\usepackage{hyperref}
\usepackage{xcolor}
\usepackage{adjustbox}
\usepackage{cmll}
\usepackage{mathtools}

\usepackage{thmtools}
\usepackage{thm-restate}

\usepackage{accents}

\usepackage{amssymb}
\usepackage{color}
\usepackage{tikz}
\usepackage{mathrsfs}
\usepackage{mleftright}

\usepackage{lscape}
\usepackage{stmaryrd}
\usepackage{bussproofs}

\EnableBpAbbreviations

\usepackage{scalerel}
\usepackage{subcaption}
\usepackage{framed}

\usepackage{varwidth}
\newsavebox{\mypti}
\newsavebox{\mypto}
\newsavebox{\mypta}
\newsavebox{\myptu}
\newsavebox{\mypte}

\usepackage{tikz}
\usetikzlibrary{positioning,shapes,shapes.multipart}
\usetikzlibrary{cd,arrows,snakes,backgrounds}

\newtheorem{theorem}{Theorem}[section]
\newtheorem{example}{Example}[section]
\newtheorem{notation}{Notation}[section]
\newtheorem{definition}{Definition}[section]
\newtheorem{remark}{Remark}[section]
\newtheorem{corollary}{Corollary}[section]

\newtheorem{lemma}[theorem]{Lemma}
\newtheorem{proposition}[theorem]{Proposition}

\newcommand{\B}[1]{\mathbf{#1}}
\newcommand{\D}[1]{\mathscr{#1}}

\newcommand{\ff}{\mathsf{fold}}
\newcommand{\uu}{\mathsf{unfold}}
\newcommand{\inn}{\mathsf{in}}
\newcommand{\outt}{\mathsf{out}}

\newcommand{\pack}{\mathsf{pack}}
\newcommand{\unpack}{\mathsf{unpack}}

\newcommand{\TT}[1]{\mathtt{#1}}

\newcommand{\C}[1]{\mathcal{#1}}
\newcommand{\BB}[1]{\mathbb{#1}}
\newcommand{\U}[1]{\underline{#1}}
\newcommand{\OV}[1]{\mathbin{\overline{#1}}}
\newcommand{\F}[1]{\mathfrak{#1}}

\newcommand{\To}{\Rightarrow}

\newcommand{\Elim}{\rr{  \TT{El}}} 
\newcommand{\Intro}{\rr{  \TT{In}}}

\newcommand{\bb}[1]{\mathsf{#1}}
\newcommand{\rr}[1]{{\color{black}#1}}
\newcommand{\cc}[1]{{\color{blue}{#1}}}
\newcommand{\mm}[1]{{\color{red}{#1}}}

\newcommand{\at}[1]{\bb{at}(#1)}

\newcommand{\ar}{\mathsf{ar}}

\newcommand{\sss}{\mathsf{s}}
\newcommand{\zz}{\mathsf{0}}
\newcommand{\REC}{\mathsf{Rec}_{Y}}

\newcommand{\mcirc}{\circ_{\mathsf m}}

\newcommand{\NI}{\mathsf{\Lambda p}}
\newcommand{\NImu}{\mathsf{\Lambda p_{ \mu\nu}}}

\newcommand{\Nd}{\mathsf{\Lambda 2}}
\newcommand{\Ndv}{\mathsf{\Lambda 2p}}
\newcommand{\Ntot}{\mathsf{\Lambda 2p_{\mu\nu}}}

\newcommand{\NRP}{\mathsf{\Lambda 2^{Poly}}}

\newcommand{\op}{\mathsf{op}}

\newcommand{\step}{\vartriangleright}

\newcommand{\nuzzet}{\blacktriangle}
\newcommand{\stand}{\leadsto_{\mathsf{stand}}}

\newcommand{\NY}{\mathsf{\Lambda 2^{\kappa\leq 0}}}
\newcommand{\NYY}{\mathsf{\Lambda 2^{\kappa\leq 1}}}

\newcommand{\Nnew}{\mathsf{\Lambda 2p_{\mu,\nu}^{*}}}

\newcommand{\Nda}{\mathsf{\Lambda 2_{at}}}

   \newcommand{\FFun}[2]{  #2\left\langle #1\right\rangle}

  \newcommand{\Fun}[2]{ \Phi_{#2}^{#1}}

 \newcommand{\nat}{\mathtt{int}}
  \newcommand{\rec}{\mathsf{rec}}

\newcommand{\ctx}[0]{{\mathsf{ctx}}}
\newcommand{\CTX}{\BB C}

\newcommand{\rop}{\sharp}
\newcommand{\Rp}[1]{#1^{\rop}}

\definecolor{color0}{HTML}{4682B4}

\newcommand{\FV}[1]{\mathsf{f}\C V(#1)}
\newcommand{\BV}[1]{\mathsf{b}\C V(#1)}

\begin{document}

\maketitle

\begin{abstract}

In this paper we explore a family of type isomorphisms in System F whose validity corresponds, semantically, to some form of the Yoneda isomorphism from category theory. These isomorphisms hold under theories of equivalence stronger than $\beta\eta$-equivalence, like those induced by parametricity and dinaturality.
We show that the Yoneda type isomorphisms yield a rewriting over types, that we call Yoneda reduction, which can  be used to eliminate quantifiers from a polymorphic type, replacing them with a combination of monomorphic type constructors. We establish some sufficient conditions under which quantifiers can be fully eliminated from a polymorphic type, and we show some application of these conditions to count the inhabitants of a type and to compute program equivalence in some fragments of System F.

\end{abstract}

\tableofcontents

\section{Introduction}


The study of type isomorphisms is a fundamental one both in the theory of programming languages and in logic, through the well-known \emph{proofs-as-programs} correspondence: type isomorphisms supply programmers with transformations allowing them to obtain simpler and more optimized code, and offer new insights to understand and refine the syntax of type- and proof-systems.

Roughly speaking, two types $A,B$ are isomorphic when one can transform any call by a program to an object of type $A$ into a call to an object of type $B$, without altering the behavior of the program. Thus, type isomorphisms are tightly related to theories of \emph{program equivalence}, which describe what counts as the observable behavior of a program, so that programs with the same behavior can be considered equivalent.

 The connection between type isomorphisms and program equivalence is of special importance for polymorphic type systems like System F (hereafter $\Nd$). In fact, while standard $\beta\eta$-equivalence for $\Nd$ and the related isomorphisms are well-understood \cite{Bruce91, DiCosmo2005}, stronger notions of equivalence (as those based on \emph{parametricity} or \emph{free theorems} \cite{Wadler1989, Johann2005, Ahmed2017}) are often more useful in practice but are generally intractable or difficult to compute, and  
 little is known about the type isomorphisms holding under such theories.
 
 %

 

A cornerstone result of category theory, the \emph{Yoneda lemma}, 
 is sometimes invoked \cite{Bainbridge1990, Hasegawa2009,Bernardy2010,Uustalu2011,Hinze2010} to justify some type isomorphisms in $\Nd$ like e.g.
\begin{equation}\label{eqa}
\forall X. (A\To X) \To (B\To X) \equiv B\To A \qquad
\forall X. (X\To A)\To (X\To B)\equiv A\To B \tag{$\star$}
\end{equation}
\noindent which do not hold under  $\beta\eta$-equivalence, but only under stronger equivalences.  
 Such isomorphisms are usually justified by reference to the interpretation of polymorphic programs as \emph{(di)natural transformations} \cite{Bainbridge1990}, a well-known semantics of $\Nd$ related to both parametricity \cite{Plotkin1993} and free-theorems \cite{Voigt2020}, and yielding a not yet well-understood equational theory over the programs of $\Nd$ \cite{Girard1992, Delatail2009, FSCD2017}, that we call here the \emph{$\varepsilon$-theory}. 
Other isomorphisms, like
those in Fig.~\ref{fig:yonexamples}, can be justified in a similar way as soon as the language of $\Nd$ is enriched with other type constructors like $1,0,+,\times,\To$ and \emph{least}/\emph{greatest fixpoints} $\mu X.A, \nu X.A$.

All such type isomorphisms have the effect of  \emph{eliminating} a quantifier, replacing it with a combination of monomorphic type constructors, and can be used to test if a polymorphic type has a 
\emph{finite} number of inhabitants (as illustrated in Fig.~\ref{fig:exasso}) or, as suggested in \cite{Bernardy2010},  
to devise \emph{decidable tests} for program equivalence.

In this paper we develop a formal study of the elimination of quantifiers from polymorphic types using a class of type isomorphisms, that we will call \emph{Yoneda type isomorphisms}, 
 which generalize the examples above. Then, we explore the application of such type isomorphisms to establish properties of program equivalence for polymorphic programs.

\subsection{A Type-Rewriting Theory of Polymorphic Types}

In the first part of the paper (Sections  \ref{secCounting}-\ref{secCharacteristic}) we investigate the type-rewriting arising from Yoneda type isomorphisms and its connection with proof-theoretic techniques to count type inhabitants.

\subparagraph*{Counting type inhabitants with type isomorphisms}

Examples like the one in Fig. \ref{fig:exasso} suggest that, while arising from a categorical reading of polymorphic programs,  Yoneda Type Isomorphisms have a proof-theoretic flavor. 
To demonstrate this, we  show that the use of such isomorphisms to count type inhabitants can be compared with some well-known proof-theoretic 
\emph{sufficient conditions} for a \emph{simple} type $A$ to have a unique or finitely many inhabitants \cite{Aoto1999,Broda2005}, namely the \emph{balancedness}, \emph{negatively-non duplicated} and \emph{positively-non duplicated} conditions. We show that  
whenever an inhabited simple type $A$ satisfies any of the first two conditions (resp.~a special case of the third), then its universal closure $\forall \vec X.A$ can be converted to 1  by applying Yoneda type isomorphisms and usual $\beta\eta$-isomorphisms, and when $A$ satisfies a special case of the third, then it can be converted to $1+\dots+1$. Furthermore, we suggest that the appearance of the fixpoints in isomorphisms like $(**)$ can be related to a well-known approach to counting type inhabitants by solving systems of \emph{polynomial fixpoint equations} \cite{Zaoinc}.

\subparagraph*{Eliminating Quantifiers with Yoneda Reduction}

We then turn to investigate in a more systematic way the quantifier-eliminating rewriting over types arising from the left-to-right orientation of Yoneda type isomorphisms.
A major obstacle here is that the rewriting must take into account possible applications of $\beta\eta$-isomorphisms, whose axiomatization is well-known to be challenging 
 in presence of the constructors $+,0$ \cite{DiCosmo, Ilik} (not to say of $\mu,\nu$).
 For this reason we introduce 
a family of rewrite rules, that we call \emph{Yoneda reduction}, defined not directly over types but 
over a class of finite trees which represent the types of $\Nd$ (but crucially \emph{not} those made with $0,+,\dots$) up to  $\beta\eta$-isomorphism.
  
  Using this rewriting we establish some sufficient conditions for eliminating quantifiers, based on elementary graph-theoretic properties of such trees, which in turn provide some \emph{new} sufficient conditions for the finiteness of type inhabitants of polymorphic types. First, we prove quantifier-elimination for the types satisfying a certain \emph{coherence} condition which can be seen as an instance of the 2-SAT problem.  
 We then introduce a more refined condition by associating each polymorphic type $A$ with 
 a value  $\kappa(A)\in \{0,1,\infty\}$, that we call the \emph{characteristic of $A$}, so that whenever $\kappa(A)\neq \infty$, $A$ rewrites into a monomorphic type, and when furthermore $\kappa(A)=0$, $A$ converges to a finite type. In the last case our method provides an effective way to count the inhabitants of $A$.

The computation of $\kappa(A)$ is somehow reminiscent of linear logic \emph{proof-nets}, as it is obtained by inspecting the existence of cyclic paths in a graph obtained by adding some ``axiom-links'' to the tree-representation of $A$.

\begin{figure}
\adjustbox{scale=0.8, center}{
\parbox{\linewidth}{
\begin{align}
\forall X. X\To  X\To A & \equiv A[X\mapsto 1 + 1]  \tag{$*$}\\
\forall X.  (A\To X)\To (B\To X)\To C & \equiv   C[ X\mapsto \mu X.A+B] \tag{$**$}\\
\forall X. (X\To A) \To (X\To B)\To D & \equiv D[X\mapsto  \nu X.A\times B] \tag{$***$}
\end{align}
}
}
\vskip -0.3cm
\caption{\small Other examples of Yoneda type isomorphisms, where $X$ only occurs positively in $A,B,C$ and only occurs negatively in $D$.}
\label{fig:yonexamples}
\end{figure}
 \begin{figure}

\adjustbox{scale=0.8}{
$
\forall X.
X \To X \To
\forall Y. (\forall Z .(Z\To X)\To (\forall W.(W\To  Z)\To W\To X)\To Z\To Y)\To (X\To Y)\To Y
$}

\adjustbox{scale=0.8}{
$
\stackrel{(*)}{\equiv}
\forall Y. (\forall Z .(Z\To 1+1)\To (\forall W.(W\To  Z)\To W\To 1+1)\To Z\To Y)\To (1+1\To Y)\To Y
$}

\adjustbox{scale=0.8}{
$
\stackrel{\eqref{eqa}}{\equiv}
\forall Y. (\forall Z .(Z\To 1+1)\To(Z\To 1+1)\To Z\To  Y)\To (1+1\To Y)\To Y
$}

\adjustbox{scale=0.8}{
$
\stackrel{(***)}{\equiv}
\forall Y. (( \nu Z.(1+1)\times (1+1)) \To  Y)\To (1+1\To Y)\To Y
$}

\adjustbox{scale=0.8}{
$
\stackrel{(**)}{\equiv}
 \mu Y.(\nu Z.(1+1)\times (1+1))  + (1+1) \equiv 1+1+1+1+1+1
 $}

\caption{\small Short proof that a $\Nd$-type has 6 inhabitants, using type isomorphisms.}
\label{fig:exasso}
\end{figure}

\subsection{Program Equivalence in System F with Finite Characteristic}

 In the second part of the paper (Sections \ref{secEpsilon}-\ref{secAtomic}) we direct our attention to programs rather than types, and we exploit our results on type isomorphisms to establish some non-trivial properties of program equivalence for polymorphic programs in some suitable fragments of $\Nd$.

\subparagraph*{Computing equivalence with type isomorphisms}
Computing program equivalence under the $\varepsilon$-theory can be a challenging task, 
as this theory involves global permutations of rules which are difficult to detect and apply \cite{Delatail2009, FSCD2017, StudiaLogica, SL2, Voigt2020}.
Things are even worse at the semantic level, since computing with dinatural transformations can be rather cumbersome, due to the well-known fact that such transformations need not compose \cite{Bainbridge1990, Santamaria2018}. 

Nevertheless, our approach to quantifier-elimination based on the notion of characteristic provides a way to compute program equivalence \emph{without} the appeal to $\varepsilon$-rules, free theorems and parametricity, since all polymorphic programs having types of finite characteristic can be embedded 
inside well-known monomorphic systems. 
%
%
%
To demonstrate this fact, we introduce two fragments $\NY$ and $\NYY$ of $\Nd$ in which types have a fixed finite characteristic, and we prove that these are equivalent, under the $\varepsilon$-theory, respectively, to the simply typed $\lambda$-calculus with finite products and co-products (or, equivalently, to the \emph{free bicartesian closed category} $\BB B$), and to its extension with $\mu,\nu$-types (that is, to the \emph{free cartesian closed $\mu$-bicomplete category} $\mu\BB B$ \cite{Santocanale2002, Basold2016}). Using well-known facts about $\BB B$ and $\mu \BB B$ \cite{Scherer2017, Basold2016, Okada1999}, we finally establish that the $\varepsilon$-theory is decidable in $\NY$ and undecidable in $\NYY$.

\subparagraph*{Program equivalence and predicativity}

We provide an example of how the correspondence between polymorphic types of finite characteristic and $\mu,\nu$-types can be used to prove non-trivial properties of program equivalence. A main source of difficulty with $\Nd$ is that polymorphic programs are \emph{impredicative}, that is, a program of universal type $\forall X.A$ can be instantiated at \emph{any} type $B$, yielding a program of type $A[B/X]$.
It is thus useful to be able to predict when a complex type instantiation can be replaced by a one of smaller complexity, without altering the program behavior. 

Using the fact that a universal type $\forall X.A $ of finite characteristic in which $A$ is of the form $A_{1}\To \dots \To A_{n}\To X$ is isomorphic to the \emph{initial algebra} $\mu X.T$ of some appropriate functor $T$, we establish a sufficient condition under which a program containing an instantiation of $\forall X.A$ as $A[B/X]$ can be transformed into one with instantiations of types strictly less complex than $B$.

We finally use this condition to provide a simpler proof of a result from \cite{SL2}, showing that all programs in a certain fragment of $\NY$ (the fragment freely generated by the embedding of finite sums and products) can be transformed into predicative programs only containing \emph{atomic} type instantiations, a result related to some recent investigations on atomic polymorphism \cite{Ferreira2013}.

\subsection*{Preliminaries and Notations}

We will presuppose familiarity with the syntax of $\Nd$ (in the version \`a la Church) and its extensions $\Ndv,\Ntot$ with sum and product types, as well as $\mu$ and $\nu$-types. We indicate by $\mathsf{\Lambda}, \NI, \NImu$ their respective quantifier-free fragments.
We recall the syntax of these systems in some more detail in Sec.~\ref{secEpsilon}. 
 We let $\TT V=\{X,Y,Z,\dots\}$ indicate the countable set of \emph{type variables}.

 Let $\mathsf S$ indicate any of the type systems above.
We let $\Gamma\vdash_{\mathsf S}t:A$ indicate that the judgement $\Gamma\vdash t:A$ is derivable in $\mathsf S$.
We indicate as $t[x]$ a term with a unique free variable $x$, and we let $t[x]:A\vdash_{\mathsf S}^{\Gamma} B$ be shorthand for $\Gamma, x:A \vdash_{\mathsf S}t:B$.


A \emph{theory} of $\mathsf S$ is a class of equations over well-typed terms satisfying usual congruence rules. Standard theories of $\Nd, \Ndv, \Ntot$ are those generated by $\beta\eta$-equivalence and by contextual equivalence, recalled in Sec..~\ref{secEpsilon}. We will also consider a less standard theory, the $\varepsilon$-theory, also described in Sec.~\ref{secEpsilon}. 
For all theory $\mathsf T$ including $\beta\eta$-equivalence, we let $\CTX_{\mathsf T}({\mathsf S})$ be the category whose objects are the types of $\mathsf S$ and whose arrows are the $\mathsf T$-equivalence classes of terms 
$t[x]: A\vdash_{\mathsf S}B$. $\CTX_{\mathsf T}(\mathsf S)$ is cartesian closed as soon as $\mathsf S$ contains products, meaning in particular that
$\CTX_{\mathsf T}(\mathsf S)(A\times B, C)\simeq \CTX_{\mathsf T}(\mathsf S)(A,B\To C)$.  

By a \emph{$\mathsf T$-isomorphism}, indicated as $A\equiv_{\mathsf T}B$, we mean a pair of terms $t[x]:A\vdash_{\mathsf S}B$, $u[x]:B\vdash_{\mathsf T}A$ such that  $t[u[x]]\simeq_{\mathsf T}x$ and $u[t[x]]\simeq_{\mathsf T}x$ (where $t[u[x]]$ is $t[x\mapsto u]$).
%
%

 \section{Yoneda Type Isomorphisms }

In this section we introduce an axiomatization for a class of type isomorphisms that we call \emph{Yoneda type isomorphisms}. For this we will rely on a standard axiomatization of the $\beta\eta$-isomorphisms of $\Ntot$ (recalled in  Fig.~\ref{fig:isond}-\ref{fig:isomunu} in Sec.~\ref{secEpsilon})\footnote{\label{ft:incomplete}Note that while this axiomatization is complete for the $\beta\eta$-isomorphisms of $\Nd$ \cite{DiCosmo2005}, it fails to be complete (already at the propositional level) in presence of sums and the empty type \cite{DiCosmo, Ilik}.} and on the well-known distinction between \emph{positive} and \emph{negative} occurrences of a variable $X$ in a type $A$.  
\begin{notation}
Throughout the text we will indicate as $\cc X$ a {positive} occurrence of $X$, and as $\mm X$ a {negative} occurrence of $X$. When $B$ occurs within a larger type $A$, we  will often note $B$ as $\FFun{\cc X}{B}$ to indicate that all occurrences of the variable $X$ in $B$ are positive occurrences \emph{in $A$}, or as $\FFun{\mm X}{B}$  to indicate that all occurrences of the variable $X$ in $B$ are negative occurrences \emph{in $A$}. So for instance, when $B$ only contains positive occurrences of $X$, we write the type $A=X\To B$ as $\mm X \To \FFun{\cc X}{B}$ (since all positive occurrences of $X$ in $B$ are positive in $A$) and the type $A'=B\To X$ as $\FFun{\mm X}{B}\To \cc X$ (since all positive occurrences of $X$ in $B$ are negative in $A$) .
\end{notation}

The focus on positive/negative occurrences highlights a connection with to the so-called \emph{functorial semantics} of $\Nd$ \cite{Bainbridge1990, Girard1992}, in which types are interpreted as functors and typed programs as (di)natural transformations between such functors. More precisely, any positive type $\FFun{\cc X}{A}$ gives rise to a functor
$\Fun{X}{A}:\CTX_{\mathsf T}(\mathsf S)\to \CTX_{\mathsf T}(\mathsf S)$, any negative type $\FFun{\mm X}{A}$ gives rise to a functor
$\Fun{X}{A}:\CTX_{\mathsf T}(\mathsf S)^{\op}\to \CTX_{\mathsf T}(\mathsf S)$ and, more generally, any type 
$A$ gives rise to a functor
$\Fun{X}{A}:\CTX_{\mathsf T}(\mathsf S)^{\op}\times\CTX_{\mathsf T}(\mathsf S) \to \CTX_{\mathsf T}(\mathsf S)$.
In all such cases, the action of the functor on a type $A$ is obtained by replacing positive/negative occurrences of $X$ by $A$, and the action on programs can be defined inductively (we recall this construction in Sec.~\ref{secEpsilon}, see also \cite{Delatail2009, LMCS}).

With types being interpreted as functors, a polymorphic term $t[x]:A\vdash_{\Nd}  B$ is interpreted as a transformation satisfying an appropriate naturality condition: when $A$ and $B$ have the same variance, $t[x]$ is interpreted as an ordinary natural transformation; instead, if $A$ and $B$ have mixed variances, then $t[x]$ is interpreted as a dinatural transformation.

In the framework of syntactic categories, such (di)naturality conditions are expressed by families of equational rules over typed programs \cite{Delatail2009, FSCD2017}, generating, along with the usual $\beta\eta$-equations, a theory of program equivalence that we here call the \emph{$\varepsilon$-theory} (and which we define formally in Sec.~\ref{secEpsilon}).
These equational rules are usually interpreted as parametricity conditions \cite{Plotkin1993}, or as instances of \emph{free theorems} \cite{Voigt2020}. 

%
%
%
%
%
%
%
%

%

As mentioned in the introduction, our goal in these first sections is not that of investigating the $\varepsilon$-theory directly, but rather to explore a class of type isomorphisms that hold under this theory (that is, of isomorphisms in the syntactic categories $\CTX_{\varepsilon}(\mathsf S)$, with $\mathsf S=\Nd, \Ndv, \Ntot$). 
For example, in functorial semantics
 a type of the form $\forall X.\FFun{\mm X}{A}\To \FFun{\cc X}{B}$ is interpreted as the set of natural transformations between the functors $\FFun{\cc X}{A}$ and $\FFun{\cc X}{B}$. 
Now, if $\FFun{\cc X}{A}$ is of the form $A_{0}\To \cc X$ (i.e. it is a \emph{representable} functor), using the $\varepsilon$-theory we can deduce (see App.~\ref{app4}) a 
 ``Yoneda lemma'' in the form of the quantifier-eliminating isomorphism below:
\begin{equation}\label{yonea}
\forall X. (A_{0}\To \mm X) \To \FFun{\cc X}{B} \ \equiv \  \FFun{\cc X\mapsto A_{0}}{B}
\end{equation}
Similarly, if $\FFun{\mm X}{A}, \FFun{\mm X}{B}$ are both negative and $\FFun{\mm X}{A}$ is of the form $\mm X\To A_{0}$ (i.e. it is a \emph{co-representable} functor), we can deduce another quantifier-eliminating isomorphism:
\begin{equation}\label{yoneb}
\forall X. (\cc X\To A_{0}) \To \FFun{\mm X}{B} \ \equiv \  \FFun{\mm X\mapsto A_{0}}{B}
\end{equation}
Observe that both isomorphisms \eqref{eqa} from the Introduction are instances of \eqref{yonea} or \eqref{yoneb}.

As we admit more type-constructors in the language, we can use the $\varepsilon$-theory to deduce stronger schemas for eliminating quantifiers. For instance, using \emph{least} and \emph{greatest fixed points} $\mu X.\FFun{\cc X}{A}, \nu X.\FFun{\cc X}{A}$  of positive types, we can deduce the stronger schemas \cite{Uustalu2011} below.
\begin{align}
\forall X. (\FFun{\cc X}{A}\To \mm X)\To \FFun{\cc X}{B} &  \equiv \FFun{\cc X\mapsto \mu X.\FFun{\cc X}{A}}{B}  \label{yonec} \\
\forall X. (\cc X\To \FFun{\mm X}{A})\To \FFun{\mm X}{C} & \equiv\FFun{\mm X\mapsto \nu X.\FFun{\mm X}{A}}{C} \label{yoned}
\end{align}
Note that \eqref{yonea} and \eqref{yoneb} can be deduced from \eqref{yonec} and \eqref{yoned} using  the isomorphisms $\mu X.A\equiv_{\beta\eta}\nu X.A\equiv_{\beta\eta}A$, when $X$ does not occur in $A$.
Moreover, adding sum and product types enables the elimination of the quantifier $\forall X$ also from a type of the form
$A=\forall X. (\FFun{\cc X}{A_{11}}\To \FFun{\cc X}{A_{12}}\To \mm X)\To (\FFun{\cc X}{A_{21}}\To \FFun{\cc X}{A_{22}}\To \mm X)\To \FFun{\cc X}{B}$ by using  $\beta\eta$-isomorphisms, as follows:
\begin{center}
\adjustbox{center, scale=0.9}{
$
A \equiv_{\beta\eta}
\forall X. \left ({\left(\sum_{i=1,2}\prod_{j=1,2}\FFun{\cc X}{A_{ij}}\right)}\To \mm X\right )\To \FFun{\cc X}{B} 
\equiv\FFun{\cc X\mapsto  \mu X.\FFun{\cc X}{\left(\sum_{i=1,2}\prod_{j=1,2}A_{ij}}\right) }{B} $
}
\end{center}

These considerations lead to introduce the following class of isomorphisms.

\begin{notation}
Given a list $L=\langle i_{1},\dots, i_{k}\rangle$, and an $L$-indexed list of types $(A_{i})_{i\in L}$, we will use $\langle A_{i}\rangle_{i\in L}\To B$ as a shorthand  for 
$A_{i_{1}}\To \dots \To A_{i_{k}} \To B$. 

\end{notation}

\begin{definition}
A \emph{Yoneda type isomorphism} is any instance of the schemas $\equiv_{\cc X}$, $\equiv_{\mm X}$ in Fig.~\ref{fig:yonedazz}, 
where the expressions $\{\mu X.\}$, $\{\nu X.\}$ indicate that 
the binder $\mu X.$ (resp.~$\nu X.$) is applied only if $\cc X$ (resp.~$\mm X$) actually occurs in some of the $A_{jk}$ (resp.~$A_{j}$), and $\exists \vec Y.A$ is a shorthand for $\forall Y'.(\forall \vec Y.A\To Y')\To Y'$.
 
For all types $A,B$ of $\Ntot$, we write $A\equiv_{\bb Y}B$ when $A$ can be converted to $B$ using $\equiv_{\cc X}, \equiv_{\mm X}$ and the (partial, see f.n.~\ref{ft:incomplete}) axiomatization of $\beta\eta$-isomorphisms in Fig.~\ref{fig:isond}-\ref{fig:isomunu}.
\end{definition}

\begin{figure}
 $$
 \boxed{
    \begin{matrix}
    \forall X. \forall \vec Y. \Big\langle \forall \vec Z_{k}.\langle \FFun{\cc X}{ A_{jk}}\rangle_{j} \To \mm X\Big \rangle_{k}\To \FFun{\cc X}{B} 
\quad  \equiv_{\cc X}\quad
 \forall\vec Y. \FFun{\cc X\mapsto  \{\mu X.\} \sum_{k}\left ( \exists \vec Z_{k}. \prod_{j}\FFun{\cc X}{A_{jk}}\right)     }{B}  \\ \ \\ 
%
%
%
   \forall X. \forall \vec Y.\left  \langle \forall \vec Z_{k}. \cc X\To  \FFun{\mm X}{A_{j}}\right\rangle_{k}\To \FFun{\mm X}{B}
  \quad \equiv_{\mm X}\quad
  \forall\vec Y. \FFun{\mm X\mapsto \{ \nu X.\} \forall \vec Z_{k}. \prod_{j}\FFun{\mm X}{A_{j}}     }{B}
  \end{matrix}
  }
$$
\caption{The two schemas $\equiv_{\cc X}$ and $\equiv_{\mm X}$ of Yoneda Type Isomorphisms.}
\label{fig:yonedazz}
\end{figure}


\begin{figure}
\begin{subfigure}{\textwidth}
$$
\boxed{
\begin{matrix}
A\To (B\To C) \equiv B\To (A\To C) \\ \\
\forall X.\forall Y.A \equiv \forall Y.\forall X.A
\\  \\
A\equiv \forall X.A \quad (X\notin FV(A)) \qquad A\To \forall X.B \equiv \forall X.A\To B
\end{matrix}
}
$$
\caption{Axiomatization of $\beta\eta$-isomorphisms for $\Nd$.}
\label{fig:isond}
\end{subfigure}

\medskip

\begin{subfigure}{\textwidth}
$$
\boxed{
\begin{matrix}
A \times 1 \equiv A \qquad  1\To A \equiv A \\ \\ 
A \times (B\times C) \equiv (A\times B)\times C \qquad A \times B \equiv B\times A 
\\ \\ 
(A\times B)\To C \equiv A\To (B\To C)
\qquad
\forall X.A\times B \equiv \forall X.A \times \forall X.B \\ \\
A + 0 \equiv  A \quad A\times 0 \equiv 0 \qquad 0 \To A \equiv 1\\ \\ 
A+ (B+ C) \equiv (A+ B)+ C \qquad A+B \equiv B+ A \\ \\ 
A\times (B+C) \equiv (A\times B)+(A\times C) \qquad
(A+B)\To C \equiv (A\To C)\times (B\To C)
\end{matrix}
}
$$
\caption{Partial axiomatization of $\beta\eta$-isomorphisms for $\Ndv$.}
\label{fig:isondv}
\end{subfigure}

\medskip

\begin{subfigure}{\textwidth}$$
\boxed{
\begin{matrix}
\mu X.A \equiv \nu X.A\equiv A \quad (X\notin FV(A)) \\ \\
\mu X.\FFun{\cc X}{A} \equiv \FFun{\cc X\mapsto \mu X.\FFun{\cc X}{A}}{A}
\qquad
\nu X.\FFun{\cc X}{A} \equiv \FFun{\cc X\mapsto \nu X.\FFun{\cc X}{A}}{A}
\end{matrix}
}
$$
\caption{Partial axiomatization of $\beta\eta$-isomorphisms for $\mu,\nu$-types.}
\label{fig:isomunu}
\end{subfigure}
\end{figure}


In App.~\ref{app4} we prove in detail that the two schemas $\equiv_{\cc X}, \equiv_{\mm X}$ are valid under the $\varepsilon$-theory, and thus that whenever $A\equiv_{\bb Y}B$, $A$ and $B$ are interpreted as isomorphic objects in all dinatural and parametric models of $\Nd$.

\section{Counting Type Inhabitants with Yoneda Type Isomorphisms}\label{secCounting}




In the previous section we introduced and motivated Yoneda Type Isomorphisms as arising from a categorical reading of polymorphic programs as (di)natural transformations. However, such isomorphisms also have a proof-theoretic flavor since, as we observed in the Introduction, they can be used to check if a closed polymorphic type has a finite number of inhabitants (up to contextual equivalence), by showing that $A$ can be transformed into a closed quantifier-free type $A'$, 
and in this case to effectively \emph{count} such inhabitants (by reducing $A'$ to a finite sum $0+1+\dots+1=\sum_{i=1}^{k}1$).

%
%

At least in some simple cases, the way in which a quantifier $\forall X.A$ is eliminated by applying the schemas $\equiv_{\cc X}, \equiv_{\mm X}$ can be related with the structure of the $\beta$-normal $\eta$-long terms of type $A$. For instance, suppose 
$A$ is the simple type $ ( C_{11}\To C_{12}\To  X) \To (C_{21}\To C_{22}\To X)\To X$ (where none of the $C_{ij}$ contains occurrences of $X$) and $t$ is a $\beta$-normal $\eta$-long term of type $A$. Then $t$ must be of one of the following forms:
\begin{equation}\label{eq:forms}
\begin{split}
t& = \lambda y_{1}y_{2}. y_{1}u_{11}u_{12} \\
t& = \lambda y_{1}y_{2}. y_{2}u_{21}u_{22}
\end{split}
\end{equation}
Since $X$ does not occur in the $C_{ij}$ we can suppose that $y_{1},y_{2}$ do not occur in the terms $u_{ij}$, so that $u_{ij}$ is a closed term of type $C_{ij}$. Hence a closed term of $A$ is obtained from closed terms of $C_{11}$ \emph{and} $C_{12}$, \emph{or} from closed terms of $C_{21}$ \emph{and} $C_{22}$, and we deduce that the number of inhabitants of $A$ must be the same as the number of inhabitants of $\sum_{i}\prod_{j}C_{ij}$, which is precisely the type we would obtain by applying $\equiv_{\cc X}$ to $\forall X.A$. 

The situation becomes tricker if we admit the types $C_{ij}$ to have positive occurrences of $X$, so that $\equiv_{\cc X}$ yields the fixpoint type $\mu X. \sum_{i}\prod_{j}\FFun{\cc X}{C_{ij}}$. In this case all we can say is that the closed terms $\lambda y_{1}y_{2}.u_{ij}$ have type $C'_{ij}:=( C_{11}\To C_{12}\To  X) \To (C_{21}\To C_{22}\To X)\To C_{ij}$. In the concluding section we suggest that a purely proof-theoretic explanation of $\equiv_{\cc X}$ could be developed also in this case 
by considering a technique of counting type-inhabitants of simple types by solving a system of \emph{polynomial fixpoint equations} \cite{Zaoinc}.

%

%
%
%
%
%
%
 
%
%
%
%
%
%

A natural question about counting type inhabitants using type isomorphisms is whether this approach can be compared with existing proof-theoretic criteria for having a finite number of inhabitants. We provide a first answer to this question by comparing the use of Yoneda Type Isomorphism with some well-known criteria for the finiteness/uniqueness of the inhabitants of a \emph{simple} type.

Let us start with a ``warm-up'' example.

\begin{example}\label{ex:broda}
In \cite{Broda2005} it is proved that the type $A=(((  (\cc X\To \mm Y)   \To \mm X\To   \cc Z)\To   ( \mm Y \To \cc Z )  \To   \mm W ) \To  (\cc Y\To \mm Z)\To \cc W$ has a unique inhabitant. Here's a quick proof of   
$\forall XYZW.A\equiv_{\bb Y}1$:

\medskip
\begin{tabular}{l}
$ \forall XYZW. (((  (\cc X\To \mm Y)   \To \mm X\To   \cc Z)\To   ( \mm Y \To \cc Z )  \To   \mm W ) \To  (\cc Y\To \mm Z)\To \cc W$  \\
$ \equiv_{\cc W}
\forall XYZ . (\cc Y\To \mm Z) \To\Big (\big((\cc X\To \mm Y)\To \mm X\To \cc Z\big)\times \big ( \mm Y\To \cc Z\big)\Big)$
\\
$ \equiv_{\cc Z}
\forall XY . \big((\cc X\To \mm Y)\To \mm X\To \cc Y\big)\times \big ( \mm Y\To \cc Y\big)$
\\
$ \equiv_{\beta\eta}
\Big(\forall XY . \big((\cc X\To \mm Y)\To \mm X\To \cc Y\big)\Big)
\times
\Big(\forall Y . \mm Y\To \cc Y\Big)$
\\
$ \equiv_{\mm X}
\Big(\forall Y . \mm Y\To \cc Y\Big) \times \Big(\forall Y . \mm Y\To \cc Y\Big)
 \equiv_{\cc Y}
1\times 1 \equiv_{\beta\eta}1$
\end{tabular}

\end{example}


The literature on counting simple type inhabitants is vast (e.g. \cite{Aoto1999,Broda2005,Salvati2011,Scherer2015}) and includes both 
complete algorithms and simpler \emph{sufficient conditions} for a given type to have a unique or finite number of inhabitants. 
The latter provide then an ideal starting point to test our axiomatic theory of type isomorphisms, 
as several of these conditions are based on properties like the number of positive/negative occurrences of variables.

We tested Yoneda type isomorphisms on two well-known sufficient conditions for unique inhabitation.
A simple type $A$ is \emph{balanced} when any variable occurring free in $A$ occurs exactly once as $\cc X$ and exactly once as $\mm X$.  $A$ is \emph{negatively non-duplicated} if no variable occurs twice in $A$ as $\mm X$. An inhabited simple type which is either balanced or negatively duplicated has exactly one inhabitant \cite{Aoto1999}. 
%
The isomorphisms $\equiv_{ \mathsf Y}$ 
subsume these two conditions in the following sense:%

\begin{restatable}{proposition}{unique}\label{prop:unique}
Let $A[X_{1},\dots, X_{n}]$ be an inhabited simple type with free variables $X_{1},\dots, X_{n}$.
If $A[X_{1},\dots, X_{n}]$ is either balanced or negatively non-duplicated, then $\forall  X_{1}\dots \forall X_{n}.A \equiv_{ \mathsf Y} 1$.
%
%
%
\end{restatable}
\begin{proof}
Write $A$ as $\langle A_{k}\rangle_{k\in K}\To X_{u}$ and 
suppose $t= \lambda x_{1}.\dots x_{\sharp K}. x_{v} t_{1}\dots t_{p}$ is a closed term of type $A$ in  $\eta$-long $\beta$-normal form. This forces
$A_{v}= \langle A_{jv}\rangle_{j\in J_{v}}\To X_{u}$, with $\sharp J_{v}=p$, and 
implies that the terms $t_{j}^{*}=\lambda x_{1}\dots x_{\sharp K}. t_{j}$ have type
$\langle A_{k}\rangle_{k\in K}\To A_{jv}$.

We argue by induction on the size of $t$. 

\begin{itemize}
\item suppose $A$ is balanced. Since $X_{u}$ occurs positively in the rightmost position of $A$ and negatively in the rightmost position of $A_{v}$, it cannot occur in any other $A_{w}$, for $w\neq v$. We deduce that the variable $x_{v}$ cannot occur in any of the terms $t_{j}$.
Then, up to some $\beta\eta$-isomorphisms, we can write $\forall \vec X. A$ as 
$\forall \vec X_{v\neq u}.\forall X_{v}. \left( \langle A_{jv}\rangle_{j\in J_{v}}\To  \mm X_{u} \right) \To 
\left (\langle A_{w}\rangle_{w\neq v} \To \cc X_{u}\right)$, and by applying isomorphism $\equiv_{\cc X}$ we obtain the type 
$\forall \vec X_{v\neq u}. \langle A_{w}\rangle_{w\neq v} \To \prod_{j\in J_{v}}A_{jv} $.

Since the terms $\lambda x_{1}\dots \hat{x}_{v}\dots x_{\sharp K}.t_{j}$ have the balanced type
$A_{j}^{*}=\langle A_{w}\rangle_{w\neq v}\To A_{jv}$, by the induction hypothesis we deduce $A_{j}^{*}\equiv_{\bb Y}1$, and thus   we deduce 
$\forall \vec X_{v\neq u}. \langle A_{w}\rangle_{w\neq v} \To \prod_{j\in J_{v}}A_{jv} 
\equiv_{\beta\eta}
\forall \vec X_{v\neq u}. \prod_{j\in J_{v}}A_{j}^{*}
\equiv_{\beta\eta}
\prod_{j\in J_{v}}\forall \vec X_{v\neq u}A_{j}^{*}
\equiv_{\bb Y} \prod_{j\in J_{v}}1
\equiv_{\beta\eta}1
$.

\item Suppose $A$ is negatively non-duplicated; we will exploit the fact that $A$ admits at most one $\eta$-long $\beta$-normal term up to equivalence (see \cite{Aoto1999}).
We argue by induction that $A$ reduces to $1$ using $\beta\eta$-isomorphisms and isomorphisms of the form $\equiv_{X^{c}}$, where the variable $X$ occurs negatively exactly once.
 Since $X_{u}$ occurs negatively once in rightmost position of $A_{v}$, it can only occur positively in the $A_{jv}$ and negatively in the $A_{w}$ for $w\neq v$. 
 
 Moreover, we claim that $x_{v}$ cannot occur in the $t_{j}$. In fact, suppose some $t_{j}$ contains $x_{v}$. Since it is $\eta$-long  it contains a subterm of the form $x_{v}w_{1}\dots w_{p}$, where $w_{j}$ cannot be $\beta\eta$-equivalent to $t_{j}$. Then by replacing $w_{j}$ with $t_{j}$ in $t$ we obtain a new term $t_{1}$ of type $A$. It is clear that this process can be repeated yielding infinitely many distinct $\beta$-normal and $\eta$-long terms $t_{2},t_{3},\dots$ of type $A$. This is impossible because $A$ admits only one term up to $\beta\eta$-equivalence. 
%
%
%
%

Hence, up to some $\beta\eta$-isomorphisms, we can write $\forall \vec X. A$ as 
$\forall \vec X_{w\neq u}.\forall X_{v}. \left( \langle \FFun{\cc X_{u}}{A_{jv}}\rangle_{j\in J_{v}}\To \mm X_{u} \right) \To 
\left (\langle \FFun{\cc X_{u}}{A_{w}}\rangle_{w\neq v} \To \cc  X_{u}\right)$, and by applying $\equiv_{\cc X}$ we obtain the type 
$\forall \vec X_{w\neq u}. \langle \FFun{\cc X_{u}\mapsto \alpha}{A_{w}}\rangle_{w\neq v} \To\alpha$, where 
$\alpha=\mu X_{u}. \prod_{j\in J_{v}}\FFun{\cc X_{u}}{A_{jv}} $,
 which is $\beta\eta$-isomorphic to
$\forall \vec X_{w\neq u}. \langle \FFun{\cc X_{u}\mapsto \alpha}{A_{w}}\rangle_{w\neq v} \To \Pi_{j\in J_{v}}\FFun{\cc X\mapsto \alpha}{A_{jv}}$, and finally to  
  $\prod_{j\in J_{v}}\left(\forall \vec X_{w\neq u}. \langle \FFun{\cc X_{u}\mapsto \alpha}{A_{w}}\rangle_{w\neq v} \To \FFun{\cc X\mapsto \alpha}{A_{jv}}\right)$.

Since the terms $\lambda x_{1}\dots \hat{x}_{v}\dots x_{\sharp K}.t_{j}$ have the negatively non-duplicated type
$\FFun{\cc X_{u}}{A_{j}^{*}}=\langle \FFun{\cc X_{u}}{A_{w}}\rangle_{w\neq v}\To \FFun{\cc X_{u}}{A_{jv}}$, by the induction hypothesis we deduce $\forall \vec X.\FFun{\cc X_{u}}{A_{j}^{*}}\equiv_{\bb Y}1$
using $\beta\eta$-isomorphisms and isomorphisms of the form $\equiv_{X_{w}^{c}}$, where $w\neq v$, so we deduce that also $\forall \vec X_{w\neq u}\FFun{\cc X_{u}\mapsto \alpha}{A_{j}^{*}}\equiv_{\bb Y}1$ holds. We can then conclude
$\forall \vec X.A\equiv_{\bb Y} \prod_{j\in J_{v}}\forall \vec X_{v\neq u}\FFun{\cc X_{u}\mapsto \alpha}{A_{j}^{*}}
\equiv_{\bb Y} \prod_{j\in J_{v}}1
\equiv_{\beta\eta}1$.
\end{itemize}
\end{proof}

Observe that, since the type in Example \ref{ex:broda} is neither balanced nor negatively non-duplicated, type isomorphisms
 provide a stronger condition than the two above.

We tested another well-known property, dual to one of the previous ones: a simple type $A$ is \emph{positively non-duplicated} if no variable occurs twice in $A$ as $\cc X$. A positively non-duplicated simple type has a finite number of proofs \cite{Broda2005}.
We reproved this fact using type isomorphisms, but this time only in a restricted case.
Let the \emph{depth} $d(A)$ of a simple type $A$ be defined by $d(X)=0$, $d(A\To B)= \max\{d(A)+1, d(B)\}$.

\begin{restatable}{proposition}{positive}\label{prop:many}
Let $A[X_{1},\dots, X_{n}]$ be an inhabited simple type with free variables $X_{1},\dots, X_{n}$.
If $A$ is positively non-duplicated and $d(A)\leq 2$, then 
$\forall X_{1}\dots \forall X_{n}.A \equiv_{ \mathsf Y} 0+1+\dots+ 1$.
\end{restatable}

We need the following lemma:
\begin{lemma}\label{lemma:depth}
Suppose  $A= \langle A_{i}\To X_{i}\rangle_{i\in L}\To U$ has free variables $X_{1},\dots, X_{n}$, where the $A_{i}$ are $\To$-free $\NImu$-types and $U$ is a closed type constructed using only 0 and 1. Then $\forall  X_{1}\dots \forall X_{n}.A\equiv_{\bb Y}U$. 

\end{lemma}
\begin{proof}
We argue by induction on the number of variables $X_{1},\dots, X_{n}$ freely occurring in $A$.
If $n$ is 0 then the list  $L$ is empty, so the claim is obvious. Otherwise, suppose $L$ has a first element $i_{0}$, so that 
$A$ can be written, up to some $\beta\eta$-isomorphism, as 
$ 
\langle \FFun{\cc X_{i_{0}}}{A_{u}}\To\mm X_{i_{0}}\rangle_{u}\To
\langle   \FFun{\cc X_{i_{0}} }{B_{v}}\To\mm Y_{v}\rangle_{v} \To  U$
where the $ \FFun{\cc X_{i_{0}} }{A_{u}}$ and  $ \FFun{\cc X_{i_{0}} }{A_{u}}$ are $\To$-free types  and $ Y_{v}\neq X$.
Then $\forall  X_{1}\dots \forall X_{n}.A\equiv_{\cc X_{i_{0}}}
\forall X_{1}\dots \widehat{\forall X}_{i_{0}}\dots \forall X_{n}.
\langle   \FFun{\cc X_{i_{0}} }{B_{v}}\theta\To\mm Y_{v}\rangle_{v} \To  U$, where
$\theta: \cc X_{i_{0}}\mapsto \mu X_{i_{0}}. \sum_{u}\FFun{\cc X_{i_{0}}}A_{u}$, so we can conclude by the induction hypothesis.
\end{proof}

\begin{proof}[Proof of Proposition \ref{prop:many}]
We will establish the following stronger claim. Let a \emph{quasi-polynomial} be a type of the form $\langle A_{k}\rangle_{k\in K}\To
\sum_{i=1}^{p}\prod_{j=1}^{p_{i}} u_{ij}$, where $p, p_{i}\geq 1$, the $A_{k}$ are simple types of depth $\leq 1$ (hence of the form $\langle X_{v_{k}}\rangle_{v_{k}}\To X_{k}$) and 
the $u_{ij}$ are either variables or 0,1. We claim that for any quasi-polynomial  $P[X_{1},\dots, X_{n}]$, if $P$ is positively non-duplicated, then  $\forall X_{1}\dots \forall X_{n}.P\equiv_{\bb Y}0+1+\dots +1$.

We argue by induction on the number of variables occurring among the $u_{ij}$.

If all $u_{ij}$ are either $0$ or $1$, then any $A_{k}=\langle X_{v_{k}}\rangle_{v_{k}}\To X_{k}$ is $\beta\eta$-isomorphic to  $(\Pi_{v_{k}}X_{v_{k}})\To X_{k}$ so we can apply Lemma \ref{lemma:depth} yielding $\forall X_{1}\dots \forall X_{n}.P\equiv_{\bb Y}
\sum_{i=1}^{p}\prod_{j=1}^{p_{i}} u_{ij}\equiv_{\beta\eta}
0+1+\dots +1$.

Suppose now $u_{i_{0}j_{0}}=Z$. Since $P$ is positively non-duplicated, we can write $P$, up to some $\beta\eta$-isomorphisms, as 
$
\langle \langle \cc Z_{lm}\rangle_{l}\To \mm Z\rangle_{m}\To \langle B_{m'}\rangle_{m'}\To \sum_{i=1}^{p}\prod_{j=1}^{p_{i}} u_{ij}
$
where none of $Z_{lm}$ coincides with $Z$ and $Z$ does not occur in any of the $B_{m'}$. Then we deduce 
$\forall X_{1}\dots \forall X_{n}.P \equiv_{\cc Z}
\forall X_{1}\dots \widehat{\forall Z}\dots \forall X_{n}.
 \langle B_{m'}\rangle_{m'}\To \sum_{i=1}^{p}\prod_{j=1}^{p_{i}} u_{ij}\theta
$, where $\theta: \cc Z \mapsto \sum_{m}\prod_{l}\cc Z_{lm}$. Observe that the type  $\langle B_{m'}\rangle_{m'}\To \sum_{i=1}^{p}\prod_{j=1}^{p_{i}} u_{ij}\theta$
 is positively non-duplicated, has one variable less than $P$, and moreover, reduces using $\beta\eta$-isomorphisms to a type of the form $\langle B_{m'}\rangle_{m'}\To \sum_{i=1}^{p'}\prod_{j=1}^{p'_{i}} v_{ij}$, where the $v_{ij}$ are among $0,1,X_{1},\dots, \widehat Z, \dots, X_{n}$. We can thus apply the induction hypothesis.
\end{proof}

%

\section{From Polymorphic Types to Polynomial Trees}



As we already observed, read from left to right, the schemas $\equiv_{\cc X},\equiv_{\mm X}$ yield rewriting rules over $\Ntot$-types which \emph{eliminate} occurrences of polymorphic quantifiers. 
Yet, a major obstacle to study this rewriting is that the application of $\equiv_{\cc X},\equiv_{\mm X}$ might depend on the former application of $\beta\eta$-isomorphisms (as we did for instance in the previous section). Already for the propositional fragment $\NI$, the $\beta\eta$-isomorphisms are not finitely axiomatizable and it is not yet clear if a decision algorithm exists at all (see \cite{DiCosmo, Ilik}).
This implies in particular that a \emph{complete} criterion for the conversion of a $\Nd$-type to a monomorphic (or even finite) type can hardly be computable.

For this reason, we will restrict our goal to establishing some efficiently recognizable (in fact, polytime) \emph{sufficient conditions} for quantifier-elimination. Moreover, 
we will exploit the well-known fact that the constructors $0,1,+,\times,\mu,\nu$ can be encoded inside $\Nd$ to describe our rewriting entirely within (a suitable representation of) $\Nd$-types, for which $\beta\eta$-isomorphisms can be finitely axiomatized \cite{DiCosmo2005}.

%
%
%
%

Even if one restricts to $\Nd$-types, recognizing if one of the schemas $\equiv_{\cc X}, \equiv_{\mm X}$ applies to a $\Nd$-type $\forall X.A$ might still require to first apply some $\beta\eta$-isomorphisms.  For example, consider the $\Nd$-type $A=\forall X.  (( \mm Y\To \cc X)\To \mm Y)\To (\cc Y\To  \mm X)\To \cc Y$.
In order to eliminate the quantifier $\forall X$ using $\equiv_{\cc X}$, we first need to apply the $\beta\eta$-isomorphism $A\To (B\To C)\equiv_{\beta\eta} B\To (A\To C)$, turning $A$ into 
 $\forall X.   (\cc Y\To  \mm X)\To(( \mm Y\To \cc X)\To \mm Y)\To  \cc Y$, which is now of the form
 $\forall X. (\FFun{\cc X}{B}\To \mm X)\To \FFun{\cc X}{C}$, with $\FFun{\cc X}{B}=\cc Y$ and 
 $\FFun{\cc X}{C}=((\mm Y\To \cc X)\To \mm Y)\To \cc Y$. We can then apply $\equiv_{\cc X}$, yielding
 $((\mm Y\To \cc Y)\To \mm Y)\To \cc Y$.

To obviate this problem, we introduce below a representation of $\Nd$-types as labeled trees so that $\beta\eta$-isomorphic types are represented by the same tree. In the next section we will 
reformulate the schemas $\equiv_{\cc X},\equiv_{\mm X}$ as reduction rules over such trees.
This approach drastically simplifies the study of this rewriting, 
and will allow us to establish conditions for quantifier-elimination based on elementary graph-theoretic properties.

\subsection{A $\beta\eta$-invariant representation of $\Nd$-types}

We introduce a representations of $\Nd$-types as rooted trees whose leaves are labeled by \emph{colored} variables, with {colors} being any $c\in \mathsf{Colors}=\{\cc{\mathrm{blue}}, \mm{\mathrm{red}}\}$. We indicate such variables as either $X^{c}$ or simply as $\cc X, \mm X$. Moreover, we indicate as $\OV c$ the unique color different from $c$. 

By a rooted tree we indicate a finite connected acyclic graph with a chosen vertex, called its root. 
If $(\D G_{i})_{i\in I}$ is a finite family of rooted trees, we indicate as 
$
\begin{tikzpicture}[baseline=-12]
\node (a) at (0,0) {\small$\big\{ \D G_{i}\big\}_{i\in I}$};

\node (b) at (0,-0.8) {\small$u$};

\draw(a) to (b);

\end{tikzpicture}
$
the tree with root $u$ obtained by adding an edge from any of the roots of the trees $\D G_{i}$ to $u$.

\begin{definition}
The sets $\cc{\C E}$ and $\mm{\C E}$ of \emph{positive} and \emph{negative $\Nd$-trees} are inductively defined by:

\adjustbox{scale=0.85, center}{$\cc{\bb E} \quad := \quad  
	\begin{tikzpicture}[baseline=-6]
	\node(a) at (0,0) {\small $\Big \{ \mm{\bb E}_{i}\Big \}_{i\in I}$};
	\node (b) at (1.4,0) {\small$ \cc X$};
	\node (c) at (0.7,-0.8) {\small$\vec{Y}$};
	\draw (a) to (c);
	\draw (b) to (c);
	\end{tikzpicture} 
	\qquad\qquad\qquad\qquad
	\mm{\bb E} \quad := \quad  
	\begin{tikzpicture}[baseline=-6]
	\node(a) at (0,0) {\small $\Big \{ \cc{\bb E}_{i}\Big \}_{i\in I}$};
	\node (b) at (1.4,0) {\small$ \mm X$};
	\node (c) at (0.7,-0.8) {\small$\vec{Y}$};
	\draw (a) to (c);
	\draw (b) to (c);
	\end{tikzpicture} 		
	$}
	
	\noindent
where $X$ is a variable, $\vec{Y}$ indicates a finite set of variables, and the edge  in $\cc{\bb E}$ (resp. $\mm{\bb E}$) with label $\cc X$ (resp. $\mm X$) is called the \emph{head} of $\cc{\bb E}$ (resp. of $\mm{\bb E}$).	The trees \adjustbox{scale=0.8}{$\begin{tikzpicture}[baseline=-12]
		\node(a0) at (-0.5,0) {\small $\{ \ \}_{ \emptyset}$};
	\node(a) at (0.5,0) {\small $\cc X$};
	\node (b) at (0,-0.7) {\small$\emptyset$};
	\draw (a) to (b);
		\draw (a0) to (b);

	\end{tikzpicture}$}
and 	
\adjustbox{scale=0.8}{$\begin{tikzpicture}[baseline=-12]
		\node(a0) at (-0.5,0) {\small $\{ \ \}_{ \emptyset}$};
	\node(a) at (0.5,0) {\small $\mm X$};
	\node (b) at (0,-0.7) {\small$\emptyset$};
	\draw (a) to (b);
		\draw (a0) to (b);

	\end{tikzpicture} $}
	are indicated simply as $\cc X$ and $\mm X$.

\end{definition}

 \emph{Free} and \emph{bound} 
variables of a tree 
$\bb E= $ \adjustbox{scale=0.8}{$
\begin{tikzpicture}[baseline=-12]
	\node(a) at (0,0) {\small $\Big \{ \bb E_{i}\Big \}_{i\in I}$};
	\node (b) at (1.4,0) {\small$ X$};
	\node (c) at (0.7,-0.7) {\small$\vec{Y}$};
	\draw (a) to (c);
	\draw (b) to (c);
	\end{tikzpicture}  $} are defined by $
\FV{\bb E}= \bigcup_{i=1}^{n}\FV{\bb E_{i}}\cup \{X\}- \vec{Y}$ and $\BV{\bb E}= \bigcup_{i=1}^{n}\BV{\bb E_{i}}\cup \vec{Y}$.

%


We can associate a positive and a negative $\Nd$-tree to any $\Nd$-type as follows.
Let us say that a type $A$ of $\Nd$ is \emph{in normal form} (shortly, \emph{in NF}) if $A= \forall \vec Y.A_{1}\To  \dots \To  A_{n} \To X$ where each of the variables in $\vec Y$ occurs in $A_{1}\To\dots \To A_{n}\To X$ at least once, and the types $A_{1},\dots, A_{n}$ are in normal form.

We will show that the tree-representation of $\Nd$-types captures $\beta\eta$-isomorphism classes, in the sense that 
$A\equiv_{\beta\eta} B$ iff $\cc{\mathsf t}(A)= \cc{\mathsf t}(B)$. For instance, the two $\beta\eta$-isomorphic types 
$\forall XY. (\cc X\To \mm X)\To  (\forall Z.\cc X\To \mm Z)\To (\cc Y\To  \mm X)\To \cc Y$ and 
 $\forall X.   (\cc X\To \mm X)\To \forall Y. (\cc Y\To  \mm X)\To  (\cc X\To\forall Z. \mm Z)\To \cc Y$
 translate into the same $\Nd$-tree, shown in Fig.~\ref{fig:root3} (where underlined node labels and dashed edges can be ignored, for now). 

The $\beta\eta$-isomorphisms of $\Nd$ are completely axiomatized as shown in Fig. \ref{fig:isond}. A type $A$ of $\Nd$ is \emph{in normal form} (shortly, \emph{in NF}) if $A= \forall \vec Y.A_{1}\To  \dots \To  A_{n} \To X$ where each variable in $\vec Y$ occurs at least once in $A_{1}\To  \dots \To  A_{n} \To X$, and the types $A_{1},\dots, A_{n}$ are in normal form. Let $\bb{NF}(\Nd)$ be the set of $\Nd$-types in NF.

Given $A, B\in \bb{NF}(\Nd)$, with $A= \forall \vec Y.A_{1}\To  \dots \To  A_{n} \To X$ and 
$B= \forall \vec Z.B_{1}\To  \dots \To  B_{m} \To X'$, let $A \sim B$ when $n=m$, $\vec Y=\vec Z$, $X=X'$ and there exists 
$\sigma\in \F S_{n}$ such that $A_{i}\sim B_{\sigma(i)}$ for $i=1,\dots,n$. 

\begin{lemma}\label{lemma:nf}
\begin{itemize}
\item[i.] For all $A\in \Nd$, there exists a type $\bb{NF}(A)\in \bb{NF}(\Nd)$ such that $A\equiv_{\beta\eta}\bb{NF}(A)$.
\item[ii.] For all $A,B\in \Nd$, $A\equiv_{\beta\eta} B$ iff $\bb{NF}(A) \sim \bb{NF}(B)$.

\end{itemize}
\end{lemma}
\begin{proof}
Claim $i.$ is a consequence of applying the isomorphisms $\forall X.A\equiv_{\beta\eta} A$ when $X$ does not occur in $A$ and the isomorphisms 
 $A\To \forall X.B \equiv_{\beta\eta} \forall X. A\To B$ to push quantifiers in leftmost position. Claim $ii.$ follows then since the only other isomorphism rule is the permutation $A\To (B \To C) \equiv_{\beta\eta} B\To (A \To C)$.
\end{proof}

The definition of the tree-type $\cc{\mathsf t}(A)$ only depends on $\bb{NF}(A)$. Moreover, it can be checked by induction on $A$ that $\bb{NF}(A)\sim \bb{NF}(B)$ holds iff $\cc{\mathsf t}(A)=\cc{\mathsf t}(B)$. Hence, we deduce from Lemma \ref{lemma:nf} that $A\equiv_{\beta\eta} B$ iff $\cc{\mathsf t}(A)= \cc{\mathsf t}(B)$.

\begin{definition}
For all $A\in \Nd$, with $\bb{NF}(A)=\forall \vec Y.A_{1}\To  \dots \To  A_{n} \To X $, let 

\adjustbox{scale=0.9, center}{$\cc{\mathsf t}(A)=
\begin{tikzpicture}[baseline=-12]
\node (a) at (0,0) {\small$\Big \{ \mm{\mathsf t}(A_{i})\Big \}_{i=1,\dots,n}$};

\node (b) at (2,0) {\small$\cc X$};

\node(c) at (1, -0.9) {\small${\vec{Y}}$};

\draw (b) to (c);
\draw (a) to (c);

\end{tikzpicture}
%
\qquad\qquad
\mm{\mathsf t}(A)=
\begin{tikzpicture}[baseline=-12]
\node (a) at (0,0) {\small$\Big \{ \cc{\mathsf t}(A_{i})\Big \}_{i=1,\dots,n}$};

\node (b) at (2,0) {\small$\mm X$};

\node(c) at (1, -0.9) {\small${\vec{Y}}$};

\draw (b) to (c);
\draw (a) to (c);

\end{tikzpicture}
$}
\end{definition}

\subsection{Polynomial Trees}

To formulate the schemas $\equiv_{\cc X},\equiv_{\mm X}$ in the language of rooted trees we exploit  an encoding of the  
types of the form $\mu X.\exists \vec Y.\sum_{i\in I}\prod_{j\in J_{i}}\FFun{\cc X}{A_{ij}}$ and $\nu X.\forall \vec Y.\prod_{j\in J}\FFun{\cc X}{A_{j}}$ as
certain special trees employing two new constants $\bullet $ and $\nuzzet$. 
This encoding is easily seen to be a small variant, in the language of finite trees, of the usual 
second-order encodings.


We first introduce a handy notation for ``polynomial'' types, i.e. types corresponding to a generalized sum of generalized products.
 Following \cite{Kock2013}, any such type $A$ is completely determined by a 
 diagram of finite sets $  I\stackrel{f}{\leftarrow} J \stackrel{g}{\to} K$ and a $I$-indexed family of types $(A_{i})_{i \in I}$, so that 
 $A=\sum_{k\in K}\prod_{j\in J_{k}}A_{f(j)}$, where $J_{k} :=g^{-1}(k)$. 
 In the following, we will call the given of a finite diagram $  I\stackrel{f}{\leftarrow} J \stackrel{g}{\to} K$, and an  $I$-indexed family $(a_{i})_{i \in I}$ a \emph{polynomial family}, and indicate it simply as $(a_{jk})_{k\in K, j\in J_{k}}$.
 
 We now enrich the class of $\Nd$-trees as follows:

\begin{definition}[Polynomial trees]

Let $\bullet, \nuzzet$ indicate two new constants. \begin{itemize}
\item the set $\cc{\C P}$ of \emph{positive polynomial trees} is defined by adding to the clauses defining positive $\Nd$-trees two new clauses:

\adjustbox{scale=0.8, center}{$
\begingroup\makeatletter\def\f@size{10}\check@mathfonts
\begin{lrbox}{\mypti}
\begin{varwidth}{\linewidth}
\begin{tikzpicture}
\node(a) at (-0.3,0) {\small$\Big \{\FFun{\cc X\mapsto \cc\bullet}{\cc{\bb E}_{jk}}\Big \}_{j}$};
\node(b) at (1.2,0) {\small$\mm \bullet$};
\node(c) at (0.6,-0.9) {\small$\vec X_{k}$};

\draw[thick] (a) to (c);
\draw[thick] (b) to (c);

\end{tikzpicture}
\end{varwidth}
\end{lrbox}
\begin{tikzpicture}[baseline=-20]

\node(a) at (-0.3,0) {\small$\left\{ \usebox{\mypti} \right\}_{k}$};
\node(b) at (2,-0.6) {\small$\cc \bullet$};
\node(c) at (1.3,-1.2) {\small$\bullet$};

\draw[thick] (0.6,-0.8) to (c);
\draw[thick] (b) to (c);

\end{tikzpicture}
\endgroup
\qquad
\begingroup\makeatletter\def\f@size{10}\check@mathfonts
\begin{lrbox}{\mypti}
\begin{varwidth}{\linewidth}
\begin{tikzpicture}
\node(a) at (-0.3,0) {\small$\Big \{\FFun{\cc X\mapsto \cc\nuzzet}{\cc{\bb E}_{j}}\Big \}_{j}$};
\node(b) at (1.2,0) {\small$\mm \nuzzet$};
\node(c) at (0.6,-0.9) {\small$\vec X$};

\draw[thick] (a) to (c);
\draw[thick] (b) to (c);

\end{tikzpicture}
\end{varwidth}
\end{lrbox}
\begin{tikzpicture}[baseline=-20]

\node(a) at (0,0) {\small$\usebox{\mypti} $};
\node(b) at (2,-0.6) {\small$\cc \nuzzet$};
\node(c) at (1.3,-1.2) {\small$\nuzzet$};

\draw[thick] (0.6,-0.8) to (c);
\draw[thick] (b) to (c);

\end{tikzpicture}
\endgroup
$}
where $X$ is some variable, $(\FFun{\cc X}{\cc{\bb E}_{jk}})_{k,j}$ (resp. $(\FFun{\cc X}{\cc{\bb E}_{j}})_{j}$) is a polynomial family (resp. a family) of positive polynomial trees 
 with no occurrence of $\mm X$, and  $(\vec X_{k})_{k\in K}$ is a $K$-indexed family of finite sets of variables. 

%
%

\item the set $\mm{\C P}$ of \emph{negative polynomial trees} is defined by adding to the clauses defining negative $\Nd$-trees two new clauses:

\adjustbox{scale=0.8, center}{$
\begingroup\makeatletter\def\f@size{10}\check@mathfonts
\begin{lrbox}{\mypti}
\begin{varwidth}{\linewidth}
\begin{tikzpicture}
\node(a) at (-0.3,0) {\small$\Big \{\FFun{\mm X\mapsto \mm\bullet}{\mm{\bb E}_{jk}}\Big \}_{j}$};
\node(b) at (1.2,0) {\small$\cc \bullet$};
\node(c) at (0.6,-0.9) {\small$\vec X_{k}$};

\draw[thick] (a) to (c);
\draw[thick] (b) to (c);

\end{tikzpicture}
\end{varwidth}
\end{lrbox}
\begin{tikzpicture}[baseline=-20]

\node(a) at (-0.3,0) {\small$\left\{ \usebox{\mypti} \right\}_{k}$};
\node(b) at (2,-0.6) {\small$\mm \bullet$};
\node(c) at (1.3,-1.2) {\small$\bullet$};

\draw[thick] (0.6,-0.8) to (c);
\draw[thick] (b) to (c);

\end{tikzpicture}
\endgroup
\qquad
\begingroup\makeatletter\def\f@size{10}\check@mathfonts
\begin{lrbox}{\mypti}
\begin{varwidth}{\linewidth}
\begin{tikzpicture}
\node(a) at (-0.3,0) {\small$\Big \{\FFun{\mm X\mapsto \mm\nuzzet}{\mm{\bb E}_{j}}\Big \}_{j}$};
\node(b) at (1.2,0) {\small$\cc \nuzzet$};
\node(c) at (0.6,-0.9) {\small$\vec X$};

\draw[thick] (a) to (c);
\draw[thick] (b) to (c);

\end{tikzpicture}
\end{varwidth}
\end{lrbox}
\begin{tikzpicture}[baseline=-20]

\node(a) at (0,0) {\small$\usebox{\mypti} $};
\node(b) at (2,-0.6) {\small$\mm \nuzzet$};
\node(c) at (1.3,-1.2) {\small$\nuzzet$};

\draw[thick] (0.6,-0.8) to (c);
\draw[thick] (b) to (c);

\end{tikzpicture}
\endgroup
$}
where $X$ is some variable, $(\FFun{\mm X}{\mm{\bb E}_{jk}})_{k,j}$ (resp. $(\FFun{\mm X}{\mm{\bb E}_{j}})_{j}$) is a polynomial family (resp. a family) of negative polynomial trees 
 with no occurrence of $\cc X$, and  $(\vec X_{k})_{k\in K}$ is a $K$-indexed family of finite sets of variables. 

\end{itemize}

We indicate by $\C P$ the set of all polynomial trees, and by $\C P_{0}$ the set of all polynomial trees with no bound variables, which are called \emph{simple}.
\end{definition}

Any polynomial tree $\bb E\in \C P$ can be converted into a type $\tau(\bb E)$ of $\Ntot$ as illustrated in Fig.~\ref{fig:tau}. It is easily checked that, whenever $\bb E$ is simple, $\tau(\bb E)$ has no quantifier. Moreover, one can check that for all $\Nd$-type,  
$\tau(\cc{\bb t}(A))=A$.


We conclude this section with some basic example of polynomial trees.

\begin{figure}
\adjustbox{center, scale=0.75}{$
\tau(X)=X
 \qquad \qquad
\begingroup\makeatletter\def\f@size{10}\check@mathfonts
\begin{lrbox}{\mypti}
\begin{varwidth}{\linewidth}
\begin{tikzpicture}[baseline=-6]
	\node(a) at (0,0) {\small $\Big \{ \bb E_{i}\Big \}_{i\in I}$};
	\node (b) at (1.4,0) {\small$ \bb F$};
	\node (c) at (0.7,-0.6) {\small$\vec{Y}$};
	\draw (a) to (c);
	\draw (b) to (c);
	\end{tikzpicture} 
\end{varwidth}
\end{lrbox}
\tau\left( \usebox{\mypti}\right) =
\forall \vec{ X}.\tau(\bb E_{1})\To \dots \To \tau(\bb E_{n}) \To \tau(\bb F)
\endgroup
$}

\

\adjustbox{scale=0.75}{$
\begingroup\makeatletter\def\f@size{10}\check@mathfonts
\begin{lrbox}{\mypti}
\begin{varwidth}{\linewidth}
\begin{tikzpicture}
\node(a) at (-0.3,0) {\small$\Big \{\bb E_{jk}[X\mapsto \bb \bullet]\Big \}_{j}$};
\node(b) at (1.2,0) {\small$\bb \bullet$};
\node(c) at (0.6,-0.9) {\small$\vec{Y}_{k}$};

\draw[thick] (a) to (c);
\draw[thick] (b) to (c);

\end{tikzpicture}
\end{varwidth}
\end{lrbox}
\begin{lrbox}{\mypta}
\begin{varwidth}{\linewidth}
\begin{tikzpicture}[baseline=-20]

\node(a) at (0,0) {\small$\left\{ \usebox{\mypti} \right\}_{k}$};
\node(b) at (2,-0.6) {\small$\bb \bullet$};
\node(c) at (1.3,-1.2) {\small$\bullet$};

\draw[thick] (0.6,-0.8) to (c);
\draw[thick] (b) to (c);

\end{tikzpicture}
\end{varwidth}
\end{lrbox}
\tau\left( \usebox{\mypta}\right) =\mu X.\exists \vec{Y}_{k}.
\sum_{k\in K}\prod_{j\in J_{k}}\tau(\bb E_{jk}[X])
\endgroup
\qquad
\begingroup\makeatletter\def\f@size{10}\check@mathfonts
\begin{lrbox}{\mypti}
\begin{varwidth}{\linewidth}
\begin{tikzpicture}
\node(a) at (-0.3,0) {\small$\Big \{\bb E_{j}[X\mapsto \bb \nuzzet]\Big \}_{j}$};
\node(b) at (1.2,0) {\small$\bb \nuzzet$};
\node(c) at (0.6,-0.9) {\small$\vec{Y}$};

\draw[thick] (a) to (c);
\draw[thick] (b) to (c);

\end{tikzpicture}
\end{varwidth}
\end{lrbox}
\begin{lrbox}{\mypta}
\begin{varwidth}{\linewidth}
\begin{tikzpicture}[baseline=-20]

\node(a) at (0,0) {\small$ \usebox{\mypti} $};
\node(b) at (2,-0.6) {\small$\bb \nuzzet$};
\node(c) at (1.3,-1.2) {\small$\nuzzet$};

\draw[thick] (0.6,-0.8) to (c);
\draw[thick] (b) to (c);

\end{tikzpicture}
\end{varwidth}
\end{lrbox}
\tau\left( \usebox{\mypta}\right) =\nu X.\forall \vec{Y}.
\prod_{j\in J}\tau(\bb E_{j}[X])
\endgroup
$}

\caption{\small Translation of simple polynomial trees into monomorphic types.}
\label{fig:tau}
\end{figure}

\begin{example}
\begin{itemize}
\item The constant types $0$ and $1$ are represented as positive/negative trees  by 
$\cc{\B 0}=\begin{tikzpicture}[baseline=-10pt]
\node (a) at (0,0) {\small$\cc \bullet$};
\node(b) at (0,-0.5) {\small$\bullet$};
\draw (a) to (b);
\end{tikzpicture}$, $\mm{\B 0}=\begin{tikzpicture}[baseline=-10pt]
\node (a) at (0,0) {\small$\mm \bullet$};
\node(b) at (0,-0.5) {\small$\bullet$};
\draw (a) to (b);
\end{tikzpicture}$
and
$\cc{\B 1}=\begin{tikzpicture}[baseline=-10pt]
\node (c) at (-0.3,0) {\small$\mm \bullet$};
\node (a) at (0.3,0) {\small$\cc \bullet$};
\node(b) at (0,-0.5) {\small$\bullet$};
\draw (a) to (b);
\draw (c) to (b);

\end{tikzpicture}$, $\mm{\B 1}=\begin{tikzpicture}[baseline=-10pt]
\node (c) at (-0.3,0) {\small$\cc \bullet$};
\node (a) at (0.3,0) {\small$\mm \bullet$};
\node(b) at (0,-0.5) {\small$\bullet$};
\draw (a) to (b);
\draw (c) to (b);

\end{tikzpicture}$.
%
\item The diagram $\{1,2\} \stackrel{1,3\mapsto1; 2\mapsto 2   }{\leftarrow }\{ 1,2,3\} \stackrel{1,2 \mapsto 1; 3\mapsto 2}{\to} \{ 1,2\}$, along with the family $(X_{i})_{i\in \{1,2\}}$,
yields the polynomial family $(\cc{\bb E}_{jk})_{k,j}$, with $\cc{\bb E}_{11}=\cc{\bb E}_{32}=\cc X_{1}$ and $\cc{\bb E}_{21}=\cc X_{2}$, and yields the polynomial tree  
\adjustbox{scale=0.8}{$\begin{tikzpicture}[baseline=-10pt]
\node (a0) at (-0.7,0.1 ) {\small$\cc\bullet$};

\node (a) at (0,0.1 ) {\small$\cc X_{2}$};
\node (aa) at (0.7,0.1) {\small$\mm\bullet$};
\node(aaa) at (0.35,-0.5) {\small$\emptyset$}; 

\node (b) at (1.5,0.1 ) {\small$\cc \bullet$};
\node (bb) at (2.2,0.1) {\small$\mm\bullet$};
\node(bbb) at (1.85,-0.5) {\small$\emptyset$}; 

\node(d) at (3, -0.5) {\small$\cc \bullet$};

\node(c) at (1.85, -0.9) {\small$\bullet$};

\draw (a0) to (aaa);

\draw (a) to (aaa);
\draw (aa) to (aaa);
\draw (b) to (bbb);
\draw (bb) to (bbb);
\draw (a) to (aaa);
\draw (c) to (aaa);
\draw (c) to (bbb);
\draw (c) to (d);
\end{tikzpicture}
$}  encoding the type $\mu X_{1}.(X_{1}\times X_{2})+ X_{1}$.

\item The diagram $\{1,2\} \stackrel{id }{\leftarrow} \{1,2\} \to \{1\}$ with the same family as above
  yields the polynomial tree
%
\adjustbox{scale=0.8}{$\begin{tikzpicture}[baseline=-10pt]
\node (a) at (0,0.1 ) {\small$\cc\nuzzet$};
\node (aa) at (0.7,0.1) {\small$\mm\nuzzet$};
\node(aaa) at (0.35,-0.5) {\small$\emptyset$}; 

\node (b) at (1.5,0.1 ) {\small$\cc X_{2}$};
\node (bb) at (2.2,0.1) {\small$\mm\nuzzet$};
\node(bbb) at (1.85,-0.5) {\small$\emptyset$}; 

\node(d) at (3, -0.5) {\small$\cc \nuzzet$};

\node(c) at (1.85, -0.9) {\small$\nuzzet$};

\draw (a) to (aaa);
\draw (aa) to (aaa);
\draw (b) to (bbb);
\draw (bb) to (bbb);
\draw (a) to (aaa);
\draw (c) to (aaa);
\draw (c) to (bbb);
\draw (c) to (d);
\end{tikzpicture}$} encoding the type $\nu X_{1}.X_{1}\times X_{2}$.

%
%
%
%
%
%
\end{itemize}
\end{example}

\section{Yoneda Reduction and the Characteristic of a Polymorphic Type}\label{secCharacteristic}

In this section we introduce a family of rewriting rules $\leadsto_{\cc X},\leadsto_{\mm X}$ over polynomial trees, that we call \emph{Yoneda reduction}, which 
correspond to the left-to-right orientation of the isomorphisms $\equiv_{\cc X},\equiv_{\mm X}$. 
We will then exploit Yoneda reduction to establish 
 two sufficient conditions to convert a $\Nd$-type $A$ into a quantifier-free type $A'$ such that $A\equiv_{\bb Y}A'$. 

\subsection{Yoneda Reduction}
We will adopt the following conventions:

\begin{notation}
We make the assumption that all bound variables of a polynomial tree $\bb E$ are distinct. More precisely, for any $X\in \BV{\bb E}$, we suppose there exist unique nodes $\bb r_{X}$ and $\bb h_{X}$ such that 
$\bb r_{X}:\vec X$, for some set of variable $\vec X$ such that 
$X\in \vec  X$, and $\bb h_{X}$ is the head of the sub-tree whose root is $\bb r_{X}$.
\end{notation}

We will call two distinct nodes \emph{parallel} if they are immediate successors of the same node, and we let the distance $d(\alpha,\beta)$ between two nodes in a polynomial tree be the number of edges of the unique path from $\alpha$ to $\beta$.

Using polynomial trees we can identify when a quantifier can be eliminated from a type independently from $\beta\eta$-isomorphisms, by inspecting a simple condition on the tree-representation of the type based on the notion of \emph{modular} node, introduced below.

\begin{definition}
%
%
%
%
%
%
%
%
%
%
%
%
%
%
%
%
%
%
%
%
%

For all $X\in \BV{\bb E}$, 
a terminal node $\alpha: X^{c}$ in $\bb E$ is said \emph{modular} if $\alpha\neq \bb h_{X}$, 
$1\leq d(\alpha, \bb r_{X})\leq 2$ 
 and $\alpha$ has no parallel node of label $X^{c}$.
A  pair of nodes of the form $(\alpha:\cc X, \beta:\mm X)$ is called a \emph{$X$-pair},  and 
 a $X$-pair is said \emph{modular} if one of its nodes is modular.

\end{definition}

In the trees in Fig.~\ref{fig:allfig} the modular nodes are underlined and 
the modular pairs are indicated as dashed edges. 

%


\begin{definition}
A variable {$X\in \BV{\bb E}$} is said \emph{eliminable} when every $X$-pair of $\bb E$ is modular.
For every color $c$, we furthermore call $X$ \emph{$c$-eliminable} if every node $\alpha: X^{\OV c}$ is modular.


%
%
%

\end{definition}

We let $\cc{\mathrm{eliminable}}$ be a shorthand for ``${\cc{\mathrm{blue}}}$-eliminable'' and $\mm{\mathrm{eliminable}}$ be a shorthand for ``${\mm{\mathrm{red}}}$-eliminable''. These notions are related as follows:
\begin{lemma}
$X$ is eliminable iff it is either $\cc{\mathrm{eliminable}}$ or $\mm{\mathrm{eliminable}}$.
\end{lemma}
\begin{proof}
If $X$ is neither $\cc{\mathrm{eliminable}}$ nor $\mm{\mathrm{eliminable}}$, then there exist non-modular nodes $\alpha:\cc X, \beta :\mm X$, whence the $X$-pair $(\alpha,\beta)$ is not modular. Conversely, suppose $X$ is eliminable but not $\cc{\mathrm{eliminable}}$. Hence there is a non modular node $\alpha: \cc X$. For all node $\beta:\mm X$, since the $X$-pair $(\alpha,\beta)$ is modular, $\beta$ is modular. We deduce that $X$ is $\mm{\mathrm{eliminable}}$.
\end{proof}


\begin{example}

The variable
$X$ is $\cc{\mathrm{eliminable}}$ but not $\mm{\mathrm{eliminable}}$ in the tree in Fig.~\ref{fig:root1}, and it is both $\mm{\mathrm{eliminable}}$ and $\cc{\mathrm{eliminable}}$ in the tree in Fig.~\ref{fig:root3}. 
\end{example}


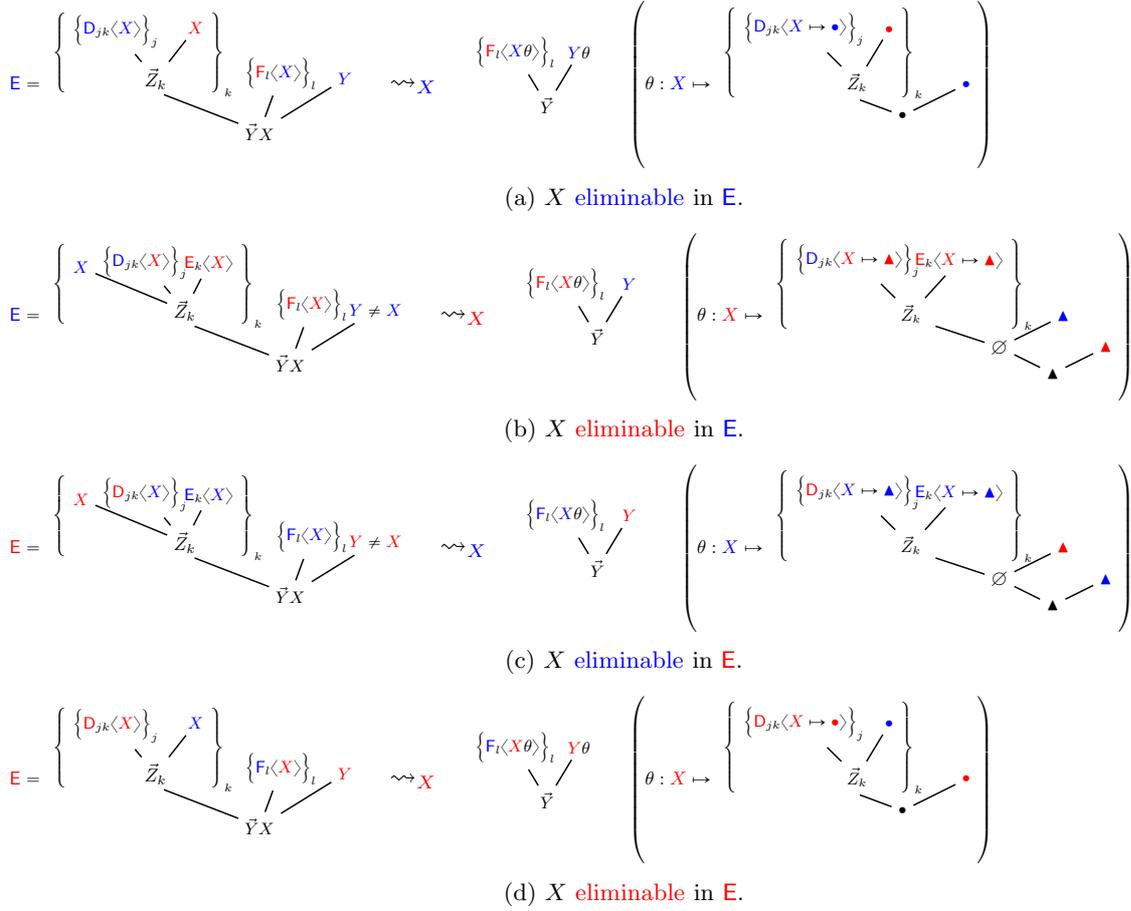
\begin{figure}[t]
\begin{subfigure}{\textwidth}
\leftskip=-0.5cm
\begin{tabular}{c c c c }
\adjustbox{scale=0.7}{$
\begingroup
\makeatletter\def\f@size{10}\check@mathfonts
\begin{lrbox}{\mypti}
\begin{varwidth}{\linewidth}
\begin{tikzpicture}[anchor=base, baseline]

\node (a) at (0,0) {\small$\Big\{\cc{\bb D}_{jk}\langle\cc X\rangle\Big \}_{j}$};
\node (b) at (1.5,0) {\small $\mm X$};

\node(c) at (0.75,-1) {$\vec Z_{k}$};

\draw[thick] (c) to (a);
\draw[thick] (c) to (b);

\end{tikzpicture}
\end{varwidth}
\end{lrbox}
\cc{\bb E} = 
\begin{tikzpicture}[baseline=-20pt]

\node (a) at (0,0) {\small$\left\{ \usebox{\mypti}\right\}_{k}$};

\node (b) at (2.6,-0.4) {\small$\Big\{ \mm{\bb F}_{l}\langle\cc X\rangle\Big \}_{l}$};

\node(c) at (3.8, -0.5) {\small$\cc Y$};

\node(d) at (2.2,-1.5) {\small$\vec YX$};

\draw[thick] (d) to (0.4,-0.8);
\draw[thick] (d) to (b);
\draw[thick] (d) to (c);

\end{tikzpicture}
\endgroup
$}
&
$ \leadsto_{\cc X} $
&
\adjustbox{scale=0.7}{$
\begin{tikzpicture}[baseline=-20pt]
\node (a) at (0,0) {\small$\Big\{\mm{\bb F}_{l}\langle\cc X\theta\rangle\Big\}_{l}$};

\node (b) at (1.2,0) {\small $\cc Y\theta$};

\node(c) at (0.6,-1) {\small$\vec Y$};

\draw[thick] (c) to (b); 
\draw[thick] (c) to (a);

\end{tikzpicture}
$} 
& 
\adjustbox{scale=0.7}{$
\begingroup\makeatletter\def\f@size{10}\check@mathfonts
\begin{lrbox}{\mypti}
\begin{varwidth}{\linewidth}
\begin{tikzpicture}
\node (a) at (-0.4,0) {\small$\Big \{\cc{\bb D}_{jk}\langle\cc X\mapsto \cc \bullet\rangle\Big \}_{j}$};

\node (b) at (1.2,0) {\small $\mm\bullet$};

\node(c) at (0.6,-1) {\small$\vec Z_{k}$};

\draw[thick] (c) to (b); 
\draw[thick] (c) to (a);

\end{tikzpicture}
%
%
\end{varwidth}
\end{lrbox}
\left (
\theta: \cc X \mapsto
\begin{tikzpicture}[baseline=-20]
\node(a) at (-0.3,0) {\small$\left\{ \usebox{\mypti}\right\}_{k}$};
\node(b) at (2.4,-0.6) {\small$\cc \bullet$};

\node(c) at (1.2,-1.2) {\small$\bullet$};

\draw[thick] (c) to (b); 
\draw[thick] (c) to (0.4,-0.9);

\end{tikzpicture}\right )
\endgroup
$}
\end{tabular}

\caption{$ X$ $\cc{\mathrm{eliminable}}$ in $\cc{\bb E}$. }
\label{fig:11}
\end{subfigure} 

\medskip

\begin{subfigure}{\textwidth}
\leftskip=-0.5cm
\begin{tabular}{c c c c }
\adjustbox{scale=0.7}{$
\begingroup\makeatletter\def\f@size{10}\check@mathfonts
\begin{lrbox}{\mypti}
\begin{varwidth}{\linewidth}
\begin{tikzpicture}[anchor=base, baseline]

\node (u) at (-1.2,-0.15) {\small $\cc X$};
\node (a) at (0,0) {\small$\Big\{\cc{\bb D}_{jk}\langle\mm X\rangle\Big \}_{j}$};
\node (b) at (1.2,-0.05) {\small $\mm{\bb E}_{k}\langle\mm X\rangle$};

\node(c) at (0.75,-1) {$\vec Z_{k}$};

\draw[thick] (c) to (a);
\draw[thick] (c) to (b);
\draw[thick] (c) to (u);

\end{tikzpicture}
\end{varwidth}
\end{lrbox}
\cc{\bb E} = 
\begin{tikzpicture}[baseline=-20pt]

\node (a) at (-0.3,0) {\small$\left\{ \usebox{\mypti}\right\}_{k}$};

\node (b) at (2.6,-0.4) {\small$\Big\{ \mm{\bb F}_{l}\langle\mm X\rangle\Big \}_{l}$};

\node(c) at (3.8, -0.5) {\small$\cc Y\neq \cc X$};

\node(d) at (2.2,-1.5) {\small$\vec YX$};

\draw[thick] (d) to (0.4,-0.8);
\draw[thick] (d) to (b);
\draw[thick] (d) to (c);

\end{tikzpicture}
\endgroup
$}
& 
$  \leadsto_{\mm X} $
& 
\adjustbox{scale=0.7}{$
%
\begin{tikzpicture}[baseline=-20pt]
\node (a) at (0,0) {\small$\Big\{\mm{\bb F}_{l}\langle\mm X\theta\rangle\Big\}_{l}$};

\node (b) at (1.2,0) {\small $\cc Y$};

\node(c) at (0.6,-1) {\small$\vec Y$};

\draw[thick] (c) to (b); 
\draw[thick] (c) to (a);

\end{tikzpicture}
$}
&
\adjustbox{scale=0.7}{
$
\begingroup\makeatletter\def\f@size{10}\check@mathfonts
\begin{lrbox}{\mypti}
\begin{varwidth}{\linewidth}
\begin{tikzpicture}
\node (a) at (-0.4,0) {\small$\Big \{\cc{\bb D}_{jk}\langle\mm X\mapsto \mm \nuzzet\rangle\Big \}_{j}$};

\node (b) at (1.5,0) {\small $\mm{\bb E}_{k}\langle\mm X\mapsto \mm \nuzzet\rangle$};

\node(c) at (0.6,-1) {\small$\vec Z_{k}$};

\draw[thick] (c) to (b); 
\draw[thick] (c) to (a);

\end{tikzpicture}
%
%
\end{varwidth}
\end{lrbox}
%
\left(
\theta: \mm X \mapsto
\begin{tikzpicture}[baseline=-20]
\node(a) at (-0.6,0) {\small$\left\{ \usebox{\mypti}\right\}_{k}$};
\node(b) at (2.4,-0.6) {\small$\cc \nuzzet$};

\node(c) at (1.2,-1.2) {\small$\emptyset$};
\node(d) at (3.2,-1.2) {\small$\mm \nuzzet$};
\node(e) at (2.2,-1.7) {\small$ \nuzzet$};

\draw[thick] (c) to (b); 
\draw[thick] (c) to (0,-0.8);
\draw[thick] (c) to (e); 
\draw[thick] (e) to (d); 

\end{tikzpicture}
\right)
\endgroup
$}
\end{tabular}
\caption{$ X$ $\mm{\mathrm{eliminable}}$ in $\cc{\bb E}$. }
\label{fig:12}
\end{subfigure}

\medskip

\begin{subfigure}{\textwidth}
\leftskip=-0.5cm
\begin{tabular}{c c c c }
\adjustbox{scale=0.7}{$
\begingroup\makeatletter\def\f@size{10}\check@mathfonts
\begin{lrbox}{\mypti}
\begin{varwidth}{\linewidth}
\begin{tikzpicture}[anchor=base, baseline]

\node (u) at (-1.2,-0.15) {\small $\mm X$};
\node (a) at (0,0) {\small$\Big\{\mm{\bb D}_{jk}\langle\cc X\rangle\Big \}_{j}$};
\node (b) at (1.2,-0.05) {\small $\cc{\bb E}_{k}\langle\cc X\rangle$};

\node(c) at (0.75,-1) {$\vec Z_{k}$};

\draw[thick] (c) to (a);
\draw[thick] (c) to (b);
\draw[thick] (c) to (u);

\end{tikzpicture}
\end{varwidth}
\end{lrbox}
\mm{\bb E} = 
\begin{tikzpicture}[baseline=-20pt]

\node (a) at (-0.3,0) {\small$\left\{ \usebox{\mypti}\right\}_{k}$};

\node (b) at (2.6,-0.4) {\small$\Big\{ \cc{\bb F}_{l}\langle\cc X\rangle\Big \}_{l}$};

\node(c) at (3.8, -0.5) {\small$\mm  Y\neq \mm X$};

\node(d) at (2.2,-1.5) {\small$\vec YX$};

\draw[thick] (d) to (0.4,-0.8);
\draw[thick] (d) to (b);
\draw[thick] (d) to (c);

\end{tikzpicture}
\endgroup
%
$}
& 
$  \leadsto_{\cc X} $
&  
\adjustbox{scale=0.7}{$
%
\begin{tikzpicture}[baseline=-20pt]
\node (a) at (0,0) {\small$\Big\{\cc{\bb F}_{l}\langle\cc X\theta\rangle\Big\}_{l}$};

\node (b) at (1.2,0) {\small $\mm Y$};

\node(c) at (0.6,-1) {\small$\vec Y$};

\draw[thick] (c) to (b); 
\draw[thick] (c) to (a);

\end{tikzpicture}
$}
& 
\adjustbox{scale=0.7}{
$
\begingroup\makeatletter\def\f@size{10}\check@mathfonts
\begin{lrbox}{\mypti}
\begin{varwidth}{\linewidth}
\begin{tikzpicture}
\node (a) at (-0.4,0) {\small$\Big \{\mm{\bb D}_{jk}\langle\cc X\mapsto \cc \nuzzet\rangle\Big \}_{j}$};

\node (b) at (1.5,0) {\small $\cc{\bb E}_{k}\langle\cc X\mapsto \cc \nuzzet\rangle$};

\node(c) at (0.6,-1) {\small$\vec Z_{k}$};

\draw[thick] (c) to (b); 
\draw[thick] (c) to (a);

\end{tikzpicture}
%
%
\end{varwidth}
\end{lrbox}
%
\left (\theta: \cc X \mapsto
\begin{tikzpicture}[baseline=-20]
\node(a) at (-0.6,0) {\small$\left\{ \usebox{\mypti}\right\}_{k}$};
\node(b) at (2.4,-0.6) {\small$\mm \nuzzet$};

\node(c) at (1.2,-1.2) {\small$\emptyset$};
\node(d) at (3.2,-1.2) {\small$\cc \nuzzet$};
\node(e) at (2.2,-1.7) {\small$ \nuzzet$};

\draw[thick] (c) to (b); 
\draw[thick] (c) to (0,-0.8);
\draw[thick] (c) to (e); 
\draw[thick] (e) to (d); 

\end{tikzpicture}\right)
\endgroup
$}
\end{tabular}
\caption{$ X$ $\cc{\mathrm{eliminable}}$ in $\mm{\bb E}$. }
\label{fig:21}
\end{subfigure} 

\medskip

\begin{subfigure}{\textwidth}
\leftskip=-0.5cm
\begin{tabular}{c c c c }
\adjustbox{scale=0.7}{$
\begingroup\makeatletter\def\f@size{10}\check@mathfonts
\begin{lrbox}{\mypti}
\begin{varwidth}{\linewidth}
\begin{tikzpicture}[anchor=base, baseline]

\node (a) at (0,0) {\small$\Big\{\mm{\bb D}_{jk}\langle\mm X\rangle\Big \}_{j}$};
\node (b) at (1.5,0) {\small $\cc X$};

\node(c) at (0.75,-1) {$\vec Z_{k}$};

\draw[thick] (c) to (a);
\draw[thick] (c) to (b);

\end{tikzpicture}
\end{varwidth}
\end{lrbox}
\mm{\bb E} = 
\begin{tikzpicture}[baseline=-20pt]

\node (a) at (0,0) {\small$\left\{ \usebox{\mypti}\right\}_{k}$};

\node (b) at (2.6,-0.4) {\small$\Big\{ \cc{\bb F}_{l}\langle\mm X\rangle\Big \}_{l}$};

\node(c) at (3.8, -0.5) {\small$\mm Y$};

\node(d) at (2.2,-1.5) {\small$\vec YX$};

\draw[thick] (d) to (0.4,-0.8);
\draw[thick] (d) to (b);
\draw[thick] (d) to (c);

\end{tikzpicture}
\endgroup
$}
& 
$  \leadsto_{\mm X} $
&
\adjustbox{scale=0.7}{$
\begin{tikzpicture}[baseline=-20pt]
\node (a) at (0,0) {\small$\Big\{\cc{\bb F}_{l}\langle\mm X\theta\rangle\Big\}_{l}$};

\node (b) at (1.2,0) {\small $\mm Y\theta$};

\node(c) at (0.6,-1) {\small$\vec Y$};

\draw[thick] (c) to (b); 
\draw[thick] (c) to (a);

\end{tikzpicture}
$}
&
\adjustbox{scale=0.7}{$
\begingroup\makeatletter\def\f@size{10}\check@mathfonts
\begin{lrbox}{\mypti}
\begin{varwidth}{\linewidth}
\begin{tikzpicture}
\node (a) at (-0.4,0) {\small$\Big \{\mm{\bb D}_{jk}\langle\mm X\mapsto \mm \bullet\rangle\Big \}_{j}$};

\node (b) at (1.2,0) {\small $\cc\bullet$};

\node(c) at (0.6,-1) {\small$\vec Z_{k}$};

\draw[thick] (c) to (b); 
\draw[thick] (c) to (a);

\end{tikzpicture}
%
%
\end{varwidth}
\end{lrbox}
\left (\theta: \mm X \mapsto
\begin{tikzpicture}[baseline=-20]
\node(a) at (-0.3,0) {\small$\left\{ \usebox{\mypti}\right\}_{k}$};
\node(b) at (2.4,-0.6) {\small$\mm \bullet$};

\node(c) at (1.2,-1.2) {\small$\bullet$};

\draw[thick] (c) to (b); 
\draw[thick] (c) to (0.4,-0.9);

\end{tikzpicture}\right)
\endgroup
%
$}
\end{tabular}
\caption{$ X$ $\mm{\mathrm{eliminable}}$ in $\mm{\bb E}$. }
\label{fig:22}
\end{subfigure}
\caption{\small Yoneda reduction of $X$-eliminable trees.}
\label{fig:rewriting}
\end{figure}

The proposition below shows that a variable $X$ is $\cc{\mathrm{eliminable}}$ (resp. $\mm{\mathrm{eliminable}}$) in the tree of a $\Nd$-type $\forall X.A$ exactly when $\forall X.A$ matches, up to $\beta\eta$-isomorphisms, with the left-hand type of the schema $\equiv_{\cc X}$  (resp. $\equiv_{\mm X}$).

%


\begin{restatable}{proposition}{normalform}\label{prop:normalform}
\begin{itemize}
\item $X$ is $\cc{\mathrm{eliminable}}$ in $\cc{\bb E}$ iff $\cc{\bb E}$ is as in Fig.~\ref{fig:11} left, 
for some polynomial family $(\cc {\mathsf D}_{jk}\langle\cc X\rangle)_{k\in K, j\in J_{k}}$, family $(\mm{\bb F}_{l}\langle\cc X\rangle)_{l\in L}$ and blue variable $\cc Y$ (possibly $\cc X$ itself).

\item $X$ is $\mm{\mathrm{eliminable}}$ in $\cc{\bb E}$ iff $\cc{\bb E}$ is as in Fig.~\ref{fig:12} left, for some polynomial family $(\cc {\mathsf D}_{jk}\langle\mm X\rangle)_{k\in K, j\in J_{k}}$, family $(\mm{\bb E}_{k}\langle\mm X\rangle)_{k\in K} $, family 
$(\mm{\bb F}_{l}\langle\mm X\rangle)_{l\in L}$ and blue variable $\cc Y\neq \cc X$. 

\item $X$ is $\cc{\mathrm{eliminable}}$ in $\mm{\bb E}$ iff $\mm{\bb E}$ is as in Fig.~\ref{fig:21} left, for some polynomial family $(\mm {\mathsf D}_{jk}\langle\cc X\rangle)_{k\in K,j\in J_{k}}$, family $(\cc{\bb E}_{k}\langle\cc X\rangle)_{k\in K}$, family
$(\cc{\bb F}_{l}\langle\cc X\rangle)_{l\in L}$ and red variable $\mm Y\neq \mm X$.
 
\item $X$ is $\mm{\mathrm{eliminable}}$ in $\mm{\bb E}$  iff $\mm{\bb E}$ is as in Fig.~\ref{fig:22} left, for some
 polynomial family $(\mm {\mathsf D}_{jk}\langle\mm X\rangle)_{k\in K, j\in J_{k}}$, family $(\cc{\bb F}_{l}\langle\mm X\rangle)_{l\in L}$ and red variable $\mm Y$ (possibly $\mm X$ itself).
 
\end{itemize}
%
\end{restatable}
\begin{proof}
If $ X$ is $\cc{\mathrm{eliminable}}$ in a positive tree $\cc{\mathsf E}=$ 
\adjustbox{scale=0.7}{\begin{tikzpicture}[baseline=-20]
	\node(a) at (0,0) {\small $\Big \{ \mm{\bb E}_{i}\Big \}_{i\in I}$};
	\node (b) at (1.4,0) {\small$ \cc X$};
	\node (c) at (0.7,-0.8) {\small${ \vec{ Y}} X$};
	\draw (a) to (c);
	\draw (b) to (c);
	\end{tikzpicture} }, 
	then let us split the negative subtrees $\mm{\bb E}_{i}$ into those which do not contain terminal nodes  with label $\mm X$, that we indicate as $\FFun{\cc X}{\mm{\mathsf F}_{l}}$, and those which contain a (necessarily unique) modular node of label $\mm X$, that we indicate
as $\mm{\bb G}_{k}$. Since $\mm{\bb G}_{k}$ contains a unique modular node $\alpha_{k}:\mm X$ which is no more than two edges far away from the root of $\cc{\mathsf E}$, then $\mm {\bb G}_{k}$ is of the form 	\adjustbox{scale=0.8}{\begin{tikzpicture}[baseline=-20]
	\node(a) at (0,0) {\small $\Big \{ \FFun{\cc X}{\cc{\mathsf D}_{jk}}\Big \}_{i\in I}$};
	\node (b) at (1.4,0) {\small$ \mm X$};
	\node (c) at (0.7,-0.8) {\small${\vec{  Z}_{k}}$};
	\draw (a) to (c);
	\draw (b) to (c);
	\end{tikzpicture} }, where possibly $I= Z_{k}=\emptyset$, (so that $\mm{\bb G}_{k}=\mm X$).

If $X$ is $\mm{\mathrm{eliminable}}$ in a positive tree $\cc{\mathsf E}=$ 
\adjustbox{scale=0.7}{\begin{tikzpicture}[baseline=-20]
	\node(a) at (0,0) {\small $\Big \{ \mm{\bb E}_{i}\Big \}_{i\in I}$};
	\node (b) at (1.4,0) {\small$ \cc X$};
	\node (c) at (0.7,-0.8) {\small$\vec{ Y}X$};
	\draw (a) to (c);
	\draw (b) to (c);
	\end{tikzpicture} }, 
	then, since all nodes $\alpha:\cc X$ are modular, it must be $\cc Y\neq \cc X$. Let us split the positive subtrees $\cc{\mathsf E}_{i}$ into those which do not contain terminal nodes  with label $\cc X$, that we indicate as $\FFun{\mm X}{\mm{\mathsf F}_{l}}$, and those which contain a (necessarily unique) modular node of label $\cc X$, that we indicate
as $\cc G_{k}$. Since $\cc G_{k}$ contains a unique modular node $\alpha_{k}:\cc \alpha$ which is no more than two edges far away from the root of $\cc{\mathsf E}$, then $\cc G_{k}$ is of the form 	\adjustbox{scale=0.8}
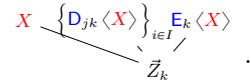
{\begin{tikzpicture}[baseline=-20]
	\node(a0) at (-1.5,0) {\small$\mm X$};
	\node(a) at (0,0) {\small $\Big \{ \FFun{\mm X}{\cc{\mathsf D}_{jk}}\Big \}_{i\in I}$};
	\node (b) at (1.4,0) {\small$ \FFun{\mm X}{\cc{\mathsf E}_{k}}$};
	\node (c) at (0.7,-0.8) {\small${\vec{{ Z}}_{k}}$};
	\draw (a) to (c);
	\draw (b) to (c);
		\draw (a0) to (c);
	\end{tikzpicture} }.
	
	The two cases with a negative tree are described in a dual way (i.e. by switching blue and red everywhere).
\end{proof}

%
%
%

For all four cases of Prop. \ref{prop:normalform} we define a rewriting rule which eliminates $X$.
\begin{definition}[Yoneda reduction]
Let $\bb F\in \C P$ and $X\in \BV{\bb F}$ be an eliminable variable. The rules $\bb F
\leadsto_{X^{c}} \bb F'$ 
consist in replacing the subtree $\bb E$ of $\bb F$ rooted in $\bb r_{X}$ as   
 illustrated in Fig.~\ref{fig:rewriting}. 

\end{definition}

\begin{example}
The tree in Fig.~\ref{fig:root3} rewrites as illustrated in  Fig.~\ref{fig:root4}.
\end{example}

By inspecting the rules in Fig.~\ref{fig:rewriting} one can check that $\bb E\leadsto_{X^{c}}\bb E'$ implies $\tau(\bb E)\equiv_{\bb Y}\tau(\bb E')$. From this we can deduce by induction: 
\begin{lemma}\label{lemma:converge}
For all $\Nd$-type $A$, if $\cc{\bb t}(A) \leadsto^{*} \bb E \in \C P_{0}$, then $A\equiv_{\bb Y}\tau(\bb E)\in \NImu$.

\end{lemma}
The lemma above suggests to study the elimination of quantifiers from $\Nd$-types by studying the convergence of $\Nd$-trees onto simple polynomial trees. This will be our next goal.

\begin{figure}

\begin{subfigure}{0.5\textwidth}
\adjustbox{scale=0.75, center}{$
\bb E_{1}=
\begin{tikzpicture}[baseline=-20]

\node (a) at (0,0) {\small$\U{\cc X}$};
\node (b) at (1,0) {\small$\U{\mm X}$};
\node (c) at (0.5,-0.7) {\small$\emptyset$};

\node (aa) at (2,0) {\small$\U{\cc X}$};
\node (bb) at (3,0) {\small${\mm Y}$};
\node (cc) at (2.5,-0.7) {\small$ Y$};

\node(d) at (4,-0.7) {\small$\cc X$};

\node(r) at (2.5,-1.4) {\small$ X$};

\draw (a) to (c);
\draw (b) to (c);
\draw (aa) to (cc);
\draw (bb) to (cc);
\draw (r) to (c);
\draw (r) to (d);
\draw (r) to (cc);

\draw[dashed, bend right=65] (d) to (b);
\draw[dashed, bend right=45] (aa) to (b);
\draw[dashed, bend left=40] (a) to (b);

\end{tikzpicture}
%
%
%
%
%
%
%
%
%
%
%
%
%
%
%
%
%
%
%
%
$}
\caption{\adjustbox{scale=0.9}{ $\forall X. (\cc X \To \mm X)\To (\forall Y.\cc X \To \mm Y)\To \cc X$.}}
\label{fig:root1}
\end{subfigure}
\begin{subfigure}{0.48\textwidth}
\adjustbox{scale=0.75, center}{$
\bb E_{2}=
\begin{tikzpicture}[baseline=-20]

\node (a) at (0,0) {\small$\U{\cc X}$};
\node (aa) at (1,0) {\small$\U{\mm X}$};
\node (u) at (0.5,-0.7) {\small$\emptyset$};

\draw (a) to (u);
\draw (aa) to (u);

\node (b) at (1.7,0) {\small$\U{\cc X}$};
\node (bb) at (2.7,0) {\small${\mm Z}$};
\node (uu) at (2.2,-0.7) {\small$Z$};

\draw (b) to (uu);
\draw (bb) to (uu);

\node (cb) at (3.5,0) {\small$\U{\cc Y}$};
\node (cbb) at (4.5,0) {\small$\U{\mm X}$};
\node (c) at (4,-0.7) {\small$\emptyset$};

\draw (cb) to (c);
\draw (cbb) to (c);

\node (d) at (5,-0.7) {\small$\cc Y$};

\node(t) at (3,-1.4) {\small$ XY$};

\draw (c) to (t);
\draw (d) to (t);
\draw (u) to (t);
\draw (uu) to (t);

\draw[dashed, bend right=35] (cbb) to (b);
\draw[dashed, bend right=35] (cbb) to (a);
\draw[dashed, bend right=55] (aa) to (a);
\draw[dashed, bend left=55] (aa) to (b);


\end{tikzpicture}
$}
\caption{\adjustbox{scale=0.9}{$\forall XY. (\cc X \To\mm X) \To (\forall Z.\cc X\To \mm Z)\To (\cc Y\To \mm X)\To \cc Y  $.}}
\label{fig:root3}
\end{subfigure}
\caption{Polynomial trees with highlighted modular nodes and modular pairs.}
\label{fig:allfig}
\end{figure}
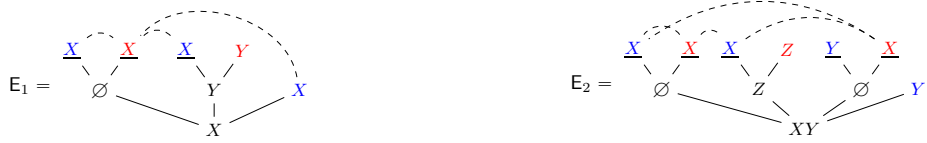
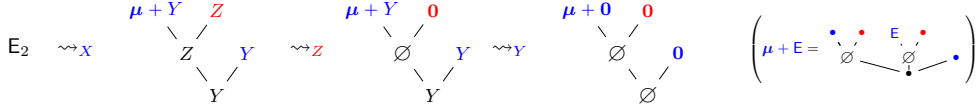
\begin{figure}
\adjustbox{scale=0.8}{$
\bb E_{2} 
\quad \leadsto_{\cc X}\quad
\begin{tikzpicture}[baseline=-20]

\node(m1) at (0,0) {\small$\cc{\boldsymbol\mu +\cc Y}$};

\node(y) at (1,0 ) {\small${\mm Z}$};

\node(yy) at (0.5,-0.7) {\small$Z$};

\draw (yy) to (m1);
\draw (yy) to (y);

\node(l) at (1.5, -0.7) {\small$\cc{Y}$};

\node(b) at (1,-1.4) {\small$Y$};

\draw (b) to (l);
\draw (b) to (yy);

\end{tikzpicture}
\quad \leadsto_{\mm Z}\ 
\begin{tikzpicture}[baseline=-20]

\node(m1) at (0,0) {\small$\cc{\boldsymbol\mu +\cc Y}$};

\node(y) at (1,0 ) {\small$\B{\mm 0}$};

\node(yy) at (0.5,-0.7) {\small$\emptyset$};

\draw (yy) to (m1);
\draw (yy) to (y);

\node(l) at (1.5, -0.7) {\small$\cc{Y}$};

\node(b) at (1,-1.4) {\small$Y$};

\draw (b) to (l);
\draw (b) to (yy);

\end{tikzpicture}
 \ \leadsto_{\cc Y}\quad
\begin{tikzpicture}[baseline=-20]

\node(m1) at (0,0) {\small$\cc{\boldsymbol\mu + \B 0}$};

\node(y) at (1,0 ) {\small$\mm{\B 0} $};

\node(yy) at (0.5,-0.7) {\small$\emptyset$};

\draw (yy) to (m1);
\draw (yy) to (y);

\node(l) at (1.5, -0.7) {\small$\cc{\B 0}$};

\node(b) at (1,-1.4) {\small$\emptyset$};

\draw (b) to (l);
\draw (b) to (yy);

\end{tikzpicture}
$}
\adjustbox{scale=0.6}{$
\quad\quad \quad
\left ( 
\cc{\boldsymbol \mu + \bb E}=\begin{tikzpicture}[scale=0.9,baseline=-10pt]
\node (a) at (0,0.1 ) {\small$\cc\bullet$};
\node (aa) at (0.7,0.1) {\small$\mm\bullet$};
\node(aaa) at (0.35,-0.5) {\small$\emptyset$}; 

\node (b) at (1.5,0.1 ) {\small$\cc{\bb E}$};
\node (bb) at (2.2,0.1) {\small$\mm\bullet$};
\node(bbb) at (1.85,-0.5) {\small$\emptyset$}; 

\node(d) at (3, -0.5) {\small$\cc \bullet$};

\node(c) at (1.85, -0.9) {\small$\bullet$};

\draw (a) to (aaa);
\draw (aa) to (aaa);
\draw (b) to (bbb);
\draw (bb) to (bbb);
\draw (a) to (aaa);
\draw (c) to (aaa);
\draw (c) to (bbb);
\draw (c) to (d);
\end{tikzpicture}
\right)
$}
\caption{Reduction of $\bb E_{2}$.}
\label{fig:root4}
\end{figure}


%


%
%
%

\subsection{The Coherence Condition}

When a reduction is applied to $\bb E$, several sub-trees of $\bb E$ can be either erased or copied and moved elsewhere. Hence, the resulting tree $\bb E'$ might well have a greater size and even a larger number of bound variables than $\bb E$. Nevertheless, sequences of Yoneda reductions always terminate, as one can define a measure which decreases at each step.

\begin{restatable}{proposition}{termination}\label{prop:termination}
There is no infinite sequence of Yoneda reductions.
\end{restatable}
\begin{proof}
Let $K$ be the number of terminal nodes of $\bb E$ and, for all bound variable $X_{i}$ of $\bb E$, let $K_{i}$ be the number of terminal nodes of the subtree $\bb E_{i}$ of $\bb E$ with root $\bb r_{X_{i}}$.
Fix some $s> K^{2}$ and let $\mu(\bb E)$ be the sum of all $s^{K_{i}}$, for all bound variable $X_{i}$.
If $\bb E \leadsto_{X_{i}^{c}} \bb E'$, then the elimination of $X_{i}$ comes at the cost of possibly duplicating the bound variables which are at distance at least one from the root of $\bb E_{i}$. These are no more than $K_{i}$ and can be duplicated for no more than $K_{i}$ times. Since for each such variable $j$, it must be $K_{j}< K_{i}$, we deduce that $\mu(\bb E')$ is obtained by replacing  in $\mu(\bb E)$ the summand $s^{K_{i}}$ by a sum of no more than $K_{i}^{2}$ summands which are bounded by $s^{K_{i}-1}$. Hence, if we write $\mu(\bb E)$ as $ a+ s^{K_{i}}$, then $\mu(\bb E')  \leq a+ K_{i}^{2}s^{K_{i}-1} < a + s^{K_{i}}=\mu(\bb E)$. We deduce that the measure $\mu(\bb E)$ strictly decreases under reduction, and so, if an infinite sequence of reductions exists, then after a large enough number of step one must obtain a tree $\bb E'$ such that  $\mu(\bb E')<0$, which is impossible.
\end{proof}

Although sequences of reductions always terminate, they need not terminate on a simple polynomial tree, that is, on the encoding of a monomorphic type. 
This can be due to several reasons. Firstly, one bound variable might not be eliminable. Secondly, even if all variables are eliminable, this property need not be preserved by reduction.
For example, take the type $A=\forall X.\forall Y. (\cc X\To \cc Y\To \mm X)\To ((\mm Y\To \cc X)\To \mm W)\To \cc Z$: although $X$ and $Y$ are both eliminable (in the associated tree), if we apply a reduction to $X$, then $Y$ ceases to be eliminable, and similarly if we reduce $Y$ first.
%
%
%
Such conflicts can be controlled by imposing a suitable \emph{coherence} relation on variables.
%
%

\begin{definition}
Let $\bb E$ be a polynomial tree, $X,Y\in \BV{\mathsf E}$ and $c,d\in  \mathsf{Colors}$. $X^{c}$ and $Y^{d}$ are said \emph{coherent} if there exists no parallel modular nodes of the form $\alpha:X^{\OV c}, \beta: Y^{\OV d}$ in $\bb E$. 


\end{definition}

\begin{example}
In the tree in Fig.~\ref{fig:root3}, $\mm Y$ and $\cc Z$ are coherent, while $\cc X$ and $\mm Y$ are not.

\end{example}



\begin{definition}[coherence condition]
Let $\mathsf E$ be a polynomial tree. A \emph{valuation of $\mathsf E$} is any map $\phi: \BV{\mathsf E} \to\mathsf{Colors}$. For all valuation $\phi$ of $\mathsf E$, we call $\mathsf E$ \emph{$\phi$-coherent} if for all $X\in \BV{\mathsf E}$, $X$ is $\phi(X)$-eliminable, and moreover for all $Y\neq X\in \BV{\mathsf E}$, $X^{\phi(X)}$ is coherent with $Y^{\phi(Y)}$. We call $\mathsf E$ \emph{coherent} if it is $\phi$-coherent for some valuation $\phi$ of $\mathsf E$.
%
%
\end{definition}

\begin{remark}[Coherence is an instance of 2-SAT]\label{rem:2sat}
The problem of checking if a polynomial tree $\bb E$ is coherent can be formulated as an instance of 2-SAT (a well-known polytime problem): consider $n$ Boolean variables $x_{1},\dots, x_{n}$ (one for each bound variable of $\bb E$), and let $a_{i}^{c}$ be $x_{i}$ if $c=\cc{\mathrm{blue}}$ and $\lnot x_{i}$ if $c=\mm{\mathrm{red}}$.
Consider then the 2-CNF $A\land B$, where $A$ is the conjunction of all $a_{i}^{c}\vee a_{i}^{c}$ such that $
 X_{i}$ is not $\OV c$-eliminable in $\bb E$, and $B$ is the conjunction of all $  a_{i}^{c}\lor a_{j}^{d}$, for all incoherent $X_{i}^{\OV c}$ and $X_{j}^{\OV d}$.
Then a coherent valuation of $\bb E$ is the same as a model of $A\land B$.

\end{remark}


As observed before, a reduction $\bb E \leadsto_{X}\bb E'$ might copy or erase some bound variables of $\bb E$. One can then define a map $g:\BV{\bb E'}\to \BV{\bb E}$ associating any variable in $\bb E'$ with the corresponding variable in $\bb E$ of which it is a copy. A sequence of reductions $\bb E_{0}\leadsto_{X_{1}^{c_{1}}}\dots\leadsto_{X_{n}^{c_{n}}}\bb E_{n}$ induces then, for $1\leq  i\leq  n$, maps 
$g_{i}: \BV{\bb E_{i}}\to \BV{ \bb E_{i-1}}$, and we let $G_{i}: \BV{\bb E_{i}}\to \BV{ \bb E_{0}}$ be $g_{1}\circ g_{2}\circ\dots \circ  g_{i}$.

The rewriting properties of coherent trees are captured by the following notion:

\begin{definition}[standard reduction]
A sequence of reductions $\bb E_{0}\leadsto_{X_{1}^{c_{1}}}\dots\leadsto_{X_{n}^{c_{n}}}\bb E_{n}$ is said \emph{standard} if for all $i,j=1,\dots,n$, $G_{i}(X_{i})^{c_{i}}$ is coherent with $G_{j}(X_{j})^{c_{j}}$ in $\bb E_{0}$.
We say that $\bb E$ \emph{strongly converges under standard reduction} if all standard reductions starting from $\bb E$ terminate on a simple polynomial tree.

\end{definition}

We indicate a standard reduction as $\bb E\stand^{*} \bb E'$.
We will show that standard reduction preserves coherence.

\begin{lemma}\label{lemma:solvable}
If $\bb E$ is coherent, and $\bb E\stand^{*} \bb E'$, then $\bb E'$ is coherent.
\end{lemma}
\begin{proof}

For any node $\alpha$ of $\bb E'$ there exists a unique node, that we indicate as $g(\alpha)$, from which 
 $\alpha$ ``comes from'' through the reduction step. Moreover, for each bound variable $Y\in \BV{\bb E'}$ there exists a unique bound variable $g(Y)\in \BV{\bb E}$ from which $Y$ ``comes from'' through the reduction step.

We indicate as $\alpha\preceq \beta$ the tree-order relation on the nodes of $\bb E$ (so that the root of $\bb E$ is the minimum node). 
Suppose $\bb E$ is $\phi$-coherent and $\bb E\leadsto_{X^{\phi(X)}} \bb E'$. 
We will show that $\bb E'$ is
 $\phi\circ g$-coherent.
 
Let $\mathsf{F}$ be the sub-tree of $\bb E$ with root $\bb r_{X}$, fix a variable $Y$ in $\BV{\bb E'} $.
 We must consider four possible cases, corresponding to the four reduction rules. Each rule requires itself four subcases. We only consider the case of the first rule (hence with $\bb E$ positive), as the other three cases can be proved in a similar way.
\begin{description}
\item[($\cc{\mathsf{F}}\leadsto_{\cc X}\cc{\mathsf F'}$)]   We first show that $Y$ is $\phi(g(Y))$-eliminable in $\cc{\mathsf{F}}'$, by showing that any modular node $\alpha: g(Y)$ in $\cc{\mathsf{F}}$ is transformed into a modular node $\beta:Y$ in $\cc{\mathsf{F}}'$. We argue by considering the relative position of $\bb r_{X}$ and $\bb r_{Z}$ in $\cc{\mathsf E}$:
	\begin{description}
	\item[($\bb r_{X}=\bb r_{g(Y)}$)] Since $\cc X$ and $g(Y)^{\phi(g(Y))}$ are coherent in $\cc{\mathsf E}$, and all nodes $\alpha:g(Y)^{\OV{\phi(g(Y))}}$ are modular, each such node $\alpha$ must be a terminal node of some of the $\FFun{\cc X}{\mm{\mathsf F}_{l}}$
	(see Fig. \ref{fig:11}). In fact, if $\alpha$ were a terminal node of some $\FFun{\cc X}{\cc{\mathsf D}_{jk}}$, then either $\FFun{\cc X}{\cc{\mathsf D}_{jk}}=g(Y)^{\OV{\phi(g(Y))}}$, which contradicts the assumption that $\cc X$ and $g(Y)^{\phi(g(Y))}$ are coherent, or its distance from $\bb r_{X}=\bb r_{g(Y)}$ is greater than 2, against the assumption that $\alpha$ is modular. 
	Now, since all nodes in $\FFun{\cc X}{\mm{\mathsf F}_{l}}$ which are not labeled with $\cc X$ are not ``moved'' by the reduction rule (see Fig. \ref{fig:11}), and no new node labeled $g(Y)^{\phi(g(Y))}$ is created by $\theta$ (since otherwise it should come from a node with same label in some $\FFun{\cc X}{\cc{\mathsf D}_{jk}}$), all nodes $\alpha:Y^{\OV{\phi(g(Y))}}$ in $\cc{\mathsf F'}$ are modular, whence $Y$ is $\phi(g(Y))$-eliminable in $\mathsf{F}'$.

%
%

%
%

	\item[($\bb r_{X}\prec\bb r_{g(Y)}$)] First, if $\bb r_{g(Y)}$ occurs within either the $\FFun{\cc X}{\cc{\mathsf D}_{jk}}$ or the $\FFun{\cc X}{\mm{\mathsf F}_{l}}$, then one can check that, by the transformation from $\cc{\mathsf{F}}$ to $\cc{\mathsf F'}$, the relative positions between the nodes $g(\alpha): g(Y)^{c}$ and $\bb r_{g(Y)}$ in $\bb E$ coincide with the relative positions between the nodes $\alpha$ and $\bb r_{Y}$, and thus all modular nodes $\delta:g(Y)^{\OV{\phi(g(Y))}}$ in $\cc{\mathsf{F}}$ yield modular nodes $\delta':Y^{\OV{\phi(g(Y))}}$ in $\cc{\mathsf{F}}'$.

	Suppose now $g(Y)$ is among some set $\vec Z_{k}$ (see Fig. \ref{fig:11}). One can check that for all node $\delta:g(Y)^{\OV{\phi(g(Y))}}$ occurring in $\mathsf{F}$ within one of the $\FFun{\cc X}{\cc{D}_{jk}}$, if, in the transformation of $\cc{\mathsf{F}}$ in $\cc{\mathsf{F}}'$, $\delta$ is transformed into $\delta'$, then $d(\delta, \bb r_{g(Y)})$ in $\cc{\mathsf E}$ is the same as $d(\delta', \bb r_{Y})$ in $\cc{\mathsf E}'$. Moreover, it is clear that $\delta'$ cannot be the head of the sub-tree rooted in $\bb r_{Y}$, and that if $\delta$ has no parallel node with same label, so does $\delta'$. Hence, if $\delta$ is modular, so is $\delta'$.
	

	\item[($\bb r_{g(Y)}\prec\bb r_{X}$)] 
	If a modular node $\delta: g(Y)^{\OV{\phi(g(Y))}}$ in $\cc{\mathsf E}$ is transformed into some node $\delta': Y^{\OV{\phi(Y)}}$ in $\cc{\mathsf E}'$, then $\delta'$ is modular too. In fact, no such $\delta$ can be actually moved by the reduction rule, since the only possibilities are that 
		either $\delta\prec \bb r_{X}$ or $\delta$ coincides with one of the $\FFun{\cc X}{\mm{\mathsf F}_{l}}$.
		Since neither the node $\bb r_{Y}$ is moved, we deduce that $\delta'$ has the same distance from $\bb r_{g(Y)}$ than $\delta$ has from $\bb r_{g(Y)}$, that it has no parallel node with same label and cannot be the head of the sub-tree rooted in $\bb r_{Y}$, whence $\delta'$ is modular.

	\item[($\bb r_{g(Y)}\not\preceq\bb r_{X},\bb r_{X}\not\preceq\bb r_{g(Y)}$]
	We can argue as in the case above, since the transformation from $\cc{\mathsf{F}}$ to $\cc{\mathsf F'}$ cannot move any node  $\delta:g(Y)^{\OV{\phi(g(Y))}}$.
	
	\end{description}

We conclude then that, for all $Y\in \BV{\bb E'}$, $Y$ is $\phi(g(Y))$-eliminable.
It remains to check that for all $Y,Z\in \BV{\bb E'}$, $Y^{\phi(g(Y))}$ and $Z^{\phi(g(Y))}$ are coherent. It can be checked that for any parallel nodes $\alpha:Y^{\OV{\phi(g(Y))}}, \beta: Z^{\OV{\phi(g(Z))}}$ in $\bb E'$, $g(\alpha)$ and $g(\beta)$ are parallel in $\bb E$. Hence, if  $Y^{\phi(g(Y))}$ and $Z^{\phi(g(Y))}$  are not coherent, then $g(Y)^{\phi(Y)}$ and $g(Z)^{\phi(Z)}$ are not coherent, against the assumption.
%
%
%
%
%
%
%
%
\end{description}
\end{proof}

Let us call a $c$-eliminable variable $ X$ \emph{$c$-canceling} in the following cases:
\begin{itemize}
\item $\bb r_{X}$ is the root of a positive tree as in Fig. \ref{fig:11} left, where $c=\cc{\mathrm{blue}}$, $Y\neq X$ and no node in any sub-tree $\FFun{\cc X}{\mm F_{l}}$ has label $\cc X$;
\item $\bb r_{X}$ is the root of a positive tree as in Fig. \ref{fig:12} left, $c=\mm{\mathrm{red}}$ and no node in any sub-tree $\FFun{\mm X}{\mm F_{l}}$ has label $\mm X$;
\item $\bb r_{X}$ is the root of a negative tree as in Fig. \ref{fig:21} left, $c=\cc{\mathrm{blue}}$ and no node in any sub-tree $\FFun{\cc X}{\cc F_{l}}$ has label $\cc X$;
\item $\bb r_{X}$ is the root of a negative tree as in Fig. \ref{fig:22}, $c=\mm{\mathrm{red}}$, $Y\neq X$ and no node in any sub-tree $\FFun{\mm X}{\cc F_{l}}$ has label $\mm X$.
\end{itemize}

In fact, $X$ is $c$-canceling precisely when by applying a $X^{c}$-reduction step, the action of the substitution $\theta$ from Fig. \ref{fig:rewriting} is empty, so that the whole left-hand part of the tree, including some of its bound variables, is erased. 
Observe that, if $X$ is not $c$-canceling, and $\bb E\leadsto_{X^{c}}\bb E'$, then the application $ Y\mapsto g( Y)$ from a bound variable of $\bb E'$ to a bound variables of $\bb E$ from which $Y$ ``comes from'' is a \emph{surjective} function $g: \BV{\bb E'}\to \BV{\bb E}-\{X\}$.

Moreover, let us call a standard reduction $\bb E\stand^{*} \bb E'$ \emph{least-canceling} if for all $c$-canceling variable $X$, a reduction $\leadsto_{X^{c}}$ is performed only after all bound variables $ Y$ occurring above the node $\bb r_{X}$ are eliminated. 
Let us indicate by $\bb E\leadsto_{\bb{canc}}\bb E'$ that $\bb E$ reduces to $\bb E'$ by a sequence of least-canceling standard reductions.

Using Lemma \ref{lemma:solvable} we can prove: 

\begin{restatable}{theorem}{solvability}\label{thm:solvability}
$\bb E$ is coherent iff $\bb E$ strongly converges under standard reduction.
\end{restatable}
\begin{proof}

Suppose $\bb E$ is coherent. From Lemma \ref{lemma:solvable} we deduce that if $\bb E\stand^{*} \bb E'$, then either $\bb E'\in \C P_{0}$ or $\bb E'\stand \bb E''$ for some $\bb E''$. Since sequences of standard reductions terminate, we deduce that $\bb E$ must in the end converge onto some simple tree.

Conversely, suppose $\bb E$ strongly converges under $\stand$. Since all standard reductions converge, there must be a converging least-canceling standard reduction $\bb E=\bb E_{0}\leadsto_{Y_{1}^{c_{1}}} \bb E_{1}\leadsto_{Y_{2}^{c_{2}}} \dots \leadsto_{Y_{n}^{c_{n}}}\bb E_{n}\in \C P_{0}$. 
Let $\delta(X)= \min\{i \mid G_{i}(Y_{i})= X\}$ and $\phi(X)=c_{\delta(X)}$. 
Using the fact that the variables $G_{i}(Y_{i})^{c_{i}}$ are pairwise coherent and include all bound variables of $\bb E$, one can check that $\delta$ and $\phi$ are well-defined and $\bb E$ is $\phi$-coherent.
%
\end{proof}

Using Lemma \ref{lemma:converge} we further deduce:
\begin{corollary}
Let $A$ be a type of $\Nd$. If $\cc{\bb t}(A)$ is coherent, then there exists a $\NImu$-type $A'$ such that 
$A\equiv_{\bb Y}A'$.
\end{corollary}
%
%

%
 
%


%

\subsection{The Characteristic}

We now introduce a refined condition for coherent trees, which can be used to predict whether a type rewrites into a finite type (i.e. one made up from $0,1,+,\times, \To$ only) or into one using $\mu,\nu$-constructors. 

An intuition from Section 2 is that, for a type of the form $\forall X. (\FFun{\cc X}{A}\To \mm X)\To \FFun{\cc X}{B}$  (which  rewrites into $\FFun{\cc X\mapsto \mu X.\FFun{\cc X}{A}}{B}$)
to reduce to one without $\mu,\nu$-types, the variable $X$ \emph{must not occur in $A$} at all. However, the property ``$X$ does not occur in $A$'' need not be preserved under reduction. Instead, we will define a stronger condition that is preserved by standard reduction by inspecting a class of \emph{paths} in the tree of a type.


\begin{definition}
Let $\bb E$ be a polynomial $\Nd$-tree and let $\preceq$ indicate the natural order on the nodes of $\bb E$ having the root of $\bb E$ as its minimum.
A \emph{down-move} in $\bb E$ is a pair $\alpha\step\beta$, such that, for some bound variable $X\in \BV{\bb E}$, $(\alpha,\beta)$ is an $X$-pair and $\beta$ is modular. 
An \emph{up-move} in $\bb E$ is a pair $\alpha\step\beta $ such that $\alpha\neq \beta$, $\alpha$ is a modular node with immediate predecessor $\gamma$, $\beta:X$ for some $X\in \BV{\bb E}$, and  $\gamma\preceq \beta$. 
An \emph{alternating path} in $\bb E$ is a sequence of nodes $\alpha_{0}\dots \alpha_{2n}$ such that $\alpha_{2i}\step\alpha_{2i+1}$ is a down-move and $\alpha_{2i+1}\step\alpha_{2i+2}$ is an up-move.
%

\end{definition}

%
 In Fig.~\ref{fig:phipath} are illustrated some alternating paths.
 Observe that whenever $\cc X$ occurs in $\FFun{\cc X}{A}$, we can construct a \emph{cyclic} alternate path in the tree of  $\forall X. (\FFun{\cc X}{A}\To \mm X)\To \FFun{\cc X}{B}$: down-move from an occurrence of $\cc X$  in $A$ to the modular node labeled $\mm X$, then up-move back to $\cc X$. We deduce that if no cyclic alternate path exists, then any subtype of the form above must be such that $\cc X$ does not occur in $A$. This leads to introduce the following:

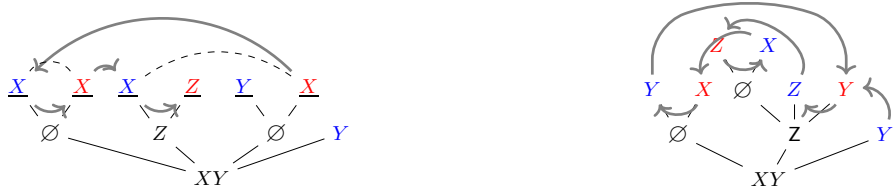
\begin{figure}
\begin{subfigure}{0.48\textwidth}
\adjustbox{scale=.85, center}{$
\begin{tikzpicture}[baseline=-20]

\node (a) at (0,0) {\small$\U{\cc X}$};
\node (aa) at (1,0) {\small$\U{\mm X}$};
\node (u) at (0.5,-0.7) {\small$\emptyset$};

\draw (a) to (u);
\draw (aa) to (u);

\node (b) at (1.7,0) {\small$\U{\cc X}$};
\node (bb) at (2.7,0) {\small$\U{\mm Z}$};
\node (uu) at (2.2,-0.7) {\small$Z$};

\draw (b) to (uu);
\draw (bb) to (uu);

\node (cb) at (3.5,0) {\small$\U{\cc Y}$};
\node (cbb) at (4.5,0) {\small$\U{\mm X}$};
\node (c) at (4,-0.7) {\small$\emptyset$};

\draw (cb) to (c);
\draw (cbb) to (c);

\node (d) at (5,-0.7) {\small$\cc Y$};

\node(t) at (3,-1.4) {\small$ XY$};

\draw (c) to (t);
\draw (d) to (t);
\draw (u) to (t);
\draw (uu) to (t);

\draw[dashed, bend right=35] (cbb) to (b);
\draw[->,very thick, gray, bend right=45] (cbb) to (a);
\draw[dashed, bend right=55] (aa) to (a);
\draw[->,very thick, gray, bend left=55] (aa) to (b);

\draw[->,very thick, gray, bend right=45] (a) to (aa);
\draw[->,very thick, gray, bend right=45] (b) to (bb);


\end{tikzpicture}
$}
\caption{\small Alternate path in \\ \adjustbox{scale=0.9}{$\forall XY. (\cc X\To \mm X)\To (\forall Z.\cc X\To \mm Z)\To (\cc Y\To \mm X)\To \cc Y$}.}
\label{}
\end{subfigure}
\begin{subfigure}{0.5\textwidth}
\adjustbox{scale=0.85, center}{$
\begin{tikzpicture}[baseline=-20]

\node(a) at (0,0) {\small$\cc Y$};
\node(b) at (0.8,0) {\small$\mm X$};

\node(c) at (0.4,-0.7) {\small$\emptyset$};

\draw(a) to (c);
\draw(b) to (c);

\node(e) at (1,0.7) {\small$\mm Z$};
\node(f) at (1.8,0.7) {\small$\cc X$};

\node(g) at (1.4,0) {\small$\emptyset$};

\draw(g) to (e);
\draw(g) to (f);

\node (gg) at (2.2,0) {\small$\cc Z$};
\node (ggg) at (3,0) {\small$\mm Y$};

\node (h) at (2.2,-0.7) {\small$\bb Z$};

\draw(h) to (gg);
\draw(h) to (ggg);
\draw(h) to (g);

\node(i) at (3.6,-0.7) {\small$\cc Y$};

\node(j) at (1.8,-1.4) {\small$XY$};

\draw(j) to (i);
\draw(h) to (j);
\draw(c) to (j);

\draw[->,very thick, gray, bend right=55] (i) to (ggg);

\draw[->,very thick, gray, bend left=55] (ggg) to (gg);
\draw[->,very thick, gray, bend right=85] (gg) to (e);
\draw[->,very thick, gray, bend right=65] (e) to (f);
\draw[->,very thick, gray, bend right=65] (f) to (b);
\draw[->,very thick, gray, bend left=65] (b) to (a);
\draw[->, very thick, gray, rounded corners=15] (a) to (0,1) to [bend left=20] (3,1) to (ggg);

\end{tikzpicture}
%
%
%
%
$}
\caption{\small Cyclic alternate path in \\ 
\adjustbox{scale=0.9}{$\forall XY. (\cc Y\To \mm X)\To (\forall Z.(\mm Z\To \cc X)\To (\cc Z\To \mm Y))\To \cc Y$}. }
\label{fig:phipath}
\end{subfigure}
\caption{Examples of alternate paths.}
\end{figure}
%
%
%

\begin{definition}
For any polynomial tree $\bb E$, the \emph{characteristic of $\bb E$}, $\kappa(\bb E)\in\{0,1,\infty\}$ is defined as follows:
if $\bb E$ is coherent, then $\kappa(\bb E)=0$ if it has no cyclic alternating path, and $\kappa(\bb E)=1$ if it has a cyclic alternating path; if $\bb E$ is not coherent, $\kappa(\bb E)=\infty$.
\end{definition}

%
%

The characteristic is indeed stable under standard reduction:

\begin{restatable}{lemma}{kappaz}
\label{lemma:kappa}
 For all $\bb E, \bb E'$, if $\bb E $ reduces to $ \bb E'$ by standard reduction, then $\kappa(\bb E')\leq\kappa(\bb E)$.
\end{restatable}
\begin{proof}
If $\kappa(\bb E)=\infty$ the claim is trivial.
If $\kappa(\bb E)=1$, then by Lemma \ref{lemma:solvable}, $\kappa(\bb E')\leq 1=\kappa(\bb E)$.


Suppose now $\kappa(\bb E)=0$. We must show that there is no cyclic alternate path in $\bb E'$, on the assumption that there is no cyclic path in $\bb E$. Suppose $\bb E'$ is obtained by a reduction as the one in Fig. \ref{fig:11}. 
As in the proof of Lemma \ref{lemma:solvable}, for any node $\alpha$ of $\bb E'$, let $g(\alpha)$ be the unique node in $\bb E$ 
 $\alpha$ ``comes from''.
We will show that any alternating path from $\gamma_{1}$ to $ \gamma_{n}$ in $\bb E'$ induces an alternate path from $g(\gamma_{1})$ to  $g(\gamma_{n})$ in $\bb E$. From this fact it immediately follows that $\kappa(\bb E')\leq 0$, since any cyclic path in $\bb E'$ would induce a cyclic path in $\bb E$, against the assumption that $\kappa(\bb E)= 0$.

First observe that if $\delta: Y^{c}$ is a modular node in $\bb E'$, then $g(\delta):g( Y)^{c}$ is a modular node in $\bb E$. In fact one can check that:
	\begin{itemize}
	\item if $g(\delta):g(Y)^{ c}$ is the head of the sub-tree of $\bb E$ rooted in $\bb r_{g(Y)}$, then also $\delta$ is the head of the sub-tree of $\bb E'$ rooted in $\bb r_{Y}$ (it suffices to check for the cases $\bb r_{g(Y)}\preceq \bb r_{X}$ and $\bb r_{g(X)}\prec \bb r_{X}$);
	\item if $d(g(\delta), \bb r_{g(Y)})>2$, then $d(\delta, \bb r_{Y})>2$ (again, one has to check the cases $\bb r_{g(Y)}\preceq \bb r_{X}$ and $\bb r_{g(X)}\prec \bb r_{X}$);
	
 	\item if $g(\delta): g(Y)^{ c}$ has a parallel node $\beta: g(Y)^{ c}$, then also $\delta:Y^{ c}$ has a parallel node $\beta':Y^{ c}$ with $g(\beta')=\beta$.
	\end{itemize}
From this it follows that if $(\gamma: Y^{c})\step (\delta: Y^{\OV c})$ is a down-move in $\bb E'$; then $(g(\gamma): g(Y)^{c})\step (g(\delta):g( Y)^{\OV c})$ is a down-move in $\bb E$. 
Moreover, if $(\gamma: Y^{c})\step (\delta: Z^{d})$ is an up-move in $\bb E'$, then two possibilities arise:

	\begin{itemize}
	\item  $\gamma: Y^{c}$ is a modular node occurring either in some of the $\FFun{\cc X\mapsto \cc\bullet}{\cc{\mathsf D}_{jk}}$ or of the $\FFun{\cc X\theta}{\mm{\mathsf F}_{l}}$, and $\delta: Z^{d}$ is a terminal node in the same sub-tree. Then $g(\gamma)$ is a modular node
 occurring either in some copy of $\FFun{\cc X}{\cc{\mathsf D}_{jk}}$ or in $\FFun{\cc X}{\mm{\mathsf F}_{l}}$
 and $g(\delta)$ is a terminal node in the same sub-tree, so we conclude that $g(\gamma)\step g(\delta)$ is an up-move in  $\bb E$.

	\item $\gamma: Y^{c}$ is a modular node occurring in some copy of $\FFun{\cc X\theta}{\mm{\mathsf F}_{l}}$, and $\delta: Z^{d}$ is a terminal node in the sub-tree $\FFun{\cc X\mapsto \cc\bullet}{\cc{\mathsf D}_{jk}}$ which is on top of $\FFun{\cc X\theta}{\mm{\mathsf F}_{l}}$. Then $g(\gamma)$ is modular in $\bb E$ and by construction there exists a terminal node $\alpha:\cc X$ in $\FFun{\cc X\theta}{\mm{\mathsf F}_{l}}$, a modular node $\beta: \mm X$ in $\FFun{\cc X}{\cc{\mathsf D}_{jk}}$, so that there is an alternate path $g(\gamma)\step \alpha\step \beta\step g(\delta)$ in $\bb E$, as illustrated in Fig. \ref{fig:cycle}. 	
	
%
%
%
	 \end{itemize}
	 
We conclude that if $\alpha_{0}\dots \alpha_{2n}$ is an alternate path in $\bb E'$, then we can find an alternate path 
$g(\alpha_{0})g(\alpha_{1})\dots g(\alpha_{2})g(\alpha_{3})\dots g(\alpha_{2n-2})g(\alpha_{2n-1})\dots g(\alpha_{2n})$ in $\bb E$.
	 
A similar argument can be developed for the other tree reduction rules.
%
%
%
\end{proof}

\begin{figure}
\adjustbox{scale=0.75, center}{$
\begin{tikzpicture}[baseline=-10]

\node[regular polygon,regular polygon sides=3, draw, shape border rotate=180] (a) at (0,0) { \ \ \ \ \ };
\node(aa) at (0.35,-1) {\small$\mm E_{l}$};
\node(y) at (0.6,0.7) {\small$\cc Y$};
\node(b) at (-0.6,0.7) {\small$\cc X$};

\node(xy) at (-1,-1.8) {\small$ X$};

\node(bz) at (-2, -1) {\small$\vec Z_{k}$};

\node(x) at (-1.3, -0.3) {\small$\mm X$};

\node[regular polygon,regular polygon sides=3, draw, shape border rotate=180] (v) at (-2.7,0.7) { \ \ \ \ \ };
\node(aa) at (-3.3,-0.3) {\small$\cc{\mathsf D}_{jk}[\cc X]$};

\node(z) at (-2.1, 1.4) {\small$\mm Z$};

\draw[thick] (xy) to (-0.05,-1.1);
\draw[thick] (xy) to (bz);
\draw[thick] (bz) to (x);
\draw[thick] (bz) to (-2.6,-0.4);

\draw[->, dotted, thick] (y) to [bend left=85] (b);
\draw[->, dotted, thick] (b) to [bend right=85] (x);
\draw[->, dotted, thick] (x) to [bend left=85] (z);

\end{tikzpicture}
\qquad \qquad\leadsto_{\cc X} \qquad \qquad 
\begin{tikzpicture}[baseline=30]

\node[regular polygon,regular polygon sides=3, draw, shape border rotate=180] (a) at (0,0) { \ \ \ \ \ };
\node(aa) at (0.35,-1) {\small$\mm E_{l}$};
\node(y) at (0.6,0.7) {\small$\cc Y$};
\node(b) at (-0.6,0.7) {\small$\bullet$};

\node(bb) at (-0.1, 1.4) {\small$\cc \bullet$};

\node(bz) at (-1, 1.4) {\small$Z_{k}$};

\node(br) at (-0.5, 2.1) {\small$\mm \bullet$};
\node[regular polygon,regular polygon sides=3, draw, shape border rotate=180] (v) at (-1.5,3.2) { \ \ \ \ \ };
\node(aa) at (-2.2,2.2) {\small$\cc{\mathsf D}_{jk}[\cc\bullet]$};

\node(z) at (-0.9, 3.9) {\small$\mm Z$};

\draw[thick] (b) to (bz);
\draw[thick] (b) to (bb);
\draw[thick] (br) to (bz);
\draw[thick] (-1.45,2.05) to (bz);

\draw[->, dotted, thick] (y) to [bend left=55] (b) to [bend left=10] (z);
\end{tikzpicture}$}
\caption{\small When $\bb E\leadsto_{\cc X}\bb E'$ is an up-move in $\bb E'$ induces an alternate path in $\bb E$.}
\label{fig:cycle}
\end{figure}

Using the observation above and Lemma \ref{lemma:kappa} we can easily prove:

\begin{proposition}\label{prop:kappone}
Suppose $\kappa(\bb E)\in \{0,1\}$ and $\bb E$ reduces to  $\bb E'\in \C P_{0}$ by standard reduction. If $\kappa(\bb E)=0$, then $\tau(\bb E')\in \NI$, and  if $\kappa(\bb E)=1$, then $\tau(\bb E')\in \NImu$.
\end{proposition}

For a $\Nd$-type $A$, we can define its characteristic as $\kappa(A)=\kappa(\cc{\mathsf t}(A))$. From Prop. \ref{prop:kappone} and Lemma \ref{lemma:converge} we deduce then a  new criterion for finiteness:

\begin{corollary}
Let $A$ be a closed $\Nd$-type. If $\kappa(A)=0$, then $A\equiv_{\bb Y}0+1+\dots+1$.
\end{corollary}

The criterion based on the characteristic can be used to  capture {yet more} finite $\Nd$-types.
In fact, whenever a type $A$ reduces to a type with characteristic $0$, we can deduce that  $A\equiv_{\bb Y}0+1+\dots+1$. Note that such a type $A$ need not even be coherent.
For instance, the (tree of) the type 
 $ A=\forall XY.( \forall Z.( (\cc Z\To \mm Z)\To \cc X) \To \mm X) \To \mm Y\To \cc Y$, is not coherent (since the variable $Z$ is not eliminable), but 
 reduces, by eliminating $X$, to (the tree of) $\forall Y. \mm Y\To \cc Y$, which has characteristic $0$, and in fact we have that $A\equiv_{\bb Y} 1$.
%
%

As this example shows, a type $A$ can reduce to a finite sum $0+1+\dots+1$ even if some of its subtypes cannot be similarly reduced. In fact, while we can eliminate all quantifiers from the type $A$ above, we cannot do this from its subtype $ \forall Z.( (\cc Z\To \mm Z)\To \cc X) \To \mm X)$.
By contrast, in the next section we show that the characteristic satisfies nice compositionality conditions that will allow us to define suitable fragments of $\Nd$.

%
%
%
%

\section{Preliminaries on Theories of Program Equivalences}\label{secEpsilon}


After developing the rewriting theory of Yoneda Type Isomorphisms, our goal in the next section will be to exploit the correspondence it provides between polymorphic and monomorphic types to investigate properties of program equivalence in $\Nd$.
In this section we introduce the fundamental notions of type systems and program equivalence we will be concerned with. 

%
%
%
%
%


\begin{figure}
\fbox{
\begin{subfigure}{\textwidth}
\adjustbox{scale=0.8, center}{$
\begin{matrix}
\qquad
\AXC{}
\UIC{$\Gamma, x:A\vdash x:A$}
\DP
\qquad
\AXC{$\Gamma, x:A\vdash t:B$}
\UIC{$\Gamma\vdash {\lambda x.t}:A\to B$}
\DP
\qquad
\AXC{$\Gamma\vdash  t:A\to B$}
\AXC{$\Gamma\vdash  u:A$}
\BIC{$\Gamma\vdash {tu}:B$}
\DP\\ \\
\AXC{$\Gamma\vdash t:A$}
\RL{\small$X\notin FV(\Gamma)$}
\UIC{$\Gamma\vdash {\Lambda X.t}:\forall X.A$}
\DP
\qquad
\AXC{$\Gamma\vdash  t:\forall X.A$}
\UIC{$\Gamma\vdash {tB}:A[B/X]$}
\DP
\end{matrix}$}
\caption{Typing rules for $\Nd$.}
\label{fig:typf}
\end{subfigure}
}

\

\

\fbox{
\begin{subfigure}{\textwidth}
\adjustbox{scale=0.8, center}{$
\begin{matrix}
\AXC{$\Gamma\vdash t:A$}
\AXC{$\Gamma\vdash u:B$}
\BIC{$\Gamma\vdash{ \langle t,u\rangle}:A\times B$}
\DP \qquad
\AXC{$\Gamma\vdash t:A_{1}\times A_{2}$}
\UIC{$\Gamma\vdash  {\pi_{i}^{A_{i}}t}:A_{i}$}
\DP
\qquad
\AXC{}
\UIC{$\Gamma\vdash  {\star}:1$}
\DP
\\
 \qquad \\
\AXC{$\Gamma\vdash t:A_{i}$}
\UIC{$\Gamma\vdash {\iota_{i}t}:A_{1}+A_{2}$}
\DP
\qquad
\AXC{$\Gamma\vdash t:A_{1}+A_{2}$}
\AXC{$(\Gamma,{y}:A_{i}\vdash{ u_{i}}:C)_{i=1,2}$}
\BIC{$\Gamma\vdash {\delta_{C}(t,y.u_{1},y.u_{2})}: C$}
\DP
\qquad
\AXC{$\Gamma\vdash t:0$}
\UIC{$\Gamma \vdash{\xi_{A} t}: A$}
\DP
\\ \qquad \\
\AXC{$\Gamma\vdash t:\FFun{\cc X\mapsto \mu X.\FFun{\cc X}{A}}{A}$}
\UIC{$\Gamma\vdash \inn_{A}t: \mu X.\FFun{\cc X}{A}$}
\DP
\qquad
\AXC{$\Gamma\vdash t: \FFun{\mm X\mapsto B}{A}\To B$}
\UIC{$\Gamma\vdash \ff_{A}(t):  \mu X.\FFun{\mm X}{A}\To B$}
\DP
\\ \qquad \\
\AXC{$\Gamma\vdash t:\nu X.\FFun{\cc X}{A}$}
\UIC{$\Gamma\vdash \outt_{A}t: \FFun{\cc X\mapsto \nu  X.\FFun{\cc X}{A}}{A}$}
\DP
\qquad
\AXC{$\Gamma\vdash t: B\To \FFun{\cc X\mapsto B}{A}$}
\UIC{$\Gamma\vdash \uu_{A}(t): B\To \nu X.\FFun{\cc X}{A}$}
\DP
\end{matrix}
$}
\caption{Typing rules for $+,\times, 0,1,\mu, \nu$.}
\label{fig:typtot}
\end{subfigure}
}

\

\

\fbox{
\begin{subfigure}{\textwidth}
\adjustbox{scale=0.8, center}{$
\begin{matrix}
{(\lambda x.t)u \ }\simeq_{\beta}{ \ t[u/x]} 
\qquad {(\Lambda X.t)B \ } \simeq_{\beta }{ \ t[B/X] }
\\
\ \\
\AXC{$\Gamma\vdash  t:A\to B$}
\UIC{$\Gamma\vdash {t }\simeq_{\eta}{ \lambda x.tx}:A\to B$}
\DP
 \qquad
 \AXC{$\Gamma\vdash  t:\forall X.A$}
\UIC{$\Gamma\vdash {t} \ \simeq_{\eta} \ {\Lambda X.tX}:\forall X.A$}
\DP
\end{matrix}
$}
\caption{$\beta$ and $\eta$-rules for $\Nd$.}
\label{eqscherer0}
\end{subfigure}
}

\

\

\fbox{
\begin{subfigure}{\textwidth}
\adjustbox{scale=0.8, center}{$
\begin{matrix}
 \Big ({\pi_{i}^{A_{i}}\langle t_{1},t_{2}\rangle \ }\simeq_{\beta}{ \ t_{i} } \Big)_{\scalerel*{i=1,2}{X}}\qquad
\Big({\delta_{C}(\iota_{i}t, y.u_{1},y.u_{2}) \ }\simeq_{\beta}{ \ u_i[t/y]} \Big)_{\scalerel*{i=1,2}{X}}
\\
\
\\
\AXC{$\Gamma\vdash  t: A\times B$}
\UIC{$\Gamma\vdash {
t }\simeq_{\eta}{  \langle\pi_{1}^{A}t, \pi_{2}^{B}t\rangle }: A\times B$}
\DP
\qquad\qquad
\qquad\AXC{$\Gamma\vdash  t:1$}
\UIC{$\Gamma\vdash {t}\simeq_{\eta}\bb{ \star}: 1$}
\DP 
\\
\
\\
\AXC{$\Gamma\vdash t:A+B$}
\AXC{$\rr{ u[x]}:  A+B \vdash^{\Gamma} C$}
\BIC{$\Gamma\vdash {\rr{ u}[t]
}  \simeq_{\eta} {
\delta_{C}(t, y.\rr{ u}[\iota_{1}y], y.\rr{ u}[\iota_{2}y]) }:C$}
\DP
\qquad\qquad
\AXC{$\Gamma\vdash  t:0$}
\AXC{$\rr{ u }[x]: 0 \vdash^{\Gamma}A$}
\BIC{$\Gamma\vdash{ \rr{ u}[t] }\simeq_{\eta} {\xi_{A} t}: A$}
\DP 
\end{matrix}
$}
\caption{$\beta$ and $\eta$-rules for $+,\times, 0,1$.}
\label{eqscherer}
\end{subfigure}
}

\

\

\fbox{
\begin{subfigure}{\textwidth}
\adjustbox{scale=0.8, center}{$
\begin{matrix}
\ff_{A}(t)(\inn_{A}u )\simeq_{\beta} t ( \Fun{X}{A}( \ff_{A}(t) x)[x\mapsto u ])
\\
\
\\
 \outt_{A}( \uu_{A}(t)u) \simeq_{\beta}
 \Fun{X}{A}( \uu_{A}(t)x     )[x\mapsto tu]
\\
\
\\
\AXC{$\Gamma\vdash u:\FFun{\mm X\mapsto C}{A}\To C$}
\AXC{$t[x]: \mu X.\FFun{\mm X}{A}\vdash^{\Gamma}C$}
\AXC{$t[\inn_{A}x]\simeq u\Fun{X}{A}(t): \FFun{\mm X\mapsto \mu X.\FFun{\mm X}{A}}{A}\vdash^{\Gamma}C$}
\TrinaryInfC{$ t[x]\simeq_{\eta}\ff_{A}(u)x: \mu X.\FFun{\mm X}{A}\vdash^{\Gamma} C$}
\DP
\\
\
\\
\AXC{$\Gamma\vdash u:C\To \FFun{\cc X\mapsto C}{A}$}
\AXC{$t[x]: C\vdash^{\Gamma}\nu X.\FFun{\cc X}{A}$}
\AXC{$\Fun{X}{A}(t)[x\mapsto ux]\simeq \outt_{A}t: C\vdash^{\Gamma}\FFun{\cc X\mapsto \nu X\FFun{\cc X}{A}}{A}$}
\TrinaryInfC{$ t[x] \simeq_{\eta} \uu_{A}(u)x : C\vdash^{\Gamma} \nu X.\FFun{\cc X}{A}$}
\DP
\end{matrix}
$}
\caption{$\beta$ and $\eta$-rules for $\mu, \nu$.}
\label{eqscherermu}
\end{subfigure}
}
\caption{Typing rules and $\beta\eta$-rules.}
\end{figure}
 
%
%
%

%

The terms of $\Nd$ are generated by usual $\lambda$-calculus constructors, i.e. variables $x,y,z,\dots$, application $tu$ and abstraction $\lambda x.t$, plus polymorphic constructors $tA$, for all type $A$ and $\Lambda  X.t$ for all $X\in \TT V$.
The terms of $\Ntot$ (and its $\mu,\nu$-free fragment $\Ndv$) involve, in addition to the constructors for $\Nd$, the following constructors:
 product constructors $\langle t,u\rangle$, $\pi_{i}t$ ($i=0,1$); sum constructors $\iota_{i}t$ ($i=0,1$) and  $\delta_{A}(t, y.u_{1},y.u_{2})$ for all type $A$;
 the $1$ constructor $\star$; the 0 destructors $\xi_{A}t$, for all type $A$;
  fixpoints constructors $\ff_{A}(t),\inn_{A}t, \uu_{A}(t), \outt_{A}t$ for any type $A$.

\begin{notation}Let $L=\langle i_{1},\dots, i_{k}\rangle$ be a list. If
  $(t_{k})_{i\in L}$ is a $L$-indexed list of terms, we let  for any term $u$, $u\langle t_{k}\rangle_{i\in L}$ be shorthand for $ut_{i_{1}}\dots t_{i_{k}}$; if $\langle x_{i_{1}},\dots, x_{i_{k}}\rangle$ is a $L$-indexed list of variables, we let  for any term $u$, $\lambda\langle x_{k}\rangle_{i\in L}.u$ be shorthand for $\lambda x_{i_{1}}\dots \lambda x_{i_{k}}.u$. 
  \end{notation}

The  typing rules of $\Ndv$ and $\Ntot$ are recalled in Fig.~\ref{fig:typf}-\ref{fig:typtot}.  
$\beta$ and $\eta$-rules are recalled in Fig.~\ref{eqscherer0}-\ref{eqscherer} and \ref{eqscherermu}.
We let the \emph{$\beta\eta$-theory} of $\Ntot$ (and its fragments) be the smallest congruent theory closed under $\beta$ and $\eta$ rules.

\medskip


In addition to the $\beta\eta$-theory, in the next sections we will investigate two other theories: the one arising from \emph{contextual equivalence} and the \emph{$\varepsilon$-theory}. 

\begin{definition}
The \emph{contextual equivalence} relation for $\Ntot$ and $\NI$ is defined by
\begin{center}
$\Gamma\vdash_{} t\simeq_{\ctx}u:A$ iff for all context $\TT C: (\Gamma\vdash_{} A)\To(\vdash_{} 1+1)$, $\TT C[t]\simeq_{\beta\eta}\TT C[u]$
\end{center}
The contextual equivalence relation for $\Nd$ is defined by
\begin{center}
$\Gamma\vdash_{} t\simeq_{\ctx}u:A$ iff for all context $\TT C: (\Gamma\vdash_{} A)\To (\vdash_{} \forall X.X\To X\To X)$, $\TT C[t]\simeq_{\beta\eta}\TT C[u]$
\end{center}
\end{definition}

It is a standard result that contextual equivalence (for either $\Ntot$ or $\Nd$) is closed under congruence rules and thus generates an equational theory $\ctx$. 

An equational theory $\bb T$ for $\Ntot$ is \emph{consistent} when it does not contain the equation $\iota_{1}(\star)\simeq_{\bb T}\iota_{2}(\star): 1+1$. Similarly, an equational theory $\bb T$ for $\Nd$ is \emph{consistent} when it does not contain the equation $\Lambda X.\lambda xy.x\simeq_{\bb T}\Lambda X.\lambda xy.y: \forall X.X\To X\To X$. It is well-known that $\ctx$ is the \emph{maximum consistent theory} (both for $\Ntot $and $\Nd$).

\medskip


To introduce the $\varepsilon$-theory we first recall how positive types $\FFun{\cc X}{A}$ and negative type $\FFun{\mm X}{B}$ are associated with functors
$\Fun{X}{A}: \CTX_{\beta\eta}(\Nd)\to \CTX_{\beta\eta}(\Nd)$ and 
$\Fun{X}{B}: \CTX_{\beta\eta}(\Nd)^{\mathsf{op}}\to \CTX_{\beta\eta}(\Nd)$. We let 
$\Fun{X}{A}(C)=A[C/X]$, $\Fun{X}{B}(C)=B[C/X]$ and for all $t[x]: C\vdash D$, 
\begin{align}
\Fun{X}{X}(t)& =t\\ 
\Fun{X}{Y}(t)& =x   \\
\Fun{X}{(C\To D)}(t) & = \lambda y. \Fun{X}{D}  ( t[\Fun{X}{C}(y)]   )  \\
\Fun{X}{(\forall X.C)}(t) & = \Lambda Y. \Fun{X}{C}(tY)
\end{align}
One can check that $\Fun{X}{A}(x)\simeq_{\eta}x$ and $\Fun{X}{A}(t[u[x]])\simeq_{\beta} \Fun{X}{A}(t)\Big[ x\mapsto \Fun{X}{A}(u)\Big ]$.

The definition above can be extended to the types defined using all other constructors, yielding functors
$\Fun{X}{A}: \CTX_{\beta\eta}(\Ntot)\to \CTX_{\beta\eta}(\Ntot)$ and 
$\Fun{X}{B}: \CTX_{\beta\eta}(\Ntot)^{\mathsf{op}}\to \CTX_{\beta\eta}(\Ntot)$.

For simplicity, we will define the $\varepsilon$-theory for an \emph{ad-hoc} fragment $\Nnew$ of $\Ntot$, 
in which types are restricted so that a universal type $\forall X.A$ is always of one 
of the two forms below:
\begin{equation}\label{eq:adhoc}
\forall X.\Big \langle \forall \vec Y_{k}. \big\langle \FFun{\cc X}{A_{jk}}\big\rangle_{j}\To \mm X\Big\rangle_{k}\To \FFun{\cc X}{B} 
\qquad\qquad
 \forall X. \Big \langle \forall \vec Y_{j}.\cc X\To \FFun{\mm X}{A_{j}}\Big\rangle_{j}\To \FFun{\mm X}{B} 
\end{equation}
As this set of types is stable by substitution, all type rules and equational rules of $\Ntot$ scale well to $\Nnew$.
We will meet again this subsystem of $\Ntot$ in the next section as well as in App.~\ref{app8}.


To each universal type $\forall X.A$ of $\Nnew$ as in Eq. \eqref{eq:adhoc} left we can associate an equational rule, that we call a \emph{$\varepsilon$-rule}, illustrated in Fig.~\ref{fig:dinazza}; similarly, to each type as in Eq. \eqref{eq:adhoc} right we can associate a {$\varepsilon$-rule} as illustrated in Fig.~\ref{fig:dinazzabis}.\footnote{Observe that $\forall X.C$ might well be of \emph{both} forms \eqref{eq:adhoc} left and right.}
From the viewpoint of category theory, the two $\varepsilon$-rules for a type $\forall X.A$ correspond to requiring 
that the transformations induced by a polymorphic programs of type $\forall X.A$  are  \emph{strongly dinatural} \cite{Uustalu2011}, or, in other words, that the diagrams illustrated in Fig.~\ref{fig:catdinazza} and Fig.~\ref{fig:catdinazzabis} commute.

\begin{definition}[$\varepsilon$-theory]
The \emph{$\varepsilon$-theory} of $\Nnew$ is the smallest congruent equational theory closed under $\beta$-, $\eta$-equations as well as $\varepsilon$-rules.
\end{definition}


\begin{figure}[t]
\begin{subfigure}{0.45\textwidth}
\parbox[h][4cm][c]{\textwidth}{
\adjustbox{scale=0.8, center}{$
\AXC{$ t: \forall X.C$}
\noLine
\UIC{$\left ( e_{k}: \forall \vec Y_{k}.\langle \FFun{E}{A_{jk}}\rangle_{j} \To E\right)_{k} $}
\noLine
\UIC{$ \left ( f_{k}:\forall \vec Y_{k}. \langle \FFun{F}{A_{jk}}\rangle_{j} \To F\right)_{k}$}
\noLine
\UIC{$ \TT v[x]: E\vdash F$}
\noLine
\UIC{$ \TT v[ e_{k}\vec Y_{k}\langle z_{j}\rangle_{j}  ] \simeq f_{k}\vec Y_{k}\langle \Fun{X}{A_{jk}}(\TT v)[z_{j}]\rangle_{j}: \langle\FFun{\cc X\mapsto E}{A_{jk}}\rangle_{j}\vdash F$}
\UIC{$
\Gamma\vdash 
\Fun{X}{B}(\TT v)\Big [x\mapsto t E \langle e_{k}\rangle_{k}\Big]
\simeq
t F \langle f_{k}\rangle_{k}: \FFun{\cc X\mapsto F}{B}
$}
\DP
$}
}
\caption{\small $\varepsilon$-rule for the left-hand type in Eq. \eqref{eq:adhoc}.}
\label{fig:dinazza}
\end{subfigure} \ \ \ \ \ 
\begin{subfigure}{0.45\textwidth}
\parbox[h][4cm][c]{\textwidth}{
\adjustbox{scale=0.8, center}{$
\AXC{$ t: \forall X.C$}
\noLine
\UIC{$\left ( e_{j}: \forall \vec Y_{k}. E \To  \FFun{E}{A_{j}} \right)_{j} $}
\noLine
\UIC{$\left (  f_{j}:\forall \vec Y_{k}. F \To   \FFun{F}{A_{j}}\right )_{j}$}
\noLine
\UIC{$ \TT v[x]: E\vdash F$}
\noLine
\UIC{$\Fun{X}{A_{j}}( \TT v[x])[ e_{j} x] \simeq
f_{j}\TT v[x]
: E \vdash \FFun{\cc X\mapsto F}{A_{j}}
$}
\UIC{$
\Gamma\vdash 
 \Fun{X}{B}(\TT v)\Big [x\mapsto t F \langle f_{j}\rangle_{j}\Big]\simeq  t E \langle e_{j}\rangle_{j}
: \FFun{\cc X\mapsto E}{B}
$}
\DP
$}
}
\caption{\small $\varepsilon$-rule for the left-hand type in Eq. \eqref{eq:adhoc}.}
\label{fig:dinazzabis}
\end{subfigure}

\bigskip

\begin{subfigure}{\textwidth}
\adjustbox{scale=0.9, center}{$
\begin{tikzcd}
  &  \big\langle \forall \vec Y_{k}.\langle A_{jk}\langle E\rangle \rangle_{j}\To E \big\rangle_{k}
   \ar{rr}{\langle y_{k}\rangle_{k}\mapsto tE\langle y_{k}\rangle_{k}} 
  \ar{rd}{ \langle \forall \vec Y_{k}. \langle A_{jk}\langle E\rangle\rangle_{j}\To \TT v  \rangle_{k} } 
   & &  B\langle E\rangle \ar{dd}{B(\TT v)}   \\
1 \ar{ru}{\langle e_{k}\rangle_{k}} \ar{rd}[below]{\langle f_{k}\rangle_{k}}  &  & \langle \forall \vec Y_{k}.\langle A_{jk}\langle E\rangle \rangle_{j}\To F\rangle_{k}  & \\
 &\big\langle \forall \vec Y_{k}.\langle A_{jk}\langle F\rangle \rangle_{j}\To F\big\rangle_{k} 
 \ar{rr}[below]{\langle y_{k}\rangle_{k}\mapsto tF\langle y_{k}\rangle_{k}}
 \ar{ru}[right]{ \langle \forall \vec Y_{k}.\langle A_{jk}\langle \TT v\rangle\rangle_{j} \To F\rangle_{k}}
    &  & B\langle F\rangle
\end{tikzcd}
$
}
\caption{Strong dinaturality diagram for the $\varepsilon$-rule {\textrm{(a)}}.}
\label{fig:catdinazza}
\end{subfigure}

\bigskip

\begin{subfigure}{\textwidth}
\adjustbox{scale=0.9, center}{$
\begin{tikzcd}
  &  \langle \forall \vec Y_{j}.E\To  A_{j}\langle E\rangle \rangle_{j}
   \ar{rr}{\langle  y_{j}\rangle_{j}\mapsto tE\langle y_{j}\rangle_{j}} 
  \ar{rd}{ \forall \vec Y_{j}.  \langle E \To  A_{j}\langle \TT v\rangle  \rangle_{j} }  & &  B\langle E\rangle    \\
1 \ar{ru}{\langle e_{j}\rangle_{j}} \ar{rd}[below]{\langle f_{j}\rangle_{j}}  &  &\langle \forall \vec Y_{j}.E\To  A_{j}\langle F\rangle \rangle_{j}  & \\
 &\langle \forall \vec Y_{j}.F\To  A_{j}\langle F\rangle \rangle_{j}  
 \ar{rr}[below]{\langle  y_{j}\rangle_{j}\mapsto t F\langle y_{j}\rangle_{j}}
 \ar{ru}[right]{ \forall \vec Y_{j}.\langle \TT v\To   A_{j}\langle F\rangle \rangle_{j}}
    &  & B\langle F\rangle\ar{uu}{B(\TT v)}
\end{tikzcd}
$
}
\caption{Strong dinaturality diagram for the $\varepsilon$-rule \textsf{\textrm{(b)}}.}
\label{fig:catdinazzabis}
\end{subfigure}
\caption{$\varepsilon$-rules and their associated strong dinaturality diagrams.}
\label{fig:dinazzone}
\end{figure}
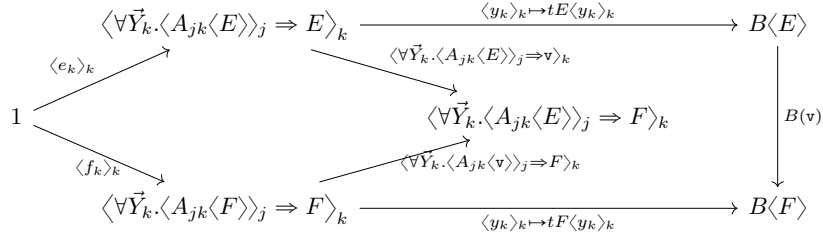
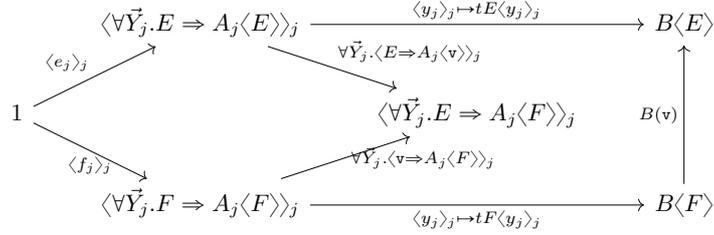

%

%

For the interested reader, in App.~\ref{app4} we check in detail that the isomorphisms $\equiv_{\cc X}, \equiv_{\mm X}$ hold under the $\varepsilon$-theory of $\Ntot$.

\begin{figure}
\fbox{
\begin{subfigure}{\textwidth}
\resizebox{0.35\textwidth}{!}{
\begin{minipage}{0.4\textwidth}
\begin{equation*}
\begin{split}
X^{\rop}  & = X \\
(A\to B)^{\rop} &= A^{\rop}\to B^{\rop} \\
(\forall X.A)^{\rop} &= \forall X.A^{\rop}\\
1^{\rop}& = \forall X.X\to X \\
0^{\rop}& = \forall X.X \\
\end{split} 
\end{equation*}
\end{minipage}
}
\ \ \ \ \ \ \ \ %
\resizebox{0.35\textwidth}{!}{
\begin{minipage}{0.4\textwidth}
\begin{equation*}
\begin{split}
(A\times B)^{\rop}& = \forall X.(A^{\rop}\to B^{\rop}\to X)\to X \\
(A+ B)^{\rop}& = \forall X.(A^{\rop}\to X)\to (B^{\rop}\to X)\to X\\
(\mu X.\FFun{\cc X}{A})^{\sharp} & = \forall X. (\FFun{\cc X}{A^{\sharp}}\To X)\To X \\
(\nu X.\FFun{\cc X}{A})^{\sharp} & = \forall Y. (\forall X.X\To (X\To \FFun{\cc X}{A^{\sharp}})\To Y)\To Y
\end{split} 
\end{equation*}
\end{minipage}
}
\caption{$^{\sharp}$-translation of types.}
\label{fig:startypes}
\end{subfigure}
}

\

\

\fbox{
\begin{subfigure}{\textwidth}
\resizebox{0.18\textwidth}{!}{
\begin{minipage}{0.19\textwidth}
\begin{align*}
{ x^{\rop} }& = { x }\\ 
{ (\lambda x.t)^{\rop}} & ={   \lambda x.t^{\rop} }\\
{  (tu)^{\rop} }& ={ t^{\rop}u^{\rop}} \\
(\Lambda X.t)^{\rop} & = \Lambda X. t^{\rop} \\
(tB)^{\rop} & = t^{\rop}B^{\rop} \\
{ \star^{\rop}}& = { \Lambda X.\lambda x.x} \\
{  (\xi_{C} t)^{\rop} } & ={  t^{\rop}C^{\rop}} \\
\end{align*}
\end{minipage}
}
\resizebox{0.72\textwidth}{!}{
\begin{minipage}{0.85\textwidth}
\begin{align*}
{ \langle t,u\rangle^{\rop}}& = { \Lambda X.\lambda y.yt^{\rop}u^{\rop} }\\
{   (\pi_{i}^{C}t)^{\rop}} & ={  t^{\rop}C^{\rop}\lambda x_{1}.\lambda x_{2}.x_{i} } \\
{ (\iota_{i}t)^{\rop} }& ={  \Lambda X.\lambda x_{1}.\lambda x_{2}.x_{i}t^{\rop}} \\ 
{ (\delta_{C}(t,y.u_{1},y.u_{2})^{\rop}} & ={  t^{\rop}C^{\rop}\lambda y.u_{1}^{\rop} \lambda y.u_{2}^{\rop}}\\
( \inn_{P}t)^{\rop} & =  \Lambda X.\lambda f.f ( \Fun{X}{(P^{\rop})}(xXf) [x\mapsto  t^{\rop}]) \\
(\ff_{P}(t))^{\rop} & = \lambda x.xC^{\rop}t^{\rop} \quad (t: \FFun{\cc X\mapsto C}{A}\To C)  \\
(\outt_{P}t)^{\rop} & = t^{\sharp}(\FFun{\cc X\mapsto \nu X.A}{A})^{\rop} \Lambda X.\lambda yz. \Fun{X}{(P^{\rop})}(   \Lambda Y.\lambda f.fXxy)[x\mapsto zy] \\
(  \uu_{P}(t))^{\rop}  & = \lambda x. \Lambda Y.\lambda f.fC^{\rop}t^{\rop}x \quad (t: C\To \FFun{\cc X\mapsto C}{A})
\end{align*}
\end{minipage}
}
\caption{$^{\rop}$-translation of terms.}
\label{fig:starbeta}
\end{subfigure}
}

\caption{$^{\sharp}$-translation of sum, product and fixpoint types.}
\label{fig:startra}
\end{figure}

\begin{figure}

\fbox{
\begin{subfigure}{\textwidth}
\begin{center}
\resizebox{0.8\textwidth}{!}{
\begin{minipage}{\textwidth}
\begin{align*}
{\bb{\TT c}_{X}}[x] & = x \\
{\rr{\TT c}_{B\to C}}[x] & ={ \lambda y. \rr{\TT c}_{C}\big[ x (\rr{\TT c}^{-1}_{B}[y])\big ]} \\
{\rr{\TT c}_{\forall Y.B}}[x]& ={ \Lambda Y. \rr{\TT c}_{B}[x Y  ]} \\ 
{\rr{\TT c}_{B\times C}}[x]& = {\Lambda Y. \lambda y. y (\rr{\TT c}_{B}[ \pi_{1}x])(\rr{\TT c}_{C}[\pi_{2}x])}\\
{\rr{\TT c}_{1}}[x]& ={ \Lambda Y. \lambda y.y} \\
{\rr{\TT c}_{B+C}}[x]& ={ \Lambda Y. \lambda ab.\delta_{X}( x,  y.a(\rr{\TT c}_{B}[ y]), y. b(\rr{\TT c}_{C}[ y]))}\\
{\rr{\TT c}_{0}}[x]&={ \xi_{\forall X.X} x}\\
\rr{\TT c}_{\mu X.B}[x]& = \ff_{B}( \lambda x.\Lambda X.\lambda f. f ( \rr{\TT c}_{B}[ \Fun{X}{B}( xXf)) \\
\rr{\TT c}_{\nu X.B}[x]& =\Lambda Y.\lambda f. f(\nu X.B)x\lambda y.\rr{\TT c}_{B}[\nu X.B/X][ \outt_{B}y] 
\end{align*}
\end{minipage}
}
\end{center}
\caption{Definition of $\bb c_{A}[x]: A\vdash A^{\sharp}$.}
\end{subfigure}
}

\medskip

\fbox{
\begin{subfigure}{\textwidth}
\begin{center}
\resizebox{0.8\textwidth}{!}{
\begin{minipage}{\textwidth}
\begin{align*}
%
\bb{\rr{\TT c}^{-1}_{X}}[x]& =x \\
{\rr{\TT c}^{-1}_{B\to C}}[x] & ={\lambda y. \rr{\TT c}^{-1}_{C} \big[x (\rr{\TT c}_{B}[y])\big]}\\
{\rr{\TT c}^{-1}_{\forall Y.B}}[x] & = {\Lambda Y. \rr{\TT c}^{-1}_{B}[ xY ]} \\
{\rr{\TT c}^{-1}_{B\times C}}[x]& ={x (B\times C)  \lambda yz. \langle \rr{\TT c}^{-1}_{B}[y],\rr{\TT c}^{-1}_{C}[z]\rangle} \\
{\rr{\TT c}^{-1}_{1}}[x]&={ \star}\\
{\rr{\TT c}^{-1}_{B+C}}[x] &={
x (B+C)  \lambda y.\iota_{1}(\rr{\TT c}^{-1}_{B}[ y]) \ \lambda y. \iota_{2}(\rr{\TT c}^{-1}_{C} [ y]) }\\
\bb{\rr{\TT c}^{-1}_{0}}[x] & = \bb{x0} \\
\rr{\TT c}^{-1}_{\mu X.B}[x] & = x(\mu X.B )\lambda x.\inn_{B}( \rr{\TT c}^{-1}_{B}[x]) \\
\rr{\TT c}^{-1}_{\nu X.B}[x]& = \uu_{B}(  \lambda z.  \rr{\TT c}^{-1}_{B}[(\nu X.B)^{\rop}/X  ][z (B^{\rop}[(\nu X.B)^{\rop}/X])\Lambda X.\lambda yg. \Fun{X}{(B^{\rop})}(  \TT w[y] )      ]      )x \\
 &( \text{ where }\TT w[y]=
\Lambda Y.\lambda h.hXyg \ )
\end{align*}
\end{minipage}
}
\end{center}
\caption{Definition of $\bb c^{-1}_{A}[x]: A^{\sharp}\vdash A$.}
\end{subfigure}
}

\caption{Definition of the isomorphisms $A\mapsto A^{\sharp}$.}
\label{fig:cc}
\end{figure}

We recall in Fig. \ref{fig:startypes} and \ref{fig:starbeta} the usual embedding $^{\sharp}$ of $\Ntot$ into $\Nd$.
It is well-known that $\beta$-equivalent terms of $\Ntot$ translate under $^{\sharp}$ into $\beta$-equivalent terms of $\Nd$, and that $\eta$-equivalent terms of $\Ntot$ translate under $^{\sharp}$ into terms which 
are equivalent up to dinaturality (see \cite{Plotkin1993, Hasegawa2009}), that is, into $\varepsilon$-equivalent terms of $\Nd$.
Moreover, for all $\Ntot$-type $A$ we can define terms  $\bb{\rr{\TT c}_{A}}[x]: A\vdash_{\Ntot} A^{\sharp}$ and 
$\bb{\rr{\TT c}^{-1}_{A}}[x]: A^{\sharp}\vdash_{\Ntot}A$ (as illustrated in Fig. \ref{fig:cc}).
It is a good exercise with $\varepsilon$-rules to check that 
$\TT c_{A}[\TT c_{A}^{-1}[x]]\simeq_{\varepsilon} x : A^{\sharp}$
 and $\TT c_{A}^{-1}[\TT c_{A}[x]]\simeq_{\varepsilon} x : A$ hold.

\begin{remark}
It is easily checked that for all type $A$ of $\Ntot$, $A\equiv_{\bb Y} A^{\sharp}$ holds, the only non immediate case being 
\begin{center}
\begin{tabular}{r  l}
$\nu X.\FFun{\cc X}{A}$ & $ \equiv_{\cc X} \forall Y.(\nu X.\FFun{\cc X}{A} \To Y)\To Y $ \\
& $ \equiv_{\mm X} 
\forall Y. ( \forall X.(X\To \FFun{\cc X}{A})\To (X\To Y)) \To Y$ \\
& $ \equiv_{\beta\eta}
\forall Y. (\forall X. X\To (X\To \FFun{\cc X}{A})\To Y)\To Y$
\end{tabular}
\end{center}
\end{remark}

\section{System F with Finite Characteristic}\label{secFin}

In this section we show that  Yoneda reduction can be used to characterize the $\varepsilon$-theory in some fragments of $\Nd$. We introduce two fragments $\NY,\NYY$ of $\Nd$ in which types have a fixed finite characteristic, and we show that 
the $\varepsilon$-theory for such fragment can be computed by embedding polymorphic programs into  well-known monomorphic systems. 

We recall that the \emph{free bicartesian closed category} is the category $\BB B=\CTX_{\beta\eta}(\NI)$ and the \emph{free cartesian closed $\mu$-bicomplete category} $\mu\BB B$ \cite{Santocanale2002,Basold2016} is the category $\CTX_{\beta\eta}(\NImu)$. 
We will show that, under the $\varepsilon$-theory, $\NY$ is equivalent to $\BB B$ and $\NYY$ is equivalent to $\mu\BB B$.
We will then use this embedding to establish that this theory is decidable in the fragment of characteristic 0 (where it also coincides with contextual equivalence) and undecidable in the fragment of characteristic 1. 


\subsection{The Systems $\NY$ and $\NYY$}
We first have to check that the types with a fixed finite characteristic do yield well-defined fragments of $\Nd$. This requires to check two properties. First, 
the characteristic has to be \emph{compositional}: a subtype of a type of characteristic $k$ cannot have a higher characteristic, 
since every subtype of a type of the fragment must be in the fragment itself.
Second, since a universally quantified variable can be instantiated with any other type of the fragment, the characteristic must be \emph{closed by instantiation}: if $\forall X.A$ and $B$ have characteristic $k$, then $A[B/X]$ must have characteristic (at most) $k$.

\begin{restatable}{lemma}{compositionality}\label{rem:compositionality}
\begin{description}
\item[(compositionality)] If $A$ is a sub-type of $B$, then $\kappa(A)\leq \kappa(B)$. 
\item[(closure by instantiation)] $\kappa(A[B/X])\leq \max\{\kappa(\forall X.A),\kappa(B)\}$.
\end{description}
\end{restatable}
\begin{proof}

If 
$\phi: \BV{B}\to \bb{Colors}$ and $\cc{\bb t}(B)$ is $\phi$-coherent, then  $\cc{\bb t}(A)$ is $\phi'$-coherent, where $\phi'$ is 
 the restriction $\phi|_{{\BV{A}}}$ of $\phi$ to $\BV{A}$. Moreover, it is clear that any cyclic alternate path in $\cc{\mathsf t}(A)$ yields a cyclic alternate path in $\cc{\mathsf t}(B)$.

The set $\BV{A[B/X]}$ can be identified with $V=\BV{A}\uplus \biguplus_{i=1}^{k}\BV{B}$, where $\uplus$ indicates disjoint union $k$ is the number of occurrences of $X$ in $A$.
Suppose $\cc{\bb t}(\forall X.A)$ is $\phi$-coherent and $\cc{\bb t}(B)$ is $\psi$-coherent, and let $\chi:=\phi|_{\BV{A}}\uplus \biguplus_{i=1}^{k}\psi$; we claim that $\cc{\bb t}(A[B/X])$ is $\chi$-coherent. First, any bound variable $\neq X$ in $A$ is coherent with any bound variable within any of the copies of $B$ in $A[B/X]$. Moreover, 
one can check that if an alternate path starts somewhere in the tree of $A[B/X]$ and enters the sub-tree 
corresponding to a copy of $B$, then the path must end within this sub-tree, since in 
any alternate path starting from (a copy of) $B$ all labels are bound variables of (this copy of) $B$, which are disjoint from all other bound variables.
We deduce then that a cyclic alternate path in the tree of $A[B/X]$ can only be a cyclic alternate path in either $\cc{\mathsf t}(A)$  or in one copy of $\cc{\mathsf t}(B)$. Since both cases are excluded by assumption, no alternate cyclic path exists.
\end{proof}

Thanks to Lemma \ref{rem:compositionality} the following fragments can be seen to be well-defined.

\begin{definition}[Systems $\Nd^{\kappa\leq k}$]


For $k=0,1$, let $\Nd^{\kappa\leq k}$ be the subsystem of $\Nd$ with same typing rules and types restricted to the types of $\Nd$ of characteristic $k$.

\end{definition}

\subsection{Embedding $\NY$ and $\NYY$ into Monomorphic Systems}\label{app5}

The standard embedding $^{\sharp}$ of $\Ntot $ into $\Nd$ recalled in the previous section restricts to an embedding of
$\BB B$ and $\mu \BB B$ in $\NY$ and $\NYY$, respectively. As this embedding preserves $\beta$ and $\eta$-rules up to the $\varepsilon$-theory, it yields then a functor $^{\sharp}: \mu \BB B \to \CTX_{\varepsilon}(\NYY)$, which restricts  to a functor from  $\BB B$ to $\CTX_{\varepsilon}(\NY)$.

We will now define a converse embedding from $\NY$ (resp.~$\NYY$) into $\NI$ (resp.~$\NImu$), yielding a functor $^{\flat}:  \CTX_{\varepsilon}(\NYY)\to \mu \BB B$, restricting to a functor from $\CTX_{\varepsilon}(\NY)$ to $\BB B$.

\paragraph{The embedding of types}

Recall that
from Proposition \ref{prop:kappone} we already know that the types of $\NY$ (resp.~$\NYY$) are isomorphic, modulo the $\varepsilon$-theory, to types of $\NI$ (resp.~$\NImu$). 
%
In fact, since all isomorphisms $\equiv_{\bb Y}$ are valid under the $\varepsilon$-theory, we can obtain $A^{\flat}$ from any standard reduction of the tree of $A$, using Prop. \ref{prop:kappone}. However, 
since our goal is to also embed the terms of $\NY/\NYY$ into $\NI/\NImu$, we cannot define $^{\flat}$ starting from an arbitrary standard reduction, since we need the map $A\mapsto A^{\flat}$ to satisfy two additional properties.
First, it has to commute with substitutions, i.e. $(A[B/X])^{\flat}=A^{\flat}[B^{\flat}/X]$; second, we need the embedding to be compositional, meaning that the translation of a complex type can be reconstructed from the translation of its subtypes.


When $\phi$ is a valuation of $\bb E$ such that $\bb E$ is $\phi$-coherent, we call $\phi$ simply a \emph{coherent valuation of $\bb E$}.
We call a reduction $\bb E_{0}\leadsto_{X_{1}^{c_{1}}}\dots \leadsto_{X_{n}^{c_{n}}}\bb E_{n}$ 
\emph{uniform} if it is standard and satisfies the following two properties:
\begin{itemize}
\item for all $1\leq i\leq n$, no bound variable occurs in $\bb E_{i-1}$ in the scope of the quantified variable $X_{i}$;
%
\item if $G_{i}(X_{i})=G_{j}(X_{j})$, then $c_{i}=c_{j}$.
\end{itemize}
We indicate a uniform reduction as $\bb E\leadsto_{\bb{uni}}\bb E'$. 
In other words, a uniform reduction only eliminates a quantifier after having eliminating all quantifiers in its scope, and when a quantified variable $X$ is copied by the effect of some other reduction, all copies of $X$ are reduced using the same rule.

Any uniform reduction $\bb E_{0}\leadsto_{X_{1}^{c_{1}}}\dots \leadsto_{X_{n}^{c_{n}}}\bb E_{n}\in \C P_{0}$ that converges to a simple polynomial trees induces 
a coherent valuation $\chi $ of $\bb E$, where  
$\chi(X)$ is the unique $c$ such that \emph{for all} $i$ such that $G_{i}(X_{i})=X$, $c_{i}=c$. We call the valuation  $\chi$ the \emph{skeleton} of the reduction. 
Conversely, given a coherent valuation $\chi $of $\bb E$ one can construct a uniform reduction of $\bb E$ of which $\chi$ is the skeleton. 

The following lemma states that the normal form of a uniform reduction is completely determined by its skeleton.

\begin{lemma}\label{lemma:skeleton}
 If two convergent uniform reductions $r_{1}:  \bb E\leadsto_{\bb{uni}} \bb E_{1}$, and 
$r_{2}:\bb E\leadsto_{\bb{uni}}\bb E_{2}$ have the same skeleton $\chi$, then $\bb E_{1}=\bb E_{2}$.
\end{lemma}
\begin{proof}
Let the depth of a variable $ X\in \BV{ \bb E}$ be the number of distinct bound variables within its scope. 
Let us call a uniform reduction \emph{hierarchical} if any variable of depth $k+1$ is eliminated only after all variables of depth $k$ have already been eliminated. 
It is not difficult to check by induction on the length of reductions that (1) any uniform reduction can be transformed into a hierarchical uniform reduction converging onto the same polynomial tree, (2) two hierarchical uniform reductions with same skeleton converge onto the same polynomial tree. The claim then follows from (1) and (2).
%
\end{proof}
We indicate the 
unique result of a converging uniform reduction from $\bb E$ with skeleton $\chi$ simply as $\chi( \bb E)$.%

For all polynomial tree $\bb F$, let $ {\OV {\bb F}}$ be the tree obtained by switching colors on all labels.
For all polynomial trees $\bb E$, $\bb F$ and variable $ X\in \FV{\bb E}$, we let $\bb E[ X\mapsto \bb F]$ be the result of replacing all terminal nodes of $\bb E$ labeled $\cc X$ by $\bb F$ and all terminal nodes of $\bb E$ labeled $\mm X$ by $ {\OV{\bb F}}$.

\begin{lemma}\label{lemma:quasisub}
\begin{itemize}
\item[i.] If $\bb E \leadsto_{\mathsf{uni}} \bb E'$ and $ \bb F\in \C P_{0}$, then $ \bb E[ X\mapsto_{\mathsf{uni}}\bb  F] \leadsto_{ {\mathsf{uni}}} \bb E'[ X\mapsto \bb F]$.

\item[ii.] If $\bb F \leadsto_{\mathsf{uni}} \bb F'$, then $ \bb E[ X\mapsto\bb  F] \leadsto_{\mathsf{uni}} \bb E[ X\mapsto\bb F']$.
\end{itemize}
\end{lemma}

%

Let $\bb E$ be a coherent polynomial tree. We define a coherent valuation $\phi_{\bb E}: \BV{ \bb E}\to  \mathsf{Colors}$ using  the well-known \emph{strongly connected components algorithm for 2-SAT} \cite{2SAT} (see Remark \ref{rem:2sat}).
We recall that a direct graph $G=(V,E)$ is \emph{skew symmetric} if there exists an involution $^{*}:V\to V$ such that for all $v,w\in V$, $v\to w\in E$ iff $w^{*}\to v^{*}\in E$. 

Let $G=(|G|,\|G\|)$ be the directed skew-symmetric graph with vertices the colored variables $ X_{i}^{c}$, for all $X_{i}\in \BV{\bb E}$ and $c\in \mathsf{Colors}$, symmetry $(X_{i}^{c})^{*}=X_{i}^{\OV c}$, and with edges 
$X_{i}^{\OV c}\to X_{i}^{c}$ if $X_{i}$ is not $\OV c$-eliminable and 
$ X_{i}^{c}\to   Y_{j}^{d}$, $  Y_{j}^{\OV d}\to X_{i}^{\OV c}$ if 
$ X_{i}^{c}$ is incoherent with $  Y_{i}^{\OV d}$.  
Let $X_{i}^{c}\sim Y_{j}^{d}$ if there is a directed path in $G$ from $X_{i}^{c}$ to $Y_{j}^{d}$ and a directed path from $Y_{j}^{d}$ to $X_{i}^{c}$. The equivalence classes of $\sim$ are usually called \emph{strong connected components} of $G$. Observe that a strong connected component cannot contain both $X_{i}^{c}$ and $X_{i}^{\OV c}$ (otherwise $\bb E$ would not be coherent). 
Let $G/\sim =(|G/\sim|, \| G/\sim\|)$ be the directed graph with vertices the strong connected components of $G$ and an edge $a\to b$ if there is $X_{i}^{c}\in a$ and $Y_{j}^{d}\in b$ and an edge $X_{i}^{c}\to Y_{j}^{d}$ in $G$. $G/\sim$ is then a directed acyclic graph, inducing thus an order relation $\prec$ on $|G/\sim|$. 
We now define a map $\chi: \BV{ \bb E}\to C$ as follows: let us split $|G/\sim|$ in its connected components $G_{1},\dots, G_{n}$; for each component $G_{i}$, choose a linear ordering $a_{1}\prec_{i}\dots \prec_{i}a_{k_{i}}$ of its vertices, compatible with the order $\prec$. We define a partial function $\phi_{i}:\BV{ \bb E}\rightharpoonup  \bb{Colors}$ by induction on $\succ_{i}$: 
\begin{itemize}
\item for all $X_{j}^{c}\in a_{k_{i}}$, set $\phi_{i}(X_{j})=c$;
\item for all $X_{j}^{c}\in a_{k_{i}-l}$, if $\phi_{i}(X_{j})$ is not yet defined, set $\phi_{i}(X_{j})=c$.

\end{itemize}
Finally let $\phi_{\bb E}= \bigcup_{i=1}^{n}\phi_{i}$.
One can check then that $\phi_{\bb E}$ is a coherent valuation of $\bb E$.
%
%
%
%
%


\begin{lemma}\label{lemma:subst}
If $\bb E$ and $\bb F$ are coherent and $ X\in \FV{\bb E}$, then $\bb E[ X\mapsto \bb F]$ is coherent and 
$\phi_{\bb E[X\mapsto \bb F]}(\bb E[ X\mapsto \bb F])= \phi_{\bb E}(\bb E)[ X\mapsto \phi_{\bb F}(\bb F)]$.
\end{lemma}
\begin{proof}
We have that $\BV{\bb E[ X\mapsto \bb F]}=\BV{\bb E} \uplus \biguplus_{i=1}^{k} \BV{\bb F}$, where $ X$ occurs $k$ times in $\bb E$. 
Observe that in $\bb E[ X\mapsto \bb F]$ any bound variable coming from a copy of $\bb F$ is coherent with any bound variables coming from $\bb E$. Hence, if we construct the graph $G/\sim$ of the strongly connected components, the variables from $\bb E$ and those from the copies of $\bb F$ belong to different connected components of $G/\sim$.
 Thus, since by construction the value $\phi_{\bb E[ X\mapsto  \bb F]}(  Y)$ only depends on the connected component of $  Y$ in $G/\sim$, we deduce that
 $\phi_{\bb E[ X\mapsto \bb F]}= \phi_{\bb E}\uplus \biguplus_{i=1}^{k}\phi_{ \bb F}$.

Using this fact along with Lemma \ref{lemma:quasisub} we deduce the existence of two uniform reductions  
$ \bb E [ X\mapsto  \bb F] \leadsto_{ {\mathsf{uni}}}  \bb E[ X\mapsto \phi_{\bb F}( \bb F)] \leadsto_{ {\mathsf{uni}}}\phi_{\bb E}(\bb E)[ X\mapsto \phi_{\bb F}(\bb F)]$ and 
$ \bb E[ X\mapsto \bb F] \leadsto_{ {\mathsf{uni}}} \phi_{\bb E[X\mapsto \bb F]}( \bb E[ X\mapsto \bb F])$
with the same skeleton, so using Lemma \ref{lemma:skeleton} we deduce our claim.
\end{proof}

We can now define $A^{\flat}$ as follows:
\begin{definition}
For all $A\in \NYY$, we let $A^{\flat}:= \tau( \phi_{\cc{\bb t}(A)}(\cc{\bb t} (A)))$.
\end{definition}
From Lemma \ref{lemma:subst} we immediately deduce that $(A[B/X])^{\flat}=A^{\flat}[B^{\flat}/X]$.
Moreover, we can check  that if $A=B\To C$, then $A^{\flat}=B^{\flat}\To C^{\flat}$ and that, if $A=\forall X.C$, then since $A^{\flat}$ is obtained by uniform reduction, we can suppose 
that $\cc{\bb t}(A)\leadsto_{\mathsf{\mathsf{uni}}} \cc{\bb E}' \leadsto_{\cc X} \cc{\bb E}''$, where $\tau(\cc{\bb E}'')=A^{\flat}$ and $\tau(\cc{\bb E}')=\forall X.C^{\flat}$, and since $X$ is eliminable in $\cc{\bb E}'$, $C^{\flat}$ must have one of forms below 
\begin{equation}\label{eqqi}
\Big \langle   \big\langle \FFun{\cc X}{A^{\flat}_{jk}}\big\rangle_{j}\To \mm X\Big\rangle_{k}\To \FFun{\cc X}{B^{\flat}} 
\qquad \qquad 
\Big \langle  \cc X\To \FFun{\mm X}{A^{\flat}_{j}}\Big\rangle_{j}\To \FFun{\mm X}{B^{\flat}} 
\end{equation}
for given types $A_{jk},A_{j},B$ which are strictly less complex than $A$.

\paragraph{The Embedding of Terms}

Using the properties of the embedding of type just defined as well as the terms for the Yoneda type isomorphisms from App.~\ref{app4}, we can define the embedding of terms from $\NYY$ to $\NImu$.

Given $\Gamma\vdash_{\NYY}t:A$ (resp. $\Gamma \vdash_{\NY}t:A$), we define a term $t^{\flat}$ such that 
$\Gamma^{\flat}\vdash_{\NImu}t^{\flat}:A^{\flat}$ (resp. $\Gamma^{\flat}\vdash_{\NI}t^{\flat}:A^{\flat}$) holds  by induction on $t$. 
If $t=x$, we let $t^{\flat}=x$, if $t=\lambda x.u$, we let $(t)^{\flat}= \lambda x.t^{\flat}$ and if $t=uv$, we let $(t)^{\flat}=u^{\flat}v^{\flat}$. The case of $t=\Lambda X.u$ and $t=uD$ are less obvious:
\begin{itemize}

\item if $t=\Lambda X.u$, then $A=\forall X.C$ and by the induction hypothesis $\Gamma^{\flat}\vdash u^{\flat}:C^{\flat}$. Moreover, we can suppose that $C^{\flat}$ is of one of the forms of Eq. \eqref{eqqi}. Hence we must consider two cases:
\begin{itemize}
\item[a.] $A^{\flat}= \FFun{\cc X \mapsto \mu  X.\sum_{k}\prod_{j}\FFun{\cc X}{A_{jk}^{\flat}}}{B^{\flat}}$. Then we let
$t^{\flat}= \TT s_{{A_{jk}^{\flat}}, {B^{\flat}}} \Big[ u^{\flat} [ X\mapsto \mu  X.\sum_{k}\prod_{j}\FFun{\cc X}{A_{jk}^{\flat}}  ]       \Big]$.

\item[b.] $A^{\flat}= \FFun{\mm X \mapsto \nu  X.\prod_{j}\FFun{\mm X}{A_{j}^{\flat}}}{B^{\flat}}$. Then we let
$t^{\flat}= \TT u_{{A_{j}^{\flat}}, {B^{\flat}}} \Big[ u^{\flat} [ X\mapsto \nu  X.\prod_{j}\FFun{\cc X}{A_{j}^{\flat}}  ]       \Big]$.

\end{itemize}
where for the terms $\bb s_{A_{jk}^{\flat},B^{\flat}}$ and $\bb u_{A_{j}^{\flat},B^{\flat}}$  see App.~\ref{app4}.

\item if $t=uD$, then $A=C[D/X]$ and by the induction hypothesis we have $\Gamma^{\flat}\vdash_{\NYY}u: (\forall X.C)^{\flat}$. Moreover we have $A^{\flat}= C^{\flat}[D^{\flat}/X]$ and as above we can suppose that $C^{\flat}$ is in one of the forms in Eq. \eqref{eqqi}, so we consider again two cases: 
	\begin{itemize}
	\item[a.] If $(\forall X.C)^{\flat}=\FFun{\cc X \mapsto \mu  X.\sum_{i}\prod_{j}\FFun{\cc X}{A_{jk}^{\flat}}}{B^{\flat}}$, then we let $t^{\flat}=  \bb t_{ A_{jk}^{\flat},B^{\flat}}[u^{\flat}, D^{\flat}]$.

	\item[b.]  If $(\forall X.C)^{\flat}=\FFun{\mm X \mapsto \nu  X.\prod_{j}\FFun{\mm X}{A_{j}^{\flat}}}{B^{\flat}}$, then we  let  $t^{\flat}=  \TT v_{A_{j}^{\flat},B^{\flat}}[u^{\flat},  D^{\flat}]$.
	\end{itemize}
where for the terms $\TT t_{A_{jk}^{\flat},B^{\flat}}$ and $\bb v_{A_{j}^{\flat},B^{\flat}}$  see App.~\ref{app4}.

\end{itemize}

Using the constructions above it is a simple exercise to define explicit isomorphisms $(\rr{\TT d}_{A}[x], \TT d_{A}^{-1}[x]):A\vdash A^{\flat}$ for all type $A$ of $\NYY$.

\subsection{The equivalence of $\NY/\NYY$ with $\BB B/\mu\BB B$}

The two functors $
^{\sharp} $ and 
$^{\flat}$ preserve all the relevant structure (products, coproducts, exponentials, initial algebras, final coalgebras), 
but they are 
not \emph{strictly} inverse: $(A^{\flat})^{\sharp}$ is not equal to $A$, but only $\varepsilon$-isomorphic to it (e.g. for $A=\forall X.((\forall Y.\mm Y\To\cc Y)\To \mm X)\To\cc  X$, we have $A^{\flat}=1$ and $(A^{\flat})^{\sharp}= \forall X.\mm X\To \cc X$). 
Nevertheless, we will establish that the following equivalences hold:

\begin{restatable}{theorem}{equiva}\label{thm:equiva}
\begin{enumerate}
\item[i.] $\CTX_{\varepsilon}(\NY)\cong \BB B$.
\item[ii.] $\CTX_{\varepsilon}(\NYY)     \cong  \mu \BB B$.


\end{enumerate}
\end{restatable}

%
%
The proof of Theorem \ref{thm:equiva} is done by checking (by way of $\beta$-, $\eta$- and \emph{$\varepsilon$-rules}) that both $\NYY$ and $\NImu$ embed \emph{fully} in $\Nnew$, using the lemma below (with $\BB C=\CTX_{\varepsilon}(\NYY)$, $\BB D=\mu \BB B$, $\BB E=\Nnew$ and  $f(A)=A^{\flat}$).

\begin{restatable}{lemma}{fullsub}\label{lemma:fullsub}
Let $\BB C, \BB D$ be full subcategories of a category $\BB E$. Let $f: \mathsf{Ob}(\BB C)\to \mathsf{Ob}(\BB D)$ be surjective and suppose there is a map $u$ associating each object $a$ of $\BB C$ with an isomorphism $u_{a}:a\to f(a)$ in $\BB E$. Then $f$ extends to an equivalence of categories $F: \BB C\to \BB D$.

\end{restatable}
\begin{proof}
Let $F(a)=f(a)$ and $F(g:a\to b)= u_{b}  \circ g\circ u_{a}^{-1}\in   \BB E(f(a),f(b))=  \BB D(f(a),f(b))$. $F$ is clearly faithful and surjective. It is also full since for any $h\in   \BB D( F(a), F(b))$ there exists $k=u^{-1}_{b}\circ h\circ u_{a}\in  \BB\bb E(a,b)= \BB C(a,b)$ such that $h=F(k)$.
\end{proof}%

Let $ \BB C=\CTX_{\varepsilon}(\NYY)$, $ \BB D=\mu  \BB B$ and let $\BB E=\CTX_{\varepsilon}(\Nnew)$. We know that there is a surjective map $^{\flat}: \mathsf{Ob}( \BB C)\to \mathsf{Ob}(  \BB D)$ as well as isomorphisms $(\TT d_{A}[x], \TT d_{A}^{-1}[x]):A\to A^{\flat}$ in $ \BB E$. 
Thus, to apply Lemma \ref{lemma:fullsub} it remains to check  that $\CTX_{\varepsilon}(\NYY)$ and $\mu  \BB B$ are \emph{full} subcategories of $\CTX_{\varepsilon}(\Ntot)$. 
Concretely, this means checking that:
\begin{itemize}
\item for any term $\Gamma\vdash_{\Nnew}t:A$ in which $\Gamma,A$ are $\NYY$-types
  we can find a term $u\simeq_{\varepsilon}t$ such that $\Gamma\vdash_{\NYY} u:A$ holds; 
\item  for any term $\Gamma\vdash_{\Nnew}t:A$ in which $\Gamma,A$ are types of $\mu \BB B$ 
we can find a term $u\simeq_{\varepsilon}t$ such that $\Gamma\vdash_{\NImu} u:A$
 holds.
\end{itemize}
This is where we will actually use our embeddings, since we take $u=t^{\sharp}$ in the first case and $u=t^{\flat}$ in the second case. The verification of $t\simeq_{\varepsilon}u$ involves both $\varepsilon$- and $\beta\eta$-equations and is postponed to App.~\ref{app8}.

\subsection{Program Equivalence in $\NY$ is decidable.}

Theorem \ref{thm:equiva} can be used to deduce properties of program equivalence in $\NY$ and $\NYY$ from well-known properties of program equivalence for $\BB B$ and $\mu \BB B$. In fact, it is known that $\beta\eta$-equivalence in $\NI$ (i.e. arrow equivalence in $\BB B$) is decidable and coincides with {contextual equivalence} \cite{Scherer2017}, while contextual equivalence for $\mu \BB B$ is undecidable \cite{Basold2016}. 
Using Theorem \ref{thm:equiva} we obtain similar facts for $\NY$ and $\NYY$:
\begin{theorem}\label{thm:decidable}
The $\varepsilon$-theory for $\NY$ is decidable and coincides with contextual equivalence.

\end{theorem}

The first claim of Theorem \ref{thm:decidable} is proved by defining a new embedding $t\mapsto t^{\natural}$ of $\NY$ into $\NI$, exploiting the well-known fact that for the terms of the system $\Ndv$ one can obtain a normal form under $\beta$-reduction and \emph{commutative conversions} \cite{Tatsuta2005}\footnote{Actually,  \cite{Tatsuta2005} does not consider commuting conversions for $0$, but these can be added without altering the existence of normal forms.}. 
Using the isomorphisms $\TT d_{A}[x], \TT d_{A}^{-1}[x]$ between $A$ and $A^{\flat}$, a term $t$ such that $\Gamma\vdash_{\NY}t:A$ holds is first translated into 
$u=\TT d_{A}[  t[ x_{i}\mapsto \TT d_{A_{i}}^{-1}[x_{i}]]]$ (where $\Gamma=x_{1}:A_{1},\dots, x_{n}:A_{n}$), and then $t^{\natural}$ is defined as the normal form of $u$. From the fact that $\Gamma^{\flat}\vdash_{\Ndv}t^{\natural}:A^{\flat}$ and that $t^{\natural}$ is in normal form, we can deduce that 
$\Gamma^{\flat}\vdash_{\NI} t^{\natural}: A^{\flat}$ holds, so the embedding is well-defined.

Using the embedding $t\mapsto t^{\natural}$ we prove the proposition below, from which the claim of Theorem \ref{thm:decidable} descends (using the fact that $\simeq_{\beta\eta}$ and $\simeq_{\ctx}$ coincide and are decidable in $\NI$). 
%
%
%
%
\begin{restatable}{proposition}{contextual}\label{prop:equivalents}
The following are equivalent:
\begin{itemize}
\item[i.] 
 $\Gamma\vdash_{\NY} t\simeq_{\varepsilon} u:A$;
%
%
\item[ii.] $\Gamma\vdash_{\NY} t\simeq_{\ctx} u:A$.
\item[iii.] $\Gamma^{\flat}\vdash_{\NI}t^{\natural}\simeq_{\beta\eta}u^{\natural}:A^{\flat}$.
\end{itemize}
\end{restatable}
\begin{proof}
\begin{description}
\item[(i.$\To$ ii.)] Since $\simeq_{\ctx}$ is the maximal consistent theory, it is clear that $\simeq_{\varepsilon}\subseteq \simeq_{\ctx}$.

\item[(ii.$\To$ iii.)] 

Suppose $\Gamma^{\flat}\vdash_{\NI} t^{\natural}\not\simeq_{\beta\eta}u^{\natural}:A^{\natural}$. Since $\simeq_{\beta\eta}$ and $\simeq_{\ctx}$ coincide in $\NI$, there exists a context $\TT C: (\Gamma^{\flat}\vdash A^{\flat})\To (\vdash 1+1)$ such that $\TT C[t^{\natural}]\not\simeq_{\beta\eta}\TT C[u^{\natural}]$. Observe that this also implies $\TT d_{B}^{-1}[\TT C[t^{\natural}]]\not\simeq_{\beta\eta}\TT d_{B}^{-1}[\TT C[u^{\natural}]]$, where $B=\forall X.X\To X \To X$.

Let $\TT F: (\Gamma\vdash A)\to (\Gamma^{\flat}\vdash A^{\flat})$ be 
$(\lambda x_{1}.\dots. \lambda x_{k}. \TT d_{B}[ x ])\TT d_{A_{1}}^{-1}[x_{1}]\dots \TT d_{A_{k}}^{-1}[x_{k}]$, where $\Gamma=\{x_{1}:A_{1},\dots, x_{k}:A_{k}\}$. Then $\TT F[t]\simeq_{\beta\eta}t^{\natural}$ and $\TT F[u]\simeq_{\beta\eta}u^{\natural}$.

Let now $\TT E:(\Gamma\vdash_{\Ndv} A)\to (\vdash_{\Ndv}\forall X.X\To X\To X)$ be the context $\TT d_{B}^{-1}\circ\TT C\circ\TT F$, so that $\TT E^{\sharp}:(\Gamma\vdash_{\NY} A)\to (\vdash_{\NY}\forall X.X\To X\To X)$. We also have $\TT E\simeq_{\varepsilon}\TT E^{\sharp}$ by Th. \ref{thm:full} in App.~\ref{app8}. We now have
$\TT E^{\sharp}[t] \simeq_{\varepsilon} \TT E[t] \simeq_{\beta\eta} \TT d_{B}^{-1}[ \TT C[t^{\natural}]] \not\simeq_{\beta\eta}
\TT d_{B}^{-1}[\TT C[u^{\natural}]] \simeq_{\beta\eta} \TT E[u]\simeq_{\varepsilon}\TT E^{\sharp}[t]$, so we can conclude that 
$t \not \simeq_{\ctx} u$.

\item[(iii.$\To$ i.)] This immediately follows from the fact $t^{\natural},u^{\natural}$ are obtained from $t,u$ by applying isomorphisms relative to $\varepsilon$-equivalence.
\end{description}
\end{proof}


\subsection{Program Equivalence in $\NYY$ is undecidable.}

It is well-known that contextual equivalence is undecidable in $\mu\BB B$ \cite{Basold2016}. From this fact we can easily deduce, using the encoding $^{\sharp}$, that contextual equivalence is also undecidable in $\NYY$.

We do not know if the $\varepsilon$-theory and contextual equivalence coincide in $\NYY$, as it is not clear whether the embedding $t\mapsto t^{\natural }$ scales to $\NYY$: this depends on the existence of normal forms for commutative conversions, which are not known to hold in presence of $\mu,\nu$-types (although this is conjectured in \cite{MatthesPhD}). 

%
%
%
%
%
%

For this reason, we think it's worth to provide a direct argument for the undecidability of the $\varepsilon$-theory in $\NYY$. 
The correspondence between universal types in $\NYY$ and $\mu,\nu$-types provides a way to compute $\varepsilon$-equivalence in an indirect way. In fact, the $\varepsilon$-rule for a type $\forall X.A$ of $\NYY$ can be interpreted as a principle expressing the \emph{uniqueness} of a function defined by an inductive/co-inductive principle. 
For example, in the case of the type $\nat=\forall X.(\cc X\To \mm X)\To (\mm X\To \cc X)$, from the isomorphism
$\nat\equiv_{\cc X} \mu X(X+1)$ we deduce the validity of the  $\varepsilon$-rule in Fig. \ref{fig:nat}.
Let $C=\nat$, $u= \sss$, $D$  (the successor function) and $\rr{\TT c}$ be given by some function $f: \nat \To A$. Then the $\varepsilon$-rule yields the \emph{uniqueness rule} in Fig. \ref{fig:uniqnat}, which expresses the uniqueness of functions defined by iterations.

\begin{figure}[t]
\begin{subfigure}{\textwidth}
\adjustbox{scale=0.8, center}
{$
\AXC{$ \Gamma\vdash u: C\To C  $} 
\AXC{$ \Gamma\vdash v: D\To D $}
\AXC{$t[x]: C\vdash^{\Gamma} D$}
\AXC{$t[ ux] \simeq v\rr{\TT c} : C\vdash^{\Gamma}D$}
\QuaternaryInfC{$
\Gamma, z:\nat, y:C\vdash t[ zCuy] \simeq zD vt[y] : D
 $}
\DP
$}
\caption{$\varepsilon$-rule for $\nat$.}
\label{fig:nat}
\end{subfigure}

\medskip

\begin{subfigure}{\textwidth}
\adjustbox{scale=0.8, center}{
$
\AXC{$\Gamma\vdash f: \nat \To D $}
\AXC{$\Gamma\vdash v: D \To D$}
\AXC{$\Gamma, x:\nat\vdash f(\sss x) \simeq_{\varepsilon} v(fx):D$}
\TrinaryInfC{$ \Gamma,x:\nat \vdash fx \simeq_{\varepsilon} xD v (f\mathsf 0):D$}
\DP 
$}
\caption{Uniqueness rule for $\nat$.}
\label{fig:uniqnat}
\end{subfigure}
\caption{The $\varepsilon$-rules for $\nat$ yields a unicity condition for iteratively-defined functions.}
\end{figure}

\begin{figure}[t]
\adjustbox{scale=0.8, center}
{$
\AXC{$ \Gamma\vdash u: \nat \To Y \To A\To A  $} 
\AXC{$ \Gamma\vdash v: \nat \To Y \To  B\To B $}
\AXC{$\rr{t}[x]: A\vdash^{\Gamma} B$}
\AXC{$\rr{t}[ uyzx] \simeq vyz\rr{t} : A\vdash^{\Gamma, y:\nat, z:Y}B$}
\QuaternaryInfC{$
\Gamma, x:\REC,y:\nat, z:Y, w:A\vdash {t}[ xAuyzw] \simeq xB v yz {t}[w] : B
 $}
\DP
$}
\caption{$\varepsilon$-rule for $\REC$.}
\label{fig:rec}
\end{figure}

We will exploit this correspondence to establish the undecidability of the equivalence generated by $\varepsilon$-rules from the result below.

\begin{theorem}[Okada, Scott]\label{th:okada}
Let $\mathsf S$ be a system containing the type $\nat$ of integers (with terms $\zz:\nat$ and  $\sss:\nat\To \nat$), closed by arrow types and having, for all type $A,B$ a recursor $\rec_{A,B}: (\nat \To A\To B\To B)\To (A\To B)\To (A\To \nat\To B)$ satisfying
\begin{equation}
\rec_{A,B} h u a \zz= ua \qquad 
\rec_{A,B}h u a (\sss x)= h x a (\rec_{A,B} h u a x)
\end{equation}
Let $\simeq_{\gamma}$ be an equivalence over $\mathsf S$ containing $\simeq_{\beta},\simeq_{\eta}$ and closed by  the following rule:
\begin{equation}\label{ruleU}
\AXC{$f x(\mathsf S y) \simeq h xy(fxy)$}
\UIC{$f \simeq \rec_{A,B} h  (\lambda x.fx0)  $}
\DP \tag{U}
\end{equation}
for all $f: A\To \nat\To B$ and $h: \nat\To A\To B\To B$.
Then $\simeq_{\gamma}$ is undecidable.
\end{theorem}

Let us first show that we can restrict to a weaker rule:

\begin{lemma}\label{lemma:UUU}
The rule \eqref{ruleU} is deducible from the rule
\begin{equation}\label{wU}
\AXC{$f x(\mathsf S y) \simeq h yx(fxy)$}
\UIC{$f ay \simeq \rec_{A,B} h  (\lambda x.fa0) ay $}
\DP\tag{wU}
\end{equation}

\end{lemma}
\begin{proof}
Let $Q = \rec_{A,B}h(\lambda x.fx\mathsf 0) $. Since 
$Q a (\mathsf Sy) \simeq
h ya( Q a y)$  and $Qa\mathsf 0\simeq fa\mathsf 0$,  by \eqref{wU} we deduce
$Q a y \simeq \rec_{A,B}h(\lambda x.Qa0) a y \simeq
\rec_{A,B}h(\lambda x.fa\mathsf 0)ay$. 

Now, if $f x(\mathsf S y) \simeq h xy(fxy)$ holds, by \eqref{wU} we deduce 
$fay\simeq \rec_{A,B}h(\lambda x.fa0) a y $ and by composing this with the equivalence above we finally get 
$fay \simeq \rec_{A,B}h(\lambda x.fx0) a y$. 
\end{proof}

We now establish some preliminary properties of $\varepsilon$-equivalence over the type $\nat$. For simplicity, we  will work in 
an extension of $\NYY$ with product types.
We let $\mathsf 0=\Lambda X.\lambda fx.x$ and $\mathsf  s= \lambda y. \Lambda X.\lambda fx.f(yXfx)$, $\mathsf P=\lambda x. \pi_{2}\left ( x (\nat\times \nat) \lambda y.\langle \sss (\pi_{1}y), \pi_{1}y\rangle \langle \zz, \zz \rangle\right)$.

\begin{lemma}\label{lemma:predecessor}
\begin{itemize}
\item $x:\nat  \vdash x \simeq_{\varepsilon} x(\nat)  ( \sss )  \mathsf 0$.
\item $\mathsf P( \sss x)\simeq_{\varepsilon}x$.
\end{itemize}
\end{lemma}
\begin{proof}
We prove both results as a consequence of unicity of iteration. If $Fx=x$, then $F\zz\simeq_{\beta}\zz (\nat)( \sss) \zz$ and $F(\sss x)= \sss x= \sss (Fx)$, whence by unicity $x=Fx\simeq_{\varepsilon} x (\nat )(\sss) \zz$. 

Let $\rr{\TT c}[x]:=x (\nat\times \nat) \lambda y.\langle \sss (\pi_{1}y), \pi_{1}y\rangle \langle \zz, \zz \rangle$.
We have $\mathsf P(\sss \zz)=\zz$ and, since $\mathsf P(\sss x)\simeq_{\beta} \pi_{1}\rr{\TT c}[x]$,
and $\pi_{1}\rr{\TT c}[\sss x]\simeq_{\beta} \sss (\pi_{1}\rr{\TT c}[x])$, we have 
 $\mathsf P(\sss(\sss x))\simeq_{\beta} \sss (\pi_{1} \rr{\TT c}[x])
\simeq_{\beta}\sss (\mathsf P(\sss x))$, so by unicity we conclude
$\mathsf P(\sss x)\simeq_{\varepsilon} x (\nat)( \sss) \zz \simeq_{\varepsilon}x$. 
\end{proof}
	
To apply Theorem \ref{th:okada} we will exploit the type $\REC=\forall X. (\nat\To Y \To X\To X)\To (Y\To  X) \To Y\To \nat \To X$. Observe that $\REC$ is isomorphic to the non-polymorphic type $(\mu X. (\nat \times Y\times X)+Y)\times \nat \times Y$. 
 Using Lemma \ref{lemma:predecessor}, it can be checked by induction that  the term $\rec: \forall Y.\REC$ given by
$\rec  (X)( Y) h x a y = y (X)(   \lambda z. h (\mathsf{P}y)a z)    (xa)  $
yields, for all types $A,B$ of $\NYY$, a recursor $\rec (A)(B)$ in $\Nd_{\kappa\leq 3}$ under $\simeq_{\varepsilon}$.


%
%
	
%


\begin{proposition}
The rule \eqref{ruleU} holds in $\NYY$ under $\simeq_{\varepsilon}$ with $\rec$ as recursor.

\end{proposition}
\begin{proof}
Let $\mathsf Z: \nat \To A \To \nat \To \nat$ be the term $\mathsf Z xyz= \mathsf sz$. We claim that 
$\rec(A)( \nat)( \mathsf  Z )(\lambda x.0 ) a y \simeq_{\varepsilon} y$.  In fact
$\rec (A) (\nat )(\mathsf  Z )(\lambda x.0 ) a y \simeq_{\beta}
y (\nat)( \lambda z. \mathsf Z (\mathsf Py) a z) \mathsf 0 \simeq_{\beta\eta} y( \nat)( \mathsf s) \mathsf 0\simeq_{\varepsilon}y$, by Lemma \ref{lemma:predecessor}.
Let now $f: A\To \nat \To B$. If $ f a( \mathsf Z x ay)  \simeq_{\varepsilon} h x a (f ax)$, we deduce, by the $\varepsilon$-rule in Fig. \ref{fig:rec} with $t[x]= fax: \nat \vdash B$, that $
f a (\rec (A)( \nat)( \mathsf Z)( \lambda  x.0 )a y) \simeq_{\varepsilon}
\rec (A)( B)( h)(  \lambda x. fa0)  \ ay
$. By what we just proved this implication corresponds to the rule \eqref{wU}, so we can conclude by Lemma \ref{lemma:UUU}.
\end{proof}

We can thus finally conclude:

\begin{theorem}\label{thm:undecidable}
Both the $\varepsilon$-theory and contextual equivalence for $\NYY$ are undecidable.

\end{theorem}
%
%



%
%


\section{Program Equivalence and Predicativity}\label{secAtomic}

%
In this section we sketch an example of how the correspondence between the polymorphic system $\NYY$ and the monomorphic system $\mu\BB B$ can be used to prove non-trivial properties of program equivalence in $\Nd$.

A main source of difficulty to compute program equivalence of polymorphic programs is that, while these can be instantiated \emph{at any type}, programs with different type instantiations can behave in the same way. This is made evident by the $\varepsilon$-rules from Fig.~\ref{fig:dinazza} and Fig.~\ref{fig:dinazzabis}, which allow one to permute an instantiation on a type $A$ into one on any type $B$, provided there is a suitable term $\rr{t}[x]:A\vdash^{\Gamma} B$.

A natural question is thus whether one can use such rules to turn a program into one with instantiations of lower complexity. This cannot be possible in all cases since the so-called \emph{predicative} fragments of $\Nd$, that is the subsystems of $\Nd$ in which type instantiations must satisfy some complexity bound, are known to be strictly less expressive than $\Nd$ (see for instance \cite{Leivant1991}). However, 
using the correspondence between universal types in $\NYY$ and \emph{initial algebras}, 
we will deduce a useful \emph{sufficient condition} under which this simplifications can be put forward in $\NYY$.

After presenting this condition, we will use it to provide a quick proof based on type isomorphisms of a result that was established in \cite{SL2} by the same authors by computing $\varepsilon$-rules directly: 
any polymorphic program in a certain fragment of $\Nd$ (in fact, a fragment of $\NY$) arising from the encoding of finite sum and product types, can be transformed into one containing only \emph{atomic} instantiations. 

%

We first illustrate our idea with an example.

\begin{example}[pointwise induction]\label{ex:indu}
Suppose we recursively define a function $t: \nat \To (A \To B)$ by ``$\nat$-induction on $A\To B$'': we let $t(0)$ be some fixed function $u: A\To B$ and we let 
$t({n+1})$ be the function $x\mapsto  h((t(n)(x))$, for some $h: B\To B$. In fact, since the variable $x:A$ is used as a parameter, we could equivalently define $t$ by  ``$\nat$-induction on $B$'': we let $t_{x}(0)$ be some fixed $c=u(x)$, and $t_{x}(n+1)$ be $h(t_{x}(n))$, so that $t(n)(x)=t_{x}(n)$.

\end{example}

A bit more formally, the argument from Example \ref{ex:indu} shows the equivalence of the two contexts
$\rr{\TT C_{1}}= [\ ] (A\To B)( \lambda fx.h(fx))( u)$ and $\rr{\TT C_{2}}=\lambda x.[\ ] (B)( h)    (ux)$ of type $\nat\vdash A\To B$. Observe that these two contexts are \emph{not} $\beta\eta$-equivalence and, in particular, that they differ for the fact that the second contains a type instantiation which is strictly less complex than the one of the first. 

We can formalize the notion of being ``less complex'' as follows:
given a type $A$ of $\Nd$, we let the set $\TT{RS}(A)$ of \emph{right simplifications of $A$} be inductively as follows:
$$
\TT{RS}(X)=\emptyset \qquad \TT{RS}(A\To B)=\{B\}\cup \TT{RS}(B) \qquad \TT{RS}(\forall Y.B)=\{B\}\cup \TT{RS}(B)
$$
Moreover, for all type $A$, we let $\at{A} \in\TT{RS}(A)$ indicate the rightmost variable of $A$.

The reasoning from Example \ref{ex:indu} can be justified at a more abstract level by exploiting the isomorphism $\nat \equiv_{\cc X} \mu X.\cc X+1$, which shows that $\nat$ is isomorphic to the \emph{initial $X+1$-algebra}.
In fact, by appealing to initial algebras we can justify similar arguments for a large class of inductive types like $\nat$, that we define as follows (essentially following a terminology from \cite{SL2}): let us call a type $P[X]$ \emph{polynomial} when it is of the form 
$$
 P[X]:=   \Big \langle  \big \langle \FFun{\cc X}{A_{jk}} \rangle_{j}\To X\Big \rangle_{k}\To X
$$
and let us call $P[X]$ \emph{constant} if $X$ does not occur in any of the $A_{jk}$.
This terminology is justified by the isomorphism
\begin{equation}\label{eq:initia}
\forall X.P[X]  \ \equiv_{\cc X}\ \{ \mu X.\} \sum_{k}\prod_{j} \FFun{\cc X}{A_{jk}}
\end{equation}
showing that the universal type $\forall X.P[X]$ coincides (under the $\varepsilon$-theory) with the initial $\sum_{k}\prod_{j} \FFun{\cc X}{A_{jk}}$-algebra. Note that the fixpoint binder $\mu X.$ occurs precisely when $P[X]$ is non-constant.

Let a \emph{$P[X]$-algebra} be a triple $(\Gamma,C, \langle u_{k}\rangle_{k})$ made of a context $\Gamma$, a type $C$ and terms $u_{k}:  \langle x_{j}: \FFun{\mm X\mapsto C}{A_{jk}} \rangle_{j} \vdash^{\Gamma}  C$, and a \emph{$P[X]$-algebra morphism} from $(\Gamma,C, \langle u_{k}\rangle_{k})$ to $(\Delta, D, \langle v_{k}\rangle_{k})$ be an inclusion $\Gamma\subseteq \Delta$\footnote{We consider here context inclusions for simplicity but one might rather consider a context morphism.} and a term $ w[x]: C\vdash^{\Delta} D$ such that 
$$
\left \langle  w[ u_{k}[x]] \ \simeq_{\beta\eta} \  v_{k}\Big [\langle \Phi^{A_{jk}}_{X}( w[ x_{j}]) \rangle_{j}\Big ] : 
\langle \FFun{\mm X\mapsto C}{A_{jk}}\rangle_{j} \vdash^{\Delta} D \right \rangle_{k}
$$
The pair 
$(\forall X.P, \langle \bb t^{P}_{k}\rangle_{k})$, where $\bb t^{P}_{k}[\langle x_{j}\rangle_{j}]$ is the term
$$
 \Lambda X.\lambda \langle f_{k} \rangle_{k}.f_{k} \Big \langle 
\Phi^{X}_{A_{jk}}(x X\langle f_{k}\rangle_{k} )\big[x\mapsto x_{j}\big ]\Big \rangle_{j} 
\ : \  \langle \FFun{\mm X\mapsto \forall X.P[X]}{A_{jk}}\rangle_{j} \vdash \forall X.P[X]
$$
 is the \emph{initial $P[X]$-algebra}, that is, the initial object in the category of $P[X]$-algebras and $P[X]$-algebra morphisms. This means in particular, that for all $P[X]$-algebra $(C, \langle u_{k}\rangle_{k})$, 
one can define a context
$$
\TT{Ind}_{C}^{P}(\langle  u_{k}\rangle_{k})= [\ ] C\langle \lambda \langle x_{j}\rangle. u_{k}[\langle x_{j}\rangle_{j}]\rangle_{k}
\ : \ \forall X.P[X] \vdash C
$$
called \emph{$P[X]$-induction on $C$}, which is a $P[X]$-algebra morphism which satisfies the following universal property: for all $P[X]$-algebra morphism $ v[x]$
 from the initial $P[X]$-algebra to $(C, \langle u_{k}\rangle_{k})$, $ v[x] \simeq_{\varepsilon} \TT{Ind}_{C}^{P}(\langle  u_{k}\rangle_{k})[x]$.

%
%
%
%
%
%
%
 
In Example \ref{ex:indu} first observe that the type $\nat$ is of the form $\forall X.\nat[X]$, where $\nat[X]=(X\To X)\To (X\To X)$ is a polynomial type. The two ways of defining our function correspond to two $\nat[X]$-algebras
$(\emptyset, A\To B, \langle  x\mapsto \lambda y.h(xy) , u \rangle)$ and 
$(\{ y:A  \}   , B, \langle x\mapsto hx, uy\rangle)$. These are related by the $\nat[X]$-algebra morphism $v[x]= xy$; moreover,  the two contexts $\rr{\TT C}_{1}$ and 
$\rr{\TT C}_{2}$ correspond to the associated $\nat[X]$-induction morphisms (respectively on $A\To B$ and $B$).
We can then deduce the equivalence $\rr{\TT C}_{1}\simeq_{\varepsilon}\rr{\TT C_{2}}$ from the fact that $\nat$ is the initial $\nat[X]$-algebra.

The moral of this argument is that one can simplify the type of a $P[X]$-induction whenever one has a 
$P[X]$-algebra morphism which is an \emph{elimination context}.

 \begin{definition}[elimination and introduction contexts]
\emph{Elimination contexts} and \emph{introduction  contexts} for $\Nd$ are defined by the grammars
$$
\rr{\TT E}:= [\ ]\mid \rr{\TT E}C \mid \rr{\TT E}u
\qquad
\rr{\TT I}:= [\ ]\mid \Lambda Y. \rr{\TT I} \mid \lambda y. \rr{\TT I}
$$
A \emph{pair of dual contexts} (noted $(\Elim, \Intro)_{A\vdash B}$) is a pair made by an elimination context $\rr{\TT E}:A\vdash B$ and an introduction context $\rr{\TT I}:B\vdash A$ such that $ \rr{\TT I}\circ \rr{\TT E}\simeq_{\eta}[\ ]$. 

%
%
%
%
%
%
%
%
 \end{definition}

For any type $A$ and right simplification $B\in \TT{RS}(A)$, we can always construct a pair of dual contexts $(\Elim, \Intro)_{A\vdash B}$ by letting $\Elim$ be $[\ ]$ followed by type and term variable applications and $\Intro$ be made of abstractions on the same variables. 
For example, if $A= \forall X. Y \To \forall Z. (X\To Y)\To Z$ and $B= (X\To Y)\To Z$, then we can let $\Elim= [\ ]X y Z$ and $\Intro= \Lambda X.\lambda y.\Lambda Z.[\ ]$.

%

The lemma below generalizes to an arbitrary polynomial type $P[X]$ the reasoning of Example \ref{ex:indu}, allowing to turn a $P[X]$ induction on $A$ into a $P[X]$-induction on $B$, for $B\in \TT{RS}(A)$:%
%
%
%

\begin{lemma}[simplification lemma]\label{lemma:simpli}
Let $P[X]$ be a polynomial type, $A$ be a $\Nd$-type, $B\in \TT{RS}(A)$, and
$(\Elim, \Intro)_{A\vdash B}$ be a pair of dual contexts.
If  $\Elim$  is a morphism between two $\Fun{X}{P}$-algebras $(\Gamma, A,\langle u_{k}^{A}\rangle_{k})$ and $(\Delta, B, \langle u_{k}^{B}\rangle)$, then
 $$
\TT{Ind}_{A}^{P}(\langle u_{k}^{A}\rangle_{k}) \ \simeq_{\varepsilon} \
\Intro\big [ \TT{Ind}_{B}^{P}[\langle u_{k}^{B}\rangle_{k}] \big ]
 \ : \ 
 \forall X.P[X] \vdash^{\Delta} A
$$
%
\end{lemma}
\begin{proof}
Since $\Elim\big [[\ ](A)\langle \lambda \langle x_{j}\rangle_{j}.u_{k}^{A}[\langle x_{j}\rangle_{j}]\rangle_{k}\big ]$ and 
$ [\ ](B)\langle \lambda \langle x_{j}\rangle_{j}.u_{k}^{B}[\langle x_{j}\rangle_{j}]\rangle_{k}$ are both $P[X]$-algebra morphisms from the initial $P[X]$-algebra to  $(B, \langle u_{k}^{B}\rangle)$, initiality implies that they are $\varepsilon$-equivalent. The claim then follows by applying $\eta$-equivalences.
\end{proof}
%
%


When $P[X]$ is a constant polynomial type, from Lemma \ref{lemma:simpli} we can deduce a rather striking fact: \emph{any} type instantiation of the initial algebra $\forall X.P[X]$ on a type $A$ can be \emph{atomized}, i.e. transformed into a term using a type instantiation on the variable $\at{A}$.

Let us fix, for all type $A$, a pair of dual contexts $(\Elim_{A}, \Intro_{A})_{A\vdash \at{A}}$.

\begin{lemma}[atomization lemma]
For all constant polynomial type $P[X]$ and type $B$ of $\Nd$,
\begin{equation}\label{eq:atomi}
[\ ] (B) \ \simeq_{\varepsilon} \ \lambda \langle y_{k}\rangle.
\Intro_{B}\Big [ [\ ] (\at{B})\langle \lambda \langle x_{j}\rangle_{j}. \Elim_{B}[y_{k} \langle x_{j}\rangle_{j} ]\rangle_{k}\Big ] \ :  \ \forall X.  P[X]\vdash P[B/X]
\end{equation}
\end{lemma}
\begin{proof}
Since $P[X]$ is constant, for any choice of $k,j$, $\FFun{\cc X\mapsto B}{A_{jk}}=A_{jk}$ and $\Phi^{X}_{A_{jk}}( t)=x$. This implies that, given the context $\Gamma=\{\big\langle y_{k}: \langle {A_{jk}} \rangle_{j}\To B\big \rangle_{k}\}$, 
 $\Elim_{B}$ is trivially a $P[X]$-algebra morphism from $(\Gamma, B, \langle y_{k}\langle x_{j}\rangle_{j}\rangle_{k})$ to $(\Gamma,\at{B}, \langle \Elim_{B}[  y_{k}\langle x_{j}\rangle_{j}\rangle_{k}])$. We then have
 $$
 [\ ]B \simeq_{\eta}
  \lambda \langle y_{k}\rangle_{k} . [\ ] (B) (\lambda \langle x_{j}\rangle_{j}.\langle y_{k}\langle x_{j}\rangle_{j}\rangle_{k}) =  \lambda \langle y_{k}\rangle_{k} . \TT{Ind}^{P}_{B}( \langle y_{k}\langle x_{j}\rangle_{j}\rangle_{k}) 
 $$
 $$
   \quad \stackrel{[\text{Lemma \ref{lemma:simpli}}]}{\simeq_{\varepsilon}}
   \quad \lambda \langle y_{k}\rangle_{k} . \Intro_{B} \Big[ \TT{Ind}^{P}_{\at{B}}( \langle \Elim_{B}\big[ y_{k}\langle x_{j}\rangle_{j}\rangle_{k}\big])  \Big]
 $$
where the last term is the desired one. 
\end{proof}

While the lemma above is rather technical, it has a rather intuitive explanation: it permits to transform any instantiations of a type of the form $\forall X.P[X]$, where $P[X]$ is a constant polynomial type, into an atomic instantiation, preserving the contextual behavior of the program. 
This property can be restated as the existence of an embedding 
between the fragment $\NRP$ (see \cite{SL2}) of $\Nd$ (in fact, of $\NY$) in which all universal types are of the form $\forall X.P[X]$, for some polynomial type $P[X]$, and the predicative fragment $\Nda$ of $\Nd$ (see \cite{Ferreira2013}), in which all type instantiations are atomic.
%
%

\begin{theorem}[atomizing embedding]\label{genff}
If $\Gamma\vdash_{\NRP}  t:A$, then there exists a term ${t^{\downarrow}}$ such that $\Gamma\vdash_{\Nda} {t^{\downarrow}}: A$ and $ t \simeq_{\varepsilon}{t^{\downarrow}}$. 
%

\end{theorem}
\begin{proof}
$t^{\downarrow}$ is obtained by replacing in $t$ any subterm of the form $t B$, with $B$ non-atomic, as in Eq.~\eqref{eq:atomi}.

\end{proof}

\begin{remark}[Comparing predicative and impredicative embeddings of $\NI$]\label{rem:IO}
The result above is discussed in \cite{SL2} in connection with some results from \cite{Ferreira2013}, which describes a variant $^{\bb{FF}}$ of the standard embedding $^{\sharp}$ of $\NI$ into $\Nd$ whose target is the predicative systems $\Nda$. In fact, by composing
 $^{\sharp}$ with atomization $t\mapsto (t^{\sharp})^{\downarrow}$ one obtains an embedding of $\NI$ into $\Nda$ which is $\beta$-equivalent to $^{\bb{FF}}$. From this fact we deduce that the standard embedding and the predicative variant from \cite{Ferreira2013} are indeed $\varepsilon$-equivalent.


\end{remark}

\section{Conclusion}

\subparagraph*{Related Work}

The  connection between parametricity, dinaturality and the Yoneda isomorphism is well-known \cite{Bainbridge1990, Plotkin1993, Hasegawa2009}. The extension of this correspondence to initial algebras comes from \cite{Uustalu2011}. \cite{Bernardy2010} exploits this connection to  
define a schema to \emph{test} the equivalence of two programs $t,u$ of type $\forall X. (\FFun{\cc X}{F}\To \mm X)\To (\FFun{\cc X}{G}\To X')\To \FFun{\cc X}{H}$ by first instantiating $X$ as $\alpha=\mu X.\FFun{\cc X}{F}$ and then applying $t,u$ to the canonical morphism $\FFun{\cc X\mapsto \alpha}{F}\To \alpha$
(in fact, this is exactly how one side of the isomorphisms $\equiv_{\cc X}$ are constructed). The possibility of expressing  program equivalence  through naturality conditions has recently attracted new attention 
 due to \cite{Awodey2018}, where these are investigated using ideas from homotopy type theory.
 Type isomorphisms in $\Nd$ with the Yoneda lemma are also discussed in \cite{Hinze2010}. In \cite{PistoneTLLA} a similar restriction based on the Yoneda isomorphism is used by the first author to describe a fragment of \emph{second order multiplicative linear logic} with a decidable program equivalence.

 
\subparagraph*{Future Work}

%

The definition of the characteristic employs an acyclicity condition which is reminiscent of linear logic \emph{proof-nets}. 
In particular, 
we would like to investigate whether the alternating paths can be related to the \emph{cyclic proofs} for linear logic systems with $\mu,\nu$-types \cite{Baelde2016}. 
Moreover, the notion of characteristic seems likely to scale to second order \emph{multiplicative-exponential} linear logic, an extension which might lead to better expose the intrinsic duality in the tree-shapes in Fig. \ref{fig:rewriting}.

The  connection between Yoneda type isomorphisms and proof-search techniques suggests to look for canonical proof-search algorithms for $\NY$ and $\NYY$ (for instance using \emph{focusing} as suggested in \cite{SchererPhD}).
 Moreover, the appeal to least/greatest fixpoints suggests a connection with 
 the technique to count inhabitants by computing \emph{fixpoints of polynomial equations} \cite{Zaoinc}. For example, given $A=(\cc Y_{1}\To \mm X)\To (  \cc X\To \cc Y_{2}   \To \mm X)\To \cc X$, one can show 
 by proof-theoretic reasoning that the number $|A|$ of inhabitants of $A$ is a solution of the fixpoint equation
$|A|= |A_{Y_{1}}|+ (|A|\times |A_{Y_{2}}|)$,  where $A_{Y_{i}}=(\cc Y_{1}\To \mm X)\To (  \cc X\To \cc Y_{2}   \To \mm X)\To \cc Y_{i}$, which implies $|A|=0$, since $|A_{Y_{i}}|=0$. On the other hand, Yoneda type isomorphisms yield the 
strikingly similar computation  
$\forall \vec YX.A \equiv_{\cc X}
\forall \vec Y.
\mu X. \cc Y_{1} + (\cc X\times \cc Y_{2})\equiv_{\vec{\cc Y}} \mu X.0+(\cc X\times 0)\equiv 0$.

%
%
 
 Finally, we would like to explore an apparent correspondence between 
Yoneda type isomorphisms of the form 
$$\forall X.\Big \langle \big \langle \FFun{\cc X}{A_{jk}}\big \rangle_{j}\To X\Big \rangle_{k}\To X  \  \equiv_{\cc X} \  \mu X.\sum_{i}\prod_{j}\FFun{\cc X}{A_{jk}}$$
\noindent and the so-called \emph{inversion principles} discussed in the proof-theoretic literature (see \cite{MT08}), which relate introduction and elimination rules for \emph{generalized connectives} \cite{Pra79,SH84,Backhouse1986} (for example, the isomorphism
$$
\forall X. (X_{1}\To X_{2}\To X)\To (X_{1}\To X_{3}\To X)\To X \ \equiv_{\cc X} 
\
(X_1\times X_2) + (X_1\times X_3)
$$
encodes the symmetry between rules for a ternary connective $\dag$ as shown in Fig. \ref{fig:rules}).

\begin{figure}
 \adjustbox{scale=0.85, center}{
 $\AXC{$X_1$}\AXC{$X_2$}\RL{$\dagger$I$_1$}\BIC{$\dagger(X_1,X_2,X_3)$}\DP \quad\AXC{$X_1$}\AXC{$X_3$}\RL{$\dagger$I$_2$}\BIC{$\dagger(X_1,X_2,X_3)$}\DP \qquad\qquad \AXC{$\dagger(X_1,X_2,X_3)$}\AXC{$[X_1,X_2]$}\noLine\UIC{$X$}\AXC{$[X_1,X_3]$}\noLine\UIC{$X$}\RL{$\dagger$E}\TIC{$X$}\DP$}
\caption{Introduction and elimination rules for the connective $\dag$.}
\label{fig:rules}
\end{figure}

\bibliography{YonedaTryAgain}

\appendix



%
%
%
%
%

%
%
\section{The $\varepsilon$-Theory and the Yoneda Isomorphisms}\label{app4}

We prove that the isomorphisms $\equiv_{\cc X}$ and $\equiv_{\mm X}$ hold under the $\varepsilon$-theory. Let us start from $\equiv_{\cc X}$, recalled below:
\begin{equation}\label{eqi}
\forall X.\Big \langle \forall \vec Y_{k}. \big\langle \FFun{\cc X}{A_{jk}}\big\rangle_{j}\To \mm X\Big\rangle_{k}\To \FFun{\cc X}{B} 
\equiv \FFun{\cc X \mapsto \{\mu  X.\}  \sum_{k}\left (\exists \vec Y_{k}.\prod_{j}\FFun{\cc X}{A_{jk}}\right)}{B}
\end{equation}

In the case of Eq. \eqref{eqi} we can construct terms  
\begin{equation*}
\begin{split}
\TT a_{k}[  \langle z_{j}\rangle_{j}]  =  \TT{in}_{k}^{\sharp K}(\pack[\vec Y_{k}](\TT{prod}_{j}^{\sharp J_{k}}\langle z_{j}\rangle_{j}))
& \ :  \ \big \langle \FFun{\mm X\mapsto \alpha}{A_{jk}}\big\rangle_{j} \vdash \FFun{\cc X\mapsto\alpha}{T}
\\
\hat{\TT a}_{k}  = 
\Lambda \vec Y_{k}. \lambda \langle z_{j}\rangle_{j}. \TT{in}_{T}(\TT a_{k}[\langle z_{j}\rangle_{j}])
& \  : \ 
\forall \vec Y_{k}. \langle A_{jk}\langle \mm X\mapsto \alpha\rangle\rangle_{j}\To \alpha
\end{split}
\end{equation*}
\noindent where $\FFun{\cc X}{T}=\sum_{k}\exists \vec Y_{k}.\prod_{j}\FFun{\cc X}{A_{jk}}$ and 
$\alpha=\mu X. \FFun{\cc X}{T}$, 
$\TT{in}_{k}^{\sharp I}:  X_{k}\To \sum_{k}^{\sharp I}X_{k}$
and $\TT{prod}_{j}: \langle X_{j}\rangle_{j}\To \prod_{j}^{\sharp J_{k}}X_{j}$ are defined composing usual sum and product constructors, and 
 $\pack$ is defined as follows:
$$
  \pack[B_{1},\dots, B_{k}]  =
  \lambda x.\Lambda  Z .\lambda f. xB_{1}\dots B_{k} f  : A[B_{1}/Y_{1},\dots, B_{k}/Y_{k}] \To \exists \vec Y. A
  $$
  
  With such terms we can then construct a term 
$$
\TT s_{ {A_{jk}}, {B}}[x]=x \langle \hat{\TT a}_{k}\rangle_{k}:
 \Big \langle\forall \vec Y_{k}. \big\langle \FFun{\mm X\mapsto \alpha}{A_{jk}}\big\rangle_{j}\To \alpha\Big\rangle_{k}\To \FFun{\mm X\mapsto \alpha}{B} 
 \quad \vdash \quad 
\FFun{\cc X \mapsto 
\alpha}{B}
$$

Moreover, using sum and product destructors we can construct terms 
\begin{equation*}
\begin{split}
\TT b[x,Z] 
& \ : \ 
\FFun{\mm X \mapsto Z}{T} 
\vdash^{\Delta}
Z
 \\
   \TT t_{ {A_{jk}}, {B}}[x,Z]
   & \ : \ 
 \FFun{\mm X\mapsto \alpha}{B}\quad  \vdash \quad
 \Big \langle\forall \vec Y_{k}. \big\langle \FFun{\cc X\mapsto Z}{A_{jk}}\big\rangle_{j}\To Z\Big\rangle_{k}\To \FFun{\cc X\mapsto Z}{B} 
\end{split}
\end{equation*}
where $\Delta= \{
  \langle f_{k}: \forall \vec Y_{k}. \langle \FFun{\mm X\mapsto Z}{A_{jk}}\rangle_{j}\To Z\rangle_{k}\}$ and \begin{equation*}
\begin{split}
\TT b[x,Z] &  =   \delta^{\sharp K}\Big ( x, \Big \langle    z. \unpack(z)
\big (\Lambda \vec Y_{k}.\lambda y. f_{k}\vec Y_{k} \langle \pi^{\sharp J_{k}}_{j}(y)\rangle_{j}\big ) \Big\rangle_{k} \Big)
%
  \\ 
\TT t_{A_{jk},B}[x,Z] & =  \lambda \langle f_{k}\rangle_{k}.\Fun{X}{B}\big ( \ff_{T}(\lambda x.\TT b[x,Z] ) x\big )
\end{split} 
\end{equation*}
\noindent
 where $\pi^{j}_{i}$ and $\delta^{\sharp K}(t, \langle z. u_{k}\rangle_{k})$ indicate suitable generalized product and sum  destructors which can be defined inductively using product and sum destructors, and the term $\unpack$ is defined as follows:
 $$
  \unpack  =
\Lambda Z.\lambda x.\lambda f. fZx
  : \forall Z. \left(\exists \vec Y. A\right) \To( \forall \vec Y.
  A \To  Z)\To Z
$$

%

One can check that $\TT b[x,{\alpha}]\big[\langle f_{k}\mapsto \hat{\TT a}_{k}\rangle_{k}\big]: \FFun{\cc X\mapsto \alpha}{T}\vdash \alpha$ is $\beta\eta$-equivalent to $\inn_{T}x$, from which we deduce
$$\ff_{T}\Big(\lambda x.\TT b[x,\alpha]\big[ \langle f_{k}\mapsto  \hat{\TT a}_{k} \rangle_{k}\big] \Big)x\simeq_{\beta\eta}
\ff_{T}(\lambda x.\inn_{T}x)x
\simeq_{\eta}x $$

In this way the isomorphism $\equiv_{\cc X}$ are realized in $\CTX_{\varepsilon}(\Nnew)$ by the two terms below:\footnote{We are here supposing that $X$ does occur in at least some of the $A_{jk}$ (so that $\mu X.$ actually occurs in the left-hand type of $\equiv_{\cc X}$). If this is not the case, the construction can be done in a similar (and simpler) way.}
$$
\TT s[x]= \TT s_{A_{jk},B}[x\alpha] 
\qquad\qquad
\TT t[x]=
\Lambda X.\TT t_{A_{jk},B}[x,X] $$
We can compute then 
$$
\TT s[ \TT t[x]] \simeq_{\beta}
\Fun{X}{B}\Big ( \ff_{T}\big (\lambda x.\TT b[x,{\alpha}]\big[ \langle f_{k}\mapsto \langle \hat{\TT a}_{k}\rangle_{k}\big] \big )x\Big)\simeq_{\beta\eta}\Fun{X}{B}(x)\simeq_{\eta}x
$$
and 
$$
\TT t[\TT s[x]] \simeq_{\beta}
\Lambda X.\lambda \langle f_{k}\rangle_{k}.
\Fun{X}{B}\Big (  \ff_{T}(\lambda x.\TT b[x,X] ) x\Big) 
\Big [x\mapsto   x\alpha \langle\hat{\TT a}_{k}\rangle_{k} \Big]
$$
$$
\qquad\qquad\qquad
\simeq_{\varepsilon}
\Lambda X.\lambda \langle f_{k}\rangle_{k}.
x X  \langle f_{k}\rangle_{k}
\simeq_{\eta}x
$$
where the central $\varepsilon$-equivalence is justified using the $\varepsilon$-rule in Fig. \ref{fig:dinazza} 
with $E=\alpha$, $F=X$, $e_{k}=\hat{\TT a}_{k}$ and $\TT v[x]=\ff_{T}(\lambda x.\TT b[x,{X}])x$, with the last premise given by the computation below:
$$
\Big(\ff_{T}(\lambda x.\TT b[x,X] )x\Big)\Big[ x\mapsto \hat{\TT a}_{k}\vec Y_{k}\langle z_{j}\rangle_{j} \Big] \simeq_{\beta}
\Big(\ff_{T}(\lambda x.\TT b[x,X] )\Big)\inn_{T}(\TT a_{k}[\langle z_{j}\rangle_{j}])  
$$
$$
\simeq_{\beta}
\big (\lambda x.\TT b[x,{X}]\big )
\Big (\big ( \Fun{X}{T}( \ff_{T}(\lambda x.\TT b[x,{X}])x )
   \big)
\Big [x\mapsto \TT a_{k}[\langle z_{j}\rangle_{j}] \Big]
\Big)
$$
$$\simeq_{\beta}
f_{k}\vec Y_{k}
\Big \langle
\Fun{X}{A_{jk}}( \ff_{T}(\lambda x.\TT b[x,{X}])x )
[x\mapsto z_{j}
]
\Big\rangle_{j}
$$ 
\noindent where the last $\beta$-equation can be checked by unrolling the definition of $\Phi^{X}_{T}$:
$$\Phi^{X}_{T}(t)= \delta^{\sharp K}\Big ( x, \Big \langle
 z. \mathsf{unpack}(z) \big (
\Lambda \vec Y_{k}. \lambda x. \TT{in}^{\sharp K}_{k} ( \mathsf{pack}[\vec Y_{k}] (  
\TT{prod}^{\sharp J_{k}} \langle   \Phi^{X}_{A_{jk}}(t[ x\mapsto \pi^{\sharp J_{k}}_{j}(x)] ) \rangle_{j}
 ))
\big)
\Big \rangle_{k}
\Big)
$$

Let us now consider the isomorphism $\equiv_{\mm X}$, recalled below:
\begin{equation}\label{eqii}
 \forall X. \Big \langle \forall \vec Y_{j}.\cc X\To \FFun{\mm X}{A_{j}}\Big\rangle_{j}\To \FFun{\mm X}{B} 
 \equiv \FFun{\mm X \mapsto 
\{\nu  X.\} \forall \vec Y_{j}.\prod_{j}\FFun{\mm X}{A_{j}}}{B}
\end{equation}
We can  
define terms\footnote{We are here supposing that $X$ does occur in at least some of the $A_{j}$ (so that $\nu X.$ actually occurs in the left-hand type of $\equiv_{\mm X}$). If this is not the case, the construction can be done in a similar (and simpler) way.}
\begin{equation*}
\begin{split}
\TT c_{j}[x,\vec Y_{j}]=  \pi^{\sharp J}_{j}(x\vec Y_{j})
& \ : \  \FFun{\mm X\mapsto \beta}{U}\vdash \FFun{\cc X\mapsto \beta}{A_{j}}
\\
\hat{\TT c}_{j}= \Lambda \vec Y_{j}.\lambda x.\TT c_{j}[\outt_{U}(x),\vec Y_{j}]
& \ 
: \
\forall \vec Y_{j}.\beta\To \FFun{\cc X\mapsto \beta}{A_{j}}
\end{split}
\end{equation*}
\noindent where
$\FFun{\cc X}{U}= \forall \vec Y_{j}.\prod_{j}\FFun{\cc X}{A_{j}}$ and $\beta= \nu X.\FFun{\cc X}{U}$, and a term
$$
\TT u_{{A_{j}, B}}[x]=x\langle \hat{\TT c}_{j}\rangle_{j}:
 \Big \langle \forall \vec Y_{j}. \beta\To  \FFun{\cc X\mapsto \beta}{A_{k}}\Big\rangle_{j}\To \FFun{\cc X\mapsto \beta}{B} 
 \quad \vdash \quad 
\FFun{\mm X \mapsto 
\beta}{B}
$$

Moreover, using product constructors we can construct a term 
$$\TT d[\langle f_{j}\rangle_{j}, Z]:  \Big \langle f_{j}: \forall \vec Y_{j}.Z\To \FFun{\mm X\mapsto Z}{A_{j}}\Big\rangle_{j} \ \vdash  \ 
\big (Z\To  \FFun{\cc X \mapsto Z}{U}\big ) $$ 
with $\TT d[\langle f_{j}\rangle_{j},Z]=  \lambda z. \Lambda \vec Y_{j}. \mathsf{prod}^{\sharp J}( \langle f_{j}\vec Y_{j} z\rangle_{j})$
and 
a term
$$
\TT v_{ {A_{j}}, {B}}[x,Z]:
\FFun{\cc X\mapsto \beta}{B}\quad \vdash \quad
 \Big \langle\forall \vec Y_{j}. Z\To  \FFun{\mm X\mapsto Z}{A_{j}}\Big\rangle_{j}\To \FFun{\mm X\mapsto Z}{B} 
$$
where $\TT v_{A_{j},B}[x,Z]=  \lambda \langle f_{j}\rangle_{j}.\Fun{X}{B}\big( \uu_{U}(\TT d[\langle f_{j}\rangle_{j}, Z] )x\big )$.

In this way the isomorphism $\equiv_{\mm X}$ are realized in $\CTX_{\varepsilon}(\Nnew)$ by the two terms below:
$$
\TT u[x]= \TT u_{A_{j},B}[x\beta] 
\qquad\qquad
\TT v[x]=
\Lambda X. \TT v_{A_{j},B}[x,X]$$
We leave it the reader to check that $\TT u[\TT v[x]]\simeq_{\varepsilon}x$ and $\TT v[\TT u[x]]\simeq_{\varepsilon}x$, using an argument similar to the one for $\TT s$ and $\TT t$ and using the $\varepsilon$-rule in Fig. \ref{fig:dinazzabis}.

\section{$\NYY$ and $\mu\BB B$ are full subcategories of $\Nnew$}\label{app8}

Given a term $t$ with $n$ free variables $x_{1},\dots, x_{n}$ and $n$ terms $\vec{\rr{ u}_{i}}[x]$, we let $t\mcirc \vec{\rr{ u}}$ be shorthand for $t[ \rr{ u}_{1}[x_{1}]/x_{1},\dots, \rr{ u}_{n}[x_{n}]/x_{n}]$.

In this section we prove the following commutations:

\begin{proposition}\label{prop:mado}
\begin{enumerate}
\item If $\Gamma\vdash_{\Nnew}  t:  A$, then 
$\Gamma\vdash_{\Nnew}{ \rr{\TT c}_{A} [ t]}   \simeq_{\varepsilon} { \Rp t \mcirc \vec{\rr{\TT c}}_{\Gamma}}: \Rp A$.


\item If $\Gamma\vdash_{\Nnew}  t:  A$, then 
$\Gamma\vdash_{\Nnew}{ \rr{\TT d}_{A} [ t]}   \simeq_{\varepsilon} {  t^{\flat} \mcirc \vec{\rr{\TT d}}_{\Gamma}}: A^{\flat}$.

\end{enumerate}
\end{proposition}

In fact, from Proposition \ref{prop:mado} we immediately deduce our fullness claims:
\begin{theorem}
\label{thm:full}
\begin{enumerate}
\item If $\Gamma\vdash_{\Nnew }   t:  A$ and $\Gamma,A$ are $\NYY$-types, then there exists $u$ such that 
$\Gamma\vdash_{\NYY}u:A$ and 
$\Gamma\vdash_{\Nnew}{ t}   \simeq_{\varepsilon} { u}: A$.

\item If $\Gamma\vdash_{\Nnew}   t:  A$ and $\Gamma,A$ are $\NImu$-types, then there exists $u$ such that 
$\Gamma\vdash_{\NImu}u:A$ and 
$\Gamma\vdash_{\Nnew}{ t}   \simeq_{\varepsilon} { u}: A$.

\end{enumerate}
\end{theorem}
\begin{proof}
In the first case let $u=t^{\rop}$, and in the second case let $u=t^{\flat}$.
\end{proof}

We limit ourselves to check the first statement of Prop. \ref{prop:mado}, as the other one is established by computing $\varepsilon$-equivalences in a similar way. We need the following two technical lemmas.

\begin{lemma}\label{lemma:reduce}
For all $t:\forall X.A$  as in \eqref{eqi} and type $C$, 
{\small\begin{equation}\label{eq:bordello}
{\Lambda \vec Y.\lambda \vec z. \Fun{X}{P}(\rr{\TT c}_{C}) \Big[  xC \vec Y \  (\Fun{X}{A_{1}}(\rr{\TT c}^{-1}_{C})[ z_{1}] ) \dots(\Fun{X}{A_{k}}(\rr{\TT c}^{-1}_{C})[ z_{k}] )  
\Big]}
 \simeq_{\varepsilon}
{ x\Rp C
}
: A[\Rp C/X]
\end{equation}}
\end{lemma}
\begin{proof}
From 
$
{
xC \vec Y \  (\Fun{X}{A_{1}}(\rr{\TT c}^{-1}_{C})[ z_{1}] ) \dots(\Fun{X}{A_{k}}(\rr{\TT c}^{-1}_{C})[ z_{k}]   )}
\simeq_{\varepsilon}^{\rr{\TT c}^{-1}_{C}}
{
\Fun{X}{P}(\rr{\TT c}^{-1}_{C}) 
\big [ x \Rp C\vec Y   z_{1}\dots z_{k} \big]
}
$
we deduce 
$
t
\simeq_{\varepsilon}
{\Lambda \vec Y.\lambda \vec z. \Fun{X}{P}(\rr{\TT c}_{C})\circ \Fun{X}{P}(\rr{\TT c}^{-1}_{C}) 
\big [
x\Rp C\vec Y \  z_{1}\dots z_{k} 
\big]
}
\simeq_{\varepsilon}
{
\Lambda \vec Y.\lambda \vec z. x \Rp C \vec Yz_{1}\dots z_{k}}
\simeq_{\eta}
{
x \Rp C
}
$,
where $t$ is the lefthand term in Equation \ref{eq:bordello}.
\end{proof}

\begin{lemma}\label{coroll:bicci}
Let $A$ be as in Lemma \ref{lemma:reduce}. Then for all type $C$,
${\rr{\TT c}_{\forall X.A}\Rp C} \simeq_{\varepsilon} {\rr{\TT c}_{A[C/X]}}$.
\end{lemma}
\begin{proof}
By a simple calculation we can deduce 
${\rr{\TT c}_{\forall X.A}\Rp C}\simeq_{\beta}
{ \rr{\TT c}_{A} [ x\Rp C]}$ and  \\
${\rr{\TT c}_{A[C/X]}}\simeq_{\beta}
{\rr{\TT c}_{A}  \Big[
\Lambda \vec Y.\lambda \vec z. \Fun{X}{F}(\rr{\TT c}_{C}) \big [
xC  \vec Y \  (\Fun{X}{A_{f(i_{1})}}(\rr{\TT c}^{-1}_{C})[ z_{1}] ) \dots(\Fun{X}{A_{f(i_{k})}}(\rr{\TT c}^{-1}_{C})[ z_{k}]   )
\big]\Big]}
$,
so we can conclude by Lemma \ref{lemma:reduce}.
\end{proof}

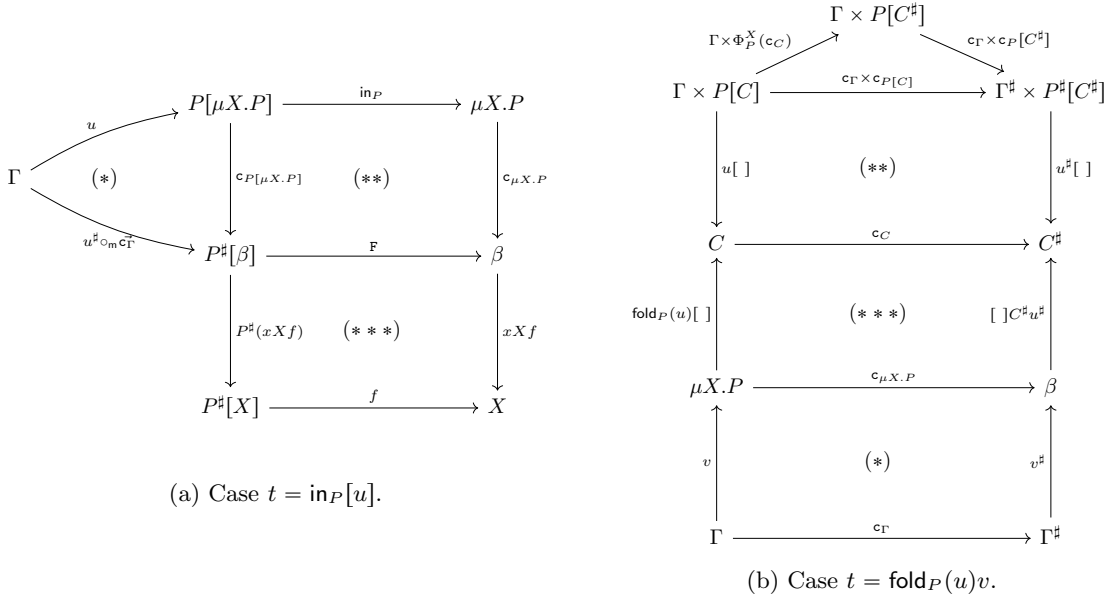
\begin{figure}[t]
\begin{subfigure}{0.48\textwidth}
\adjustbox{scale=0.8,center}{
\begin{tikzcd}
    &   &  P[\mu X.P] \ar{dd}[black]{\rr{\TT c}_{P[\mu X.P]}}  
     \ar{rr}[black]{\inn_{P}} &   & \mu  X.P    \ar{dd}[black,right]{  \rr{\TT c}_{\mu X.P}  }     \\
  \Gamma \ar[bend left=10]{urr}[black]{u}        \ar[ bend right=10]{rrd}[black,below]{u^{\rop}\mcirc \vec{\rr{\TT c}_{\Gamma}}}  & 
  (*) &  & (**) &  \\
      &    &  P^{\rop}[\beta] \ar{dd}{P^{\rop}(xXf)}  \ar{rr}[black]{\rr{\TT F}}   & &   \beta  \ar{dd}{xXf}   \\
      & & &(***)  & \\
  &  &P^{\rop}[X]  \ar{rr}{f}  &  &  X    \\
\end{tikzcd}
}
\caption{Case $t= \inn_{P}[u]$.}
\label{fig:sbatti}
\end{subfigure}
\ \ 
\begin{subfigure}{0.48\textwidth}
\adjustbox{scale=0.8,center}{
\begin{tikzcd}
       &   \Gamma \times P[C^{\rop}]  \ar{rd}{\rr{\TT c}_{\Gamma}\times \rr{\TT c}_{P}[ C^{\rop}  ]} &   \\
           \Gamma\times P[C] \ar{rr}{\rr{\TT c}_{\Gamma}\times \rr{\TT c}_{P[C]}} \ar{ru}{\Gamma\times \Fun{X}{P}(\rr{\TT c}_{C})}\ar{dd}{u[\ ]}  &   & \Gamma^{\rop}\times P^{\rop}[C^{\rop}]  \ar{dd}{u^{\rop}[\ ]} \\
             & (**)  &   \\
  C \ar{rr}{\rr{\TT c}_{C}}  &   &  C^{\rop} \\
& (***) &  \\
    \mu X.P \ar{uu}{\ff_{P}(u)[\ ]}  \ar{rr}{\rr{\TT c}_{\mu X.P}}& &  \beta \ar{uu}{[\ ]C^{\rop}u^{\rop}}\\
    & (*) & \\
    \Gamma  \ar{uu}{v} \ar{rr}{\rr{\TT c}_{\Gamma}}& & \Gamma^{\rop}\ar{uu}{v^{\rop}}
   \end{tikzcd}
}
\caption{Case $t=\ff_{P}(u)v$.}
\label{fig:sbatti2}
\end{subfigure}
\caption{Commutation cases.}
\end{figure}

\begin{proof}[Proof of Proposition \ref{prop:mado}]
By induction on $  t$, we only consider some interesting cases:
\begin{itemize}

\item if $  t=  {uC}$, then $A= B[C/Y]$, $ {\Rp t}=  {\Rp u\Rp C}$ and by the induction hypothesis
$
{\rr{\TT c}_{\forall Y.B}[ u]}\simeq_{\varepsilon}
{\Rp u \mcirc \rr{\TT c}_{\Gamma}}
$. By Lemma \ref{coroll:bicci} and the induction hypothesis we have
${\rr{\TT c}_{B[C/Y]}[ u]} \simeq_{\varepsilon}  { ( \rr{\TT c}_{\forall Y.B}[ u])\Rp C  }
\simeq_{\varepsilon}
{ (\Rp u \mcirc \rr{\TT c}_{\Gamma})\Rp C} =
{ \Rp t\mcirc \rr{\TT c}_{\Gamma}}$ (as $Y\notin FV(\Gamma)$). 

\item if $ t=  {\delta_{C}(u, x.v_{1}, x.v_{2})}$, then ${\Rp t}=  {\Rp u \Rp C \ \lambda x.\Rp{v_{1}} \ \lambda x.\Rp{v_{2}}}$ and by the induction hypothesis we have
$
{  \rr{\TT c}_{B_{1}+B_{2}}[ u] }\simeq_{\varepsilon}  { \Rp u \mcirc \rr{\TT c}_{\Gamma}}
$ and
${ \rr{\TT c}_{C} [ v_{i}]} \simeq_{\varepsilon}  {\Rp{v_{i}}\mcirc \rr{\TT c}_{\Gamma,  {x}: B_{i}}}$.
We can compute then
\begin{equation*}
\begin{split}
\Rp t \mcirc \rr{\TT c}_{\Gamma}
& =
 (\Rp u \mcirc\rr{\TT c}_{\Gamma})\Rp C \ \lambda x . (\Rp{v_{1}}\mcirc \rr{\TT c}_{\Gamma, x:B_{1}}) \ \lambda x. (\Rp{v_{2}}\mcirc \rr{\TT c}_{\Gamma, x:B_{2}})
\\ &
\simeq_{\varepsilon}
( \rr{\TT c}_{B_{1}+B_{2}}\circ u)\Rp C \ \lambda x_{1}. (\Rp{v_{1}}\mcirc \rr{\TT c}_{\Gamma}) \ \lambda x_{2}. (\Rp{v_{2}}\mcirc \rr{\TT c}_{\Gamma})
\\ &
\simeq_{\beta}
( \Lambda X.\lambda ab. \delta_{X}( u, x.a(\rr{\TT c}_{B_{1}}[y])), x.a(\rr{\TT c}_{B_{1}}[y]))
\Rp C \ \lambda x_{1}. (\Rp{v_{1}}\mcirc \rr{\TT c}_{\Gamma, x:B_{1}}) \ \lambda x_{2}. (\Rp{v_{2}}\mcirc \rr{\TT c}_{\Gamma, x:B_{2}})
\\ &\simeq_{\beta}
\delta_{\Rp C}( u, 
x. \Rp{v_{1}}[\rr{\TT c}_{B_{1}}[x]/x] \mcirc \rr{\TT c}_{\Gamma}, x. \Rp{v_{2}}[\rr{\TT c}_{B_{2}}[x]/x] \mcirc \rr{\TT c}_{\Gamma}) \\
& 
=
\delta_{\Rp C}( u, 
x. (\Rp{v_{1}}\mcirc \rr{\TT c}_{\Gamma,  {x: B_{1}}}), x. (\Rp{v_{2}}\mcirc \rr{\TT c}_{\Gamma,  {x: B_{2}}}))
\simeq_{\varepsilon}
\delta_{\Rp C}( u, 
x. \rr{\TT c}_{C}[ \Rp{v_{1}}], x. \rr{\TT c}_{C}[ \Rp{v_{2}}])
\\ &
\simeq_{\eta}
\rr{\TT c}_{C}[ \delta_{C}(u, x.v_{1},x.v_{2})]
=
\rr{\TT c}_{C} [ t]
\end{split}
\end{equation*}

\item if $  t=  {\xi_{C}u}$, then $ {\Rp{t}}= {\Rp u C}$ and
by the induction hypothesis
$
{
\rr{\TT c}_{0}[ u]} \simeq_{\varepsilon}  {\Rp u \mcirc \rr{\TT c}_{\Gamma}
}
$ so we have
$
{
\Rp t \mcirc \rr{\TT c}_{\Gamma}
}
=
{
(\Rp u \mcirc \rr{\TT c}_{\Gamma})C
}
\simeq_{\varepsilon}
{(\rr{\TT c}_{C} [ u])C}
\simeq_{\varepsilon}
{
 \rr{\TT c}_{C} [ uC]
}
$
where the last step is an application of Lemma \ref{coroll:bicci}.

\item if $t= \inn_{P}u$, then $t^{\rop}[x]= \Lambda X.\lambda f. f(\Fun{X}{(P^{\rop})}( xX f   )[u^{\rop}] )$ and  by the induction hypothesis
$\rr{\TT c}_{P[\mu X.P]}[ u] \simeq_{\varepsilon} u^{\rop}\mcirc \vec{\rr{\TT c}_{\Gamma}}$.
Moreover we have that $A=\mu X.P$ and 
$\rr{\TT c}_{\mu X.P}[ t]= \ff_{P}(  \rr{\TT f}\circ \rr{\TT c}_{P}[\beta/X]  )$
where $\beta= \forall X.(P^{\rop}[X]\To X)\To X$ and 
$\rr{\TT f}[x]=    \Lambda X.\lambda f. f( \Fun{X}{P^{\rop}}(xXf)) $. 
The equivalence $(t^{\rop}\circ \vec{\rr{\TT c}}_{\Gamma})Xf \simeq_{\varepsilon} \rr{\TT c}_{\mu X.P}[t]Xf$ 
 is then seen from the diagram in Fig. \ref{fig:sbatti}, where $(*)$ commutes by the induction hypothesis, $(**)$ commutes since $\rr{\TT c}_{P[\mu X.P/X]}\simeq_{\beta\eta} \rr{\TT c}_{P}[  \beta/X]\circ \Fun{X}{P}(  \rr{\TT c}_{\mu X.P})$, and $(***)$ commutes by an instance of a $\varepsilon$-rule.

%

\item if $t= \ff_{P}(u)v$, then 
$t^{\rop}= v^{\rop}C^{\rop}u^{\rop}   $ and by the induction hypothesis we know that 
$\rr{\TT c}_{\mu X.P}[v]\simeq_{\varepsilon}v^{\rop} \mcirc \vec{\rr{\TT c}}_{\Gamma}$ and 
$\rr{\TT c}_{P(C)\To C }[u]\simeq u^{\rop}\mcirc \vec{\rr{\TT c}}_{\Gamma}$. 
The equivalence between $\rr{\TT c}_{C}[ \ff_{P}(u)v]$ and $t^{\rop}\mcirc \vec{\rr{\TT c}}_{\Gamma} = (v^{\rop}C^{\rop}u^{\rop})\mcirc \vec{\rr{\TT c}}_{\Gamma}$ can be seen from the diagram in Fig. \ref{fig:sbatti2}, where 
$(*)$ and $(**)$ commute by the induction hypothesis and the commutation of $(***)$ follows from
$\rr{\TT c}_{C}[\ff_{P}(u)]\simeq_{\varepsilon} \ff_{P}( u^{\rop}\rr{\TT c}_{P}[C^{\rop}])$
 and 
$\ff_{P}(u^{\rop}\rr{\TT c}_{P})\simeq_{\varepsilon} C^{\rop}u^{\rop}\rr{\TT c}_{\mu X.P}$, where the latter holds since
$\rr{\TT c}_{\mu X.P}$ is an isomorphism of initial algebras. 
\end{itemize}
\end{proof}

%

\end{document}